\documentclass[english,12pt]{article}

\usepackage{amssymb}
\usepackage{amsfonts} 
\usepackage{amssymb}
\usepackage[toc,page]{appendix}
\usepackage{amsmath}
\usepackage[utf8]{inputenc}
\usepackage{amsthm}
\usepackage[margin=0.5in]{geometry}
\usepackage{float}
\usepackage{footmisc}
\usepackage{bookmark}
\usepackage{accents}

\usepackage{chngcntr}
\counterwithout{figure}{subsection}

\usepackage{enumerate}

\usepackage{hyperref}
\hypersetup{
    colorlinks,
    citecolor =blue,
    linkcolor=blue,
    filecolor=blue,      
    urlcolor=blue,
}

\usepackage{xpatch}
\xpatchcmd{\proof}{\itshape}{\bfseries\proofnameformat}{}{}

\newcommand{\proofnameformat}{}

\usepackage{color}
\usepackage{graphicx}

\newcommand{\equivalent}{{\Longleftrightarrow}}
\newcommand{\imply}{{\Longrightarrow}}

\newcommand{\ra}{{\rightarrow}}

\newcommand{\better}{{\succeq}}
\newcommand{\sbetter}{{\succ}}
\newcommand{\worse}{{\preceq}}

\newcommand{\R}{{\mathbb R}} 
\newcommand{\F}{{\mathbb F}} 
\newcommand{\N}{{\mathbb N}}

\newcommand{\E}{{\mathbb E}}

\newcommand{\D}{\mathcal{D}}
\newcommand{\mc}{\mathcal{C}}
\newcommand{\h}{\mathcal{H}}

\newcommand{\A}{\mathcal{A}}

\newcommand{\bi}{\begin{itemize}}
\newcommand{\ei}{\end{itemize}}

\newtheorem{theorem}{Theorem}
\newtheorem{lemma}{Lemma}
\newtheorem{proposition}{Proposition}
\newtheorem{definition}{Definition}

\newtheorem{remark}{Remark}
\date{\vspace{-5ex}}

\begin{document}

\title{Dynamic Random Subjective Expected Utility}

\author{Jetlir Duraj\footnote{duraj@g.harvard.edu, Acknowledgments: I am indebted to Drew Fudenberg and Tomasz Strzalecki for their continuous encouragement and support in this project. I thank Jerry Green, Kevin He, Eric Maskin and Nicola Rosaia for numerous comments during different stages of this project. I also thank Arjada Bardhi, Krishna Dasaratha, Ryota Iijima, Jonathan Libgober, Jay Lu, Maria Voronina and the audience of Games and Markets at Harvard for their helpful comments. Any errors are mine.}}

\maketitle


\begin{abstract}

Dynamic Random Subjective Expected Utility (DR-SEU) allows to model choice data observed from an agent or a population of agents whose beliefs about objective payoff-relevant states and tastes can both evolve stochastically. Our observable, the augmented Stochastic Choice Function (aSCF) allows, in contrast to previous work in decision theory, for a direct test of whether the agents' beliefs reflect the true data-generating process conditional on their private information as well as identification of the possibly incorrect beliefs. We give an axiomatic characterization of when an agent satisfies the model, both in a static (R-SEU) as well as in a dynamic setting (DR-SEU). We look at the case when the agent has correct beliefs about the evolution of objective states, as well as at the case when her beliefs are incorrect but unforeseen contingencies are impossible. 

We also distinguish in some detail two sub-variants of the dynamic model which coincide in the static setting: Evolving SEU, where a sophisticated agent's utility evolves according to a Bellman equation and Gradual Learning, where the agent is learning about her taste over time. We prove easy and natural comparative static results on the degree of belief incorrectness as well as on the speed of learning about taste.

Auxiliary results contained in the online appendix extend previous decision theory work in the menu choice and stochastic choice literature from a technical as well as a conceptual perspective.

\end{abstract}



\maketitle

\section{Introduction}\label{sec:intro}






The study of stochastic choice has found renewed popularity in economics. Along with a considerable amount of research on \emph{static} stochastic choice models, several recent works have pioneered foundational work into \emph{dynamic} stochastic choice models.\footnote{\cite{fs}, \cite{fis}, \cite{ssm} to name a few.} In a dynamic setting the agent solves a dynamic decision problem and learns as time passes about either the environment she is facing or her own evolution of preferences or both. In many applications an analyst only observes choices of an agent as well as some (possibly public) signals about payoff-relevant objective states. He doesn't have information about the stochastic process of the preferences of the agent (the private information of the agent). 

In this paper we consider such a general environment: there are payoff-relevant objective states every period, an agent has every period standard \emph{subjective expected utility} (SEU) preferences, comprised of beliefs about the objective, payoff-relevant state as well as a Bernoulli utility over a set of prizes. The \emph{subjective state} of an agent in each period consists of her realized SEU. We assume that these follow an exogenously given stochastic process which is well-known to the agent (albeit unknown to the analyst). We assume the agent can't influence the given stochastic process and allow for both stochastic tastes and stochastic beliefs. In many real life examples this two-fold randomness is present, e.g. investment and saving behavior may depend both on exogenous, objective randomness such as market conditions as well as on the stochastic evolution of the risk aversion of the agent.

We assume that after each history of choices and realizations of the objective states the analyst observes limiting frequencies of the choice of the agent in decision problems/menus of the current period as well as the realization of the objective states in the current period. The data also reflect variation of the decision problems/menus. Thus, the observable is in every period, after each history of choices and objective states, a \emph{probability distribution} over choices from a menu in the current period \emph{and} over realizations of the objective state.\footnote{Our identification results are valid under more general conditions -- see Remark \ref{thm:remk1} in section \ref{sec:rseu}.} Many situations in real life deliver such data, from employment situations in the labor market, consumption and investment decisions, to educational choices of students, loan practices, etc.\footnote{This type of data also allows an alternative \emph{heterogeneous population} interpretation: there is a population of agents facing similar choice situations. The analyst observes in many instances the choice of an agent \emph{as well as} the realization of some payoff-relevant objective state. We focus on the single-agent case in the exposition, but intuitions and results can be readily translated.} 

Our focus is axiomatic throughout. 
Under the assumption that the distribution of the private information of the agent doesn't depend on the decision problem she faces and that the analyst has access to a rich observable featuring variation in the decision problems, we give conditions on the observable which allow the analyst to uncover the distribution of the private information of the agent regardless of its arbitrariness. Under these conditions the analyst can also study whether the agent's beliefs when making choices reflect the \emph{correct} data-generating process \emph{conditional} on her private information and whenever that is not the case he can identify the biases conditional on the agent's private information. While the study of misspecified learning is not new, this is the first work, to the best of knowledge, where there are no \emph{a priori} assumptions on the origin of the misspecification. The misspecified beliefs may be because of misspecified priors, because of imprecise observation of private signals by the agent or because conditional on her private information the agent has some arbitrary behavioral biases in beliefs.\footnote{E.g. this model allows for the case of \emph{confirmatory bias} studied in \cite{rs} where an agent may misread signals in a way favorable to her current hypothesis. The agent in their model is not Bayesian with respect to the correct prior but is so \emph{within her model}.} 

The model we consider is still falsifiable as we require the agent to be Bayesian with respect to the stochastic process describing the evolution of her private information, even though she may be non-Bayesian with respect to the true data generating process of the objective states. Moreover, we don't allow any misspecified learner to receive \emph{hard} evidence about misspecification, such as the occurrence of an unforeseen contingency. Thus, in this paper the agent is able to explain any observed string of objective states \emph{within} her model, even though as time passes her beliefs might diverge more and more from the true data-generating process.\footnote{The time horizon is assumed to be finite. Thus the agent cannot resort to statistical tests of arbitrary accuracy to determine that her beliefs might indeed be misspecified.}

The richer observable allows comparative static results about the degree of biasedness of beliefs. We show how an analyst can use the data to construct a precise estimator of the extent of the belief biasedness of the agent and how he can compare different agents using this estimator. Moreover, since our model allows for both stochastic taste and beliefs, we show what an analyst can say about the relative \emph{speed} with which two different agents learn their taste, given their respective datasets. 

This paper is most related to \cite{lu} -- who studies the same static model but with unobservable objective states, and \cite{fis} -- who study a fully non-parametric dynamic model as here but without payoff-relevant objective states. Relatedly, \cite{dlst} study the ex-ante menu preference of the agent modeled by \cite{lu}. Among other things we extend their work to allow for stochastic taste.\footnote{Our proofs modify and extend the proofs of \cite{lu} and \cite{fis} in multiple directions as well as extending several other models in the literature. E.g. we extend \cite{as} to include objective states and stochastic beliefs. Details are in the online appendix.} Conceptually the paper is also related to \cite{lu2} who shows how a combination of ex-ante preference over acts and post-signal random choice can overcome the classical issue of identification in the Expected \emph{State-dependent} Utility model. Our model illustrates the strong identification properties of random choice data for the case of state-independent utilities in a rich dynamic environment allowing for stochastic taste. Finally, the observable in this paper can be interpreted as a likelihood function of a dynamic choice model in the spirit of \cite{rust} and the literature that it inspired.\footnote{See \cite{rust2} and \cite{ae} for surveys on the dynamic discrete choice literature.} Whereas that literature has focused on identification and inference of controlled stochastic processes, this paper offers an axiomatic treatment of such likelihood functions for choice behavior in a general set up with both observable and unobservable states. 

In the following we explain in detail the organization of the paper mentioning its contribution at each step.

Section \ref{sec:rseu} focuses on the static model. For each decision problem $A$ an analyst observes the frequency of an agent's choice and the realization of a payoff-relevant objective state $s$ (we say agent picks act $f$ from menu $A$ and objective state $s$ is realized with a certain probability $\rho(f,A,s)$). We call this observable an \emph{augmented stochastic choice function (aSCF)}. We show how the analyst can identify from this observable the space of the subjective states of the agent. We call this the \emph{revealed subjective support} of the data. We impose axioms similar to the ones in \cite{lu} to ensure that the revealed subjective support consists of SEUs that are identified by a belief $q$ about the realization of $s$ as well as a Bernoulli utility $u$. Furthermore, we show how the analyst can use the concept of the revealed subjective support to test whether the agent is using the correct data-generating process of objective states, conditional on her private information. This corresponds to the classical statistical concept of well-calibrated beliefs originating in \cite{dawid} but now in a general setting which allows for stochastic taste.
Intuitively, an agent has \emph{correct interim beliefs} only if the observed frequency of the realization of $s$ \emph{conditional on observing} $f$ chosen from $A$ is a mixture of beliefs in the subjective support of the data which can rationalize the choice of $f$ from $A$.\footnote{The last section of \cite{lu} also studies the property of well-calibrated beliefs but in a setting of non-stochastic taste.} 
Whenever this condition fails the analyst can identify the incorrect beliefs as well as the true data-generating process, conditional on the private information. We also give a relaxation of the correct interim beliefs condition which restricts the extent of belief incorrectness: the agent never receives hard evidence that her beliefs may be incorrect because the realization of $s$ is always in the support of her belief $q$. 

Section \ref{sec:drseu} introduces the dynamic model. The observable is now a \emph{history-dependent aSCF}: for every history $h^{t-1}$ occurring with positive probability, the analyst observes frequencies of the choice in a subsequent decision problem $A_{t}$ together with the realization of the objective state in the respective period (we say agent picks $f_t$ from menu $A_t$ after history $h^{t-1}$ and objective state $s_t$ is realized with probability $\rho_t(f_t,A_t,s_t|h^{t-1})$). Histories have empirical content, as they help the analyst identify the serial correlation in the private information of the agent, i.e. in her tastes and beliefs. 

We assume these history-dependent aSCFs satisfy the assumptions of the static model. In contrast to the static case there is now \emph{limited observability}: not every menu is observable after every history. This is a similar observability problem as in \cite{fis} and technically its solution in this paper adapts theirs to our more general setting with payoff-relevant states. It relies in identifying two classes of histories which reveal the same private information.\footnote{The two equivalence properties are called \emph{Contraction history independence} and \emph{Linear history independence}.} 
Whenever the observable satisfies the history-dependent version of the static model and the two history equivalence properties the analyst can identify the stochastic evolution of the private information of the agent as well as the true data-generating process of the objective states. This is the DR-SEU model, the namesake of the paper.

After establishing the main characterization result we focus on two special cases of DR-SEU whose static versions are indistinguishable: Evolving SEU, where the evolution of agent's Bernoulli utility is given through a Bellman equation and its specialization, Gradual Learning, where the agent is learning about a \emph{fixed} but unknown taste. 
Additionally, and because we need it for the dynamic characterization results, we describe when a menu preference comes from an agent who is subjectively learning both about objective states and about her Bernoulli utility/taste through a new axiom called \emph{Weak Dominance}. Intuitively, such an agent would always prefer to exchange any menu of acts $A$ for a menu $\bar A$ which allows her to pick any of the prizes occurring in $A$ with positive probability \emph{irrespective} of the realization of the uncertainty she's facing ex-ante. 

Section \ref{sec:comp} leverages the characterization theorems to prove comparative statics results. In a setting of non-stochastic taste we address the question of how an analyst can compare agents with respect to their biasedness of beliefs. Namely, given a commonly observable characteristic, e.g. gender, race or letter grades, if the analyst fixes a direction of biased beliefs for every characteristic, he can tell from stochastic choice data when an agent is more biased than another agent. Intuitively, the choice data give evidence that the more biased agent values menus \emph{uniformly} more differently to a fictitious unbiased agent than the less biased agent. Finally, in the special case of the Gradual Learning representation, we show how an analyst may distinguish when an agent's uncertainty for taste fully resolves and how the analyst may compare different agents with respect to the speed of learning their taste. Intuitively, agent 2 learns her taste more slowly than agent 1 whenever the data suggests that agent 2 satisfies Weak Dominance whenever agent 1 does.  

Section \ref{sec:conc} concludes and comments on avenues for future work. The appendix contains the proofs of the main theorem for the static setting as well as of the main theorem for the dynamic setting accompanied by a set of auxiliary results necessary to understand the main proofs. Other characterization theorems as well as technical extensions of results from several papers in the literature which are needed for the proofs are relegated to the online appendix. The latter also contains a section considering the case when the analyst does not observe the realization of objective states.

Before continuing with the theoretical set up and the results we note two examples which illustrate the questions and issues this paper addresses.

\subsection{Examples}


\subsubsection{A model of discrimination}\label{sec:disc}




Consider an employer at a job fair looking at applications for a job vacancy.\footnote{Many situations have the same structure: lending activity of a bank, university applications, etc.} The job consists of performing a task, after the job fair is concluded, whose outcome has two potential values coming from $S_1=\{g,b\}$ ($g$ stands for \emph{good} and $b$ for \emph{bad}). We assume that whether $g$ or $b$ is realized depends on both the ability of the employee as well as other randomness outside of the control of the employee.

During the job fair, in the first period of the model ($t=0$) some characteristic $s_0\in S_0 = \{s_0',s_0''\}$ of the applicant is revealed to the employer, say ethnicity, gender, education level, etc. We assume the distribution of $s_0$ over $S_0$ is known to the employer. This may be justified e.g. if the data about the prevalence of the characteristic $s_0$ in the population of the applicants at the job fair is public. In the second period ($t=1$) the employer has beliefs about the outcome of the task, conditional on the revealed characteristic $s_0$. These are coded by $(\hat q_1,\hat q_2)=(\hat q_1(g|s_0'),\hat q_1(g|s_0'')) \in (0,1)^2$. These can potentially be different from the true data generating process which here for simplicity is given by $q_1(g|s_0')=q_1(g|s_0'')=\frac{1}{2}$. Assume here for simplicity that the analyst knows this data-generating process.
In our example we say that the employer has incorrect beliefs if the following holds.
\vspace{3mm}\\
\textbf{Incorrect Beliefs:} $1>\hat q_1(g|s_0') >\frac{1}{2} >  \hat q_1(g|s_0'')>0$.\footnote{Other assumptions are possible. These here are for definiteness.} 
\vspace{3mm}\\




We assume in the following that the objective state $s_1$ (task outcome) is also observable to the analyst after the choice of the employer.



Given the observed characteristic $s_0$ the employer can choose in $t=1$ whether to hire the candidate (formally, act $h_{s_0}:S_1\ra \R$ ) or not hire (act $nh_{s_0}:S_1\ra \R$). In the case of not hiring, the utility of the employer is always zero $u_{s_0}(nh_{s_0}(s_1))= 0$ for all $s_0\in S_0,s_1\in S_1$. 

In the case  of hiring the employer has (possibly) stochastic utility $u_{s_0}:\R\ra\R$ which satisfies 

\[
u_{s_0}(h_{s_0}(g)) = g_{s_0},\quad u_{s_0}(h_{s_0}(b)) = b_{s_0}\text{ with } g_{s_0}> 0> b_{s_0}\text{ almost surely}. 
\]

Stochastic utility conditional on the realization of $s_0$ is meant to capture the possibility that the utility of a successful task for the employer may depend on the specific task to be solved, here assumed unobservable to the analyst, besides on the characteristic $s_0$ of the employee. It may also happen due to other characteristics of the candidate besides $s_0$ which are unobservable to the analyst but relevant to the employer.\footnote{The employer may have lexicographic preferences; she cares about $s_0$ first and foremost but given $s_0$ also takes into account other unobservable features of the candidate.} 
Finally, we assume that whenever the employer is indifferent between hiring and not hiring a candidate he uses an unbiased coin to break ties.



Besides biases in beliefs we allow for the possibility that the employer cares about the realization of $s_0$ as well. We require for the random variables $g_i,b_i,i=1,2$ to be jointly continuously distributed and to fulfill the following condition.

\[
(C)\quad g_{s_0'}\geq g_{s_0''} > 0> b_{s_0'}\geq b_{s_0''}\quad\text{almost surely}.
\]
A successful task benefits the employer more -- and a failed one hurts him less -- if it is the deed of an agent of characteristic $s_0'$ rather than $s_0''$. That is, the employer incurs uniformly lower payoffs from $s_0''$ for each outcome.


We say that the employer cares about $s_0$ if the following holds. 
\vspace{3mm}\\
\textbf{Preference for $s_0'$:} $\quad g_{s_0'}> g_{s_0''} > 0> b_{s_0'}> b_{s_0''}\quad\text{almost surely}$.
\vspace{4mm}\\
Here we ask for the `extreme' inequalities in condition $(C)$ to hold strictly almost surely.\footnote{Just as for beliefs other assumptions are here possible.}

Assume now that an analyst has frequency data on both hiring decisions at the job fair and on the outcome of the task, even though she may not observe the precise type of the task in every instance. Thus for all $s_0 = s_0',s_0''$ and $s_1 = g,b$ the analyst observes the limiting frequency that candidate $s_0$ is hired, and that state $s_1$ is realized, denoted by $\rho_{s_0}(h_{s_0},\{h_{s_0}, nh_{s_0}\},s_1)$. This paper gives conditions on stochastic choice data which allows the following. 

- As a first step the analyst can confirm that the true data-generating process is unbiased, i.e. that $q_1(g|s_0')=q_1(g|s_0'')=\frac{1}{2}$ holds. This corresponds to the constraint $$\rho_{s_0'}(f_{s_0'},\{f_{s_0'}, h_{s_0'}\},g) = \rho_{s_0''}(f_{s_0''},\{f_{s_0''}, h_{s_0''}\},g) = \frac{1}{2}.$$

- The analyst can also discern from stochastic choice data whether there is bias in beliefs, whether the employer cares about the realization of $s_0$ or whether both are occurring simultaneously. 

Namely, whenever the employer is unbiased in beliefs and doesn't care about the realization of $s_0$ \emph{per se} he chooses to hire either candidate with the \emph{same} positive probability. This corresponds to the constraint

\[
\sum_{s_1}\rho_{s_0'}(f_{s_0'},\{f_{s_0'}, h_{s_0'}\},s_1) = \sum_{s_1}\rho_{s_0''}(f_{s_0''},\{f_{s_0''}, h_{s_0''}\},s_1). 
\]

Whenever there is either bias in beliefs or the employer has preference for $s_0'$ he hires candidate $s_0'$ with strictly higher probability than candidate $s_0''$.

\begin{equation}\label{eq:helpemp}
\sum_{s_1}\rho_{s_0'}(f_{s_0'},\{f_{s_0'}, h_{s_0'}\},s_1) > \sum_{s_1}\rho_{s_0''}(f_{s_0''},\{f_{s_0''}, h_{s_0''}\},s_1)
\end{equation}

Finally, whenever the employer has incorrect beliefs \emph{and} has preference for $s_0'$, all else equal he hires candidate $s_0'$ with a (weakly) higher probability than in the case of either bias in beliefs only or preference for $s_0'$ only. This corresponds to a larger gap in \eqref{eq:helpemp}. 

This example shows that stochastic choice data coming from standard subjective expected utility (SEU) maximizers can be used to identify biases, whenever the analyst gets information for the realization of the objective state (here whether the task is successful or not). As we show, stochastic choice data allow comparisons of different employers in terms of their biases in much more complicated examples than the current one. 








\subsubsection{Educational choices}\label{sec:educ}


Consider an undergraduate student who adheres to subjective expected utility (SEU) and has beliefs about the final outcome in the job market once she graduates. This outcome comes from a finite objective state space, say,  
$$S =  \{\text{job in finance, job in tech industry, job in government, graduate school, start-up}\}.$$ 
At the beginning of the undergraduate education the student is also learning about her taste regarding possible careers and so has stochastic tastes $\tilde v_0,\tilde v_1,\dots,\tilde v_{\tau}$ about the final outcome. At the end of some student-specific year $\tau\ge 1$, learning about taste ceases: the student has a fixed Bernoulli utility $v$ about the final outcome $S$ even though her beliefs $q_t$ about the final outcome in $S$ remain stochastic throughout the whole higher education experience. 

Formally, let school years be encoded by $t\in \{0,1,\dots,T\}$. Let $s_t$ be a period-$t$ signal about final outcome coming from a finite space of objective signals $S_t$. These can be grades or feedback from faculty, experiences in internships, etc. Let acts (decisions of a student) correspond to jobs/projects/classes she engages with in each year and menus $A_t$ be finite collections of such acts the student can choose from in each education year. Denote the set of menus available in period $t$ by $\A_t$. Given a realized signal $s_t$ each act $f_t$ in period $t$ delivers a lottery over pairs consisting of an instantaneous prize from a finite set of prizes $Z$ and a continuation decision problem $A_{t+1}$ from $\A_{t+1}$.\footnote{There is no continuation problem in period $t=T$.} The realization of the continuation problem $A_{t+1}$ corresponds to jobs/internships/classes possibly available to the student, after she has taken a current class corresponding to the act $f_t$. Say that an act $f_t$ is \emph{constant}, if the lottery over pairs of current prize and continuation decision doesn't depend on the realization of the signal $s_t$, i.e. it is the same for all $s_t$ in $S_t$. E.g. a constant act is a summer job a student may take only due to financial reasons and which doesn't enhance her intellectual skills in the job market for any possible career. 

The analyst observes past choices of, say act $f_{l}$ chosen from menu $A_{l}$ as well as the realization of signal $s_l\in S_l$, and for the current period $t\in \{1,\dots,T\}$ she observes after the history $h^{t-1}=(f_0,A_0,s_0;\dots;f_{t-1},A_{t-1},s_{t-1})$ the frequencies of triples $(f_t,A_t,s_t)$. These history-dependent frequencies, denoted by $\rho_t(f_t,A_t,s_t|h^{t-1})$, are to be interpreted as \emph{after history $h^{t-1}$ student chose $f_t$ when facing $A_t$ and the objective signal $s_t$ was realized.}

If the history-dependent preference of the student over menus/decision problems from $\A_{t},t=0,\dots,T$ were observable, it is intuitive to expect it satisfies the following properties. 

\begin{enumerate}
\item \textbf{Preference for Flexibility:} Every year the student prefers menus which are larger rather than subsets thereof. That is, $B_t\in \A_t$ is less valuable than $A_t\in \A_t$ if $B_t\subset A_t$. This is because a strict subset offers less option value for a SEU agent than a full menu. 
\item \textbf{Weak Dominance for $t\le \tau$:} At $\tau = 0$, say, she prefers to replace a single act $f_1$ whose utility depends on the realization of the signal $s_1$ with a menu of constant acts $\bar A=\{f_1(s_1):s_1\in S_1\}$ offering the same outcomes (lotteries over $Z\times\A_{2}$) as every $s_1-$dependent outcome of $f_1$. This is because menu $\bar A$ offers insurance against her stochastic taste in $t=1$. Intuitively, summer jobs where the student doesn't learn new specialized skills for the job market may be more valuable to a student who is still unsure of her taste about different careers than committing to an internship whose outcome is highly dependent on what she learns about her career taste at the end of the current period.
\item \textbf{Strong Dominance for $t> \tau$:} From \emph{the end} of period $\tau$ on, whenever the act $f_{t+1}$ delivers weakly \emph{better} utility for each realization of the signal in period $t+1$ than $g_{t+1}$, from the perspective of the \emph{end} of year $t$, the menu $\{f_{t+1},g_{t+1}\}$ is as good as $\{f_{t+1}\}$. This unambiguous comparison of continuation problems in the end of year $t$ becomes possible because at the end of period $\tau$ the career tastes of the student have stabilized and are deterministic.\footnote{The names are justified: after formally introducing the technical set up and the axioms in the main body of the paper we show in online appendix section 5 that under Preference for Flexibility, Strong Dominance implies Weak Dominance but not the other way around.} Given a fixed taste about distinct careers she is able to at least determine when an act is uniformly more valuable than another, no matter the realization of the objective signal in the current period $t$.
\end{enumerate}
We show how the properties A-C can be derived from ex-post stochastic choice from menus without knowing anything about the preference over menus of the student. Moreover, our methods allow the analyst to also determine the \emph{speed} with which an agent, such as the student in this example, learns her final taste $v$ (e.g. to determine the $\tau$ of the student). For example, if the act $f_1$ is taking an internship which requires substantial investment in learning new skills in a very specific field like finance, i.e. an act whose outcome is highly dependent on $s_1$ as well as the realization of the future taste $\tilde v_1$, we should expect an agent who knows by the end of period $t=0$ that her taste is so that she likes to get a job in finance, to prefer committing to $f_1$ at the end of $t=0$. This should be especially the case if the alternative is to face a menu which offers acts whose outcomes don't depend much on $s_1$ or the realization of $\tilde v_1$ such as helping out with grading an undergrad class, taking up a summer job in the library, etc, even though they might be as financially profitable as picking the internship in finance $f_1$. 

Finally, given richness of the data, our characterization results show how an analyst is able to compare different agents according to their speed of learning about taste in similar situations.\footnote{Intuitively in our example, student 1 learns her taste faster than student 2, if stochastic choice data give evidence that student 1 satisfies Strong Dominance whenever student 2 does.} 

\section{Static Random Subjective Expected Utility with observable objective states}\label{sec:rseu}

In this section we introduce and characterize the static model. This is the crucial building block of the dynamic model of section \ref{sec:drseu}. 

\paragraph{Set up in the static model.}

Let $Z$ be a prize space assumed to have a separable, metric topological structure. Let $S$ be a finite set of objective states\footnote{The wording \emph{objective} means that the state $s$ is verifiable by both agent and analyst after it occurs.} and $\F$ the set of Anscombe-Aumann \emph{acts} (AA acts) with a typical element given by $f:S\ra\Delta(Z)$ where $\Delta(Z)$ denotes the space of \emph{simple} lotteries over prizes in $Z$.\footnote{A lottery is called \emph{simple} if only finitely many prizes can happen with positive probability. $\Delta(Z)$ is equipped with the topology of weak convergence of probability measures. The set of acts $\F$ is equipped with the product-topology over $\Delta(Z)^S$.} Finally, denote by $\A$ the collection of finite, nonempty subsets of $\F$. A typical element in $\A$ is called a \emph{menu} and denoted by a capital letter, e.g. $A\in \A$. $\A$ is equipped with the Hausdorff topology.  

Given a belief of the agent over $S$, i.e. an element $q$ from $\Delta(S)$ and an Expected Utility function $u:\Delta(Z)\ra\R$ to evaluate simple lotteries, we say the agent satisfies \emph{Subjective Expected Utility (SEU) with beliefs $q$ and taste $u$} if the utility of an act $f$ is given by $q\cdot (u\circ f): = \sum_{s\in S}q(s)u(f(s))$.\footnote{In the following we often identify the EU-functional $u:\Delta(Z)\ra\R$ with its Bernoulli utility from $\R^Z$.} 

Define 
\[
N(A,f) = \{(q,u)\in \Delta(S)\times \R^X: q\cdot (u\circ f)\geq q\cdot (u\circ g), g\in A\}. 
\]
This is the set of SEUs which can rationalize the choice of $f$ from menu $A$.

 Denote $N^+(A,f)$ the respective subset of $N(A,f)$ where $f$ is not tied to other acts from $A$. 

Moreover, define
\[
M(A;u,q) = \{f\in A: q\cdot (u\circ f)\geq q\cdot (u\circ g), g\in A\}.
\]
This is the set of maximizers when the agent's belief about objective state of the world is $q$ and her Bernoulli utility is $u$. 

The timeline of the one-period model is the following.

\begin{figure}[H]
\centering
\includegraphics[width=13cm]{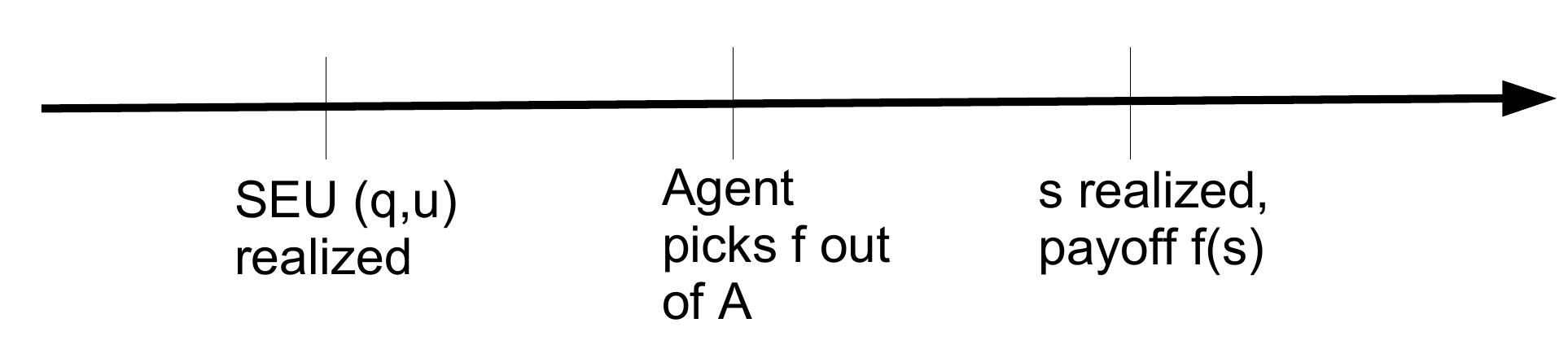}
\caption{Timeline for the static setting.}
\label{fig:timelinestatic}
\end{figure}

Let $\mu$ be a probability measure over $\Delta(S)\times \R^X$, equipped with the sigma-algebra $\mathcal{F}$ generated by sets of the form $N^+(A,f), N(A,f)$ or alternatively with the Borel sigma-algebra of $\Delta(S)\times \R^X$.\footnote{This is constructed as a product sigma-Algebra of the respective Borel sigma-Algebra of weak convergence on $\Delta(S)$ and the Borel one from $\R^X$ (the latter again a product sigma-Algebra).} 

We say that $\mu$ is \emph{regular} if $\mu(N^+(A,f)) = \mu(N(A,f))$ for any $A\in \A,f\in \F$. In this paper regular measures $\mu$ have the following form: whenever there are ties, i.e. $M(A;u,q)$ is not a singleton for some $A$ and SEU pair $(q,u)$ the agent \emph{randomly} picks an auxiliary SEU pair $(p,v)$ such that 
$M(M(A;u,q);p,v)$ \textbf{is} a singleton.\footnote{The interested reader can peruse the proofs in Section 1 in the online appendix for the mathematical details. This tie-breaking rule is special and it will be reflected in the properties of the data in the form of a specific axiom: Extremeness-type of Axioms (see next subsections) \emph{imply} tie-breaking through SEUs.}

\paragraph{Observable in the static model.}
We assume the analyst observes an \emph{augmented stochastic choice function} defined as in part 1) of the following definition.

\begin{definition}\label{thm:defaSCFstatic}

1) An \emph{augmented stochastic choice function (aSCF)} is a map $\rho:\F\times \A\times S\ra[0,1]$ with the properties
\begin{enumerate}[(a)]
\item 
\[
\sum_s\sum_{f\in A}\rho(f,A,s) = 1,\quad\forall A\in\A.
\]
\item 
\[
\rho(s):=\sum_{f\in A}\rho(f,A,s) = \sum_{f\in B}\rho(f,B,s)>0,\quad \forall A,B\in \A,s\in S.
\]
\vspace{2mm}
2) A \emph{stochastic choice function (SCF)} $\zeta$ is a map $\zeta:\F\times \A\ra[0,1]$ with the property 
\[
\sum_{f\in A}\zeta(f,A) = 1,\quad\forall A\in\A
\]

\end{enumerate} 
\end{definition}

The second requirement in the definition of aSCF makes sure that we can define the observed frequency of objective state $s$ independently of the decision problem the agent is facing. This says that objective uncertainty is fully exogenous and independent from the problem the agent is facing in addition to being outside the influence of the analyst. Formally, it allows the definition of $\rho(s) = \sum_{f\in A}\rho(f,A,s)$ for any $A$ and $s\in S$, i.e. the probability of observing $s$ in the data. 

For a given aSCF $\rho$ we denote in the following by $\bar \rho$ the SCF derived from summing each $\rho(f,A,s)$ across states. Formally, 

\[
\bar\rho(f,A) := \sum_{s\in S}\rho(f,A,s),\quad f\in A,A\in \A.
\]
\paragraph{Discussion of the Observable.}
Assuming that the data of the analyst comes in the form of aSCFs characterizes an analyst with superior information compared to the set up of \cite{lu}. In many realistic situations this is a viable assumption: loan performance data, how students perform in school or how an employee performs in some task is often observable to an outside analyst.\footnote{Section \ref{sec:rseu} of the online appendix considers extensively the case when the observable corresponds to SCF, that is the realization of $s$ is not observable by the analyst. For the static setting the whole theory, up to explicit modeling of the tie-breaking is contained in \cite{lu}, whereas the dynamic version of his model can be derived easily using the approach of \cite{fis}. See the online appendix for more details.} 

\cite{ellis} and \cite{cd2} also consider state-dependent choice data but have a different focus: that of information acquisition in a static setting. They don't study the question of misspecified learning, either because the analyst doesn't get to see the realization of the objective state or because they assume from the start that the agent is using the correct prior. \cite{cm} considers state-dependent stochastic choice data in a passive learning model similar to ours but assume that the taste of the agent is deterministic and known to the analyst. 



\begin{remark}\label{thm:remk1} The observable in Definition \ref{thm:defaSCFstatic} has more general applicability, e.g. it can be used even if there is partial observability of $s$ as long as there is full identification \emph{in the aggregate}. 

In more detail, assume the analyst observes a signal $y\in Y$ about the true realization of the objective state $s\in S$ instead of its realization. If $\hat\mu(y|s)$ gives the (menu-independent) conditional probability of observing signal $y$ when the realized state is $s$ the assumption of aSCFs as observable is valid for the analysis if the following two conditions hold:
\\- $\hat\mu$ is known by the analyst,
\\- The matrix $(\hat\mu(y|s))_{y\in Y,s\in S}$ is quadratic and has full rank. 
\end{remark}



\subsection{Representation in the Static Setting}

We now introduce the Random Subjective Expected Utility representation for an aSCF $\rho$ we are after. An agent has private information about both beliefs over the realization of the objective state $s$ as well as her taste $u\in\R^Z$. The analyst observes only aggregate frequencies of choice data and realizations of the objective state from the same agent in many choice instances or similar aggregate data choices from a population of agents. 

\begin{definition}\label{thm:filtras-r-seu}
A Random SEU representation (R-SEU) of the aSCF $\rho$ is a tuple \\$(\Omega,\mathcal{F}^*,\mu,(q,u,s), (\hat q,\hat u))$ such that 
\begin{enumerate}
\item $(\Omega,\mathcal{F}^*,\mu)$ is a probability space with finite $\Omega$, 
\item $(q,u,s):\Omega\ra \Delta(S)\times\R^{Z}\times S $ is an injective map, has non-constant SEU $(q(\omega),u(\omega))$ and $s(\omega)\in supp(q(\omega))$ for all $\omega\in\Omega$.
\item Either 

C1. The representation has correct interim beliefs (cib):  $\mu(s\in\cdot|q) = q(\cdot)$

or otherwise

C2. The representation has no unforeseen contingencies (nuc):  \\$supp\left(\mu(s\in\cdot|q,u)\right) \subset supp\left(q(\cdot)\right)$.
\item the $(q,u)$-measurable tiebreaking process $(\hat q,\hat u):\Omega\ra\R^{Z}$ is regular and for all $f\in A$,
\[
\rho(f,A,s) = \mu(C(f,A,s)). 
\]
\end{enumerate}
Here, $C$ is defined as 
\[
C(f,A,s) = \{\omega\in \Omega: f\in M\left(M(A,q(\omega),u(\omega)),\hat q(\omega),\hat u(\omega)\right), s(\omega) = s\}.
\]
\end{definition}

In the following $\omega$ are called \emph{states of the world.}
$C(f,A,s)$ denotes then the collection of states of the world where the agent chooses $f$ from $A$ and the objective state $s$ is realized.

Before continuing, we note down the true data-generating process (DGP) derived from the representation. 

\begin{definition}\label{thm:DGP}
For an aSCF $\rho$ that satisfies a R-SEU representation define the DGP, a $\Delta(S)$-valued random variable $\bar q:\Omega\ra\Delta(S)$ as\\ 
$\bar q(\omega)(\cdot) = \mu(s\in\cdot|q,u)(\omega)$. Then the property of correct interim beliefs (cib) can be written as 
\[
\bar q = q
\]
whereas that of unforeseen contingencies (nuc) is written as 
\[
supp(\bar q)\subset supp (q). 
\]
\end{definition}

\subsection{The revealed subjective support of a SCF}

For this subsection only, we look at an agent whose preference $\better$ over acts is \emph{continuous} but otherwise arbitrary (i.e. not necessarily SEU) and introduce a concept which is helpful in the characterization results of this paper in addition to having general applicability outside of this model as well. If the only fact the analyst knows about the stochastic choice of an agent is that it comes from a continuous preference, the sets $N(f,A)$ can be written as

\[
N(f,A) = \{\better\text{ continuous preference over }\F:f\better g, g\in A\}. 
\]
Say that the stochastic choice data of an agent satisfies  a \emph{Random Utility Model} if the stochasticity in choice follows from the randomness of her preference. Formally, we define as follows.
\begin{definition}[Random Utility Model]\label{thm:defrum}
Say that a SCF $\zeta$ on $\mathcal{F}$ satisfies a \emph{Random Utility Model} (RUM) if there exists a regular probability measure $\mu$ over continuous preferences over $\mathcal{F}$ so that for every $A\in \A$ and $f\in A$ we have 
\[
\zeta(f,A) = \mu(N(f,A)).
\]
\end{definition}

The randomness in preferences may originate from her stochastic perceptions of the decision environment she faces, for example in the special case of SEUs her beliefs may be stochastic. In the case of SEUs randomness can also come from stochastic tastes. 

Alternatively, a RUM may be interpreted as representing data from a population of heterogeneous agents who have deterministic preferences. The following definition shows how to identify from data the collection of preferences underlying a RUM. 


\begin{definition}\label{thm:RSSuppdef}
For a SCF $\zeta$ which satisfies a RUM let $RSSupp(\zeta)$, the revealed subjective support of $\zeta$, be defined through

\begin{align*}
RSSupp(\zeta) = &\{\better \text{ over }\F: \forall A\in \A, f\in A,\text{if }\better\in N(A,f)\\&\text{ then there exists }(f_n,A_n)\ra(f,A)\text{ with }\zeta(f_n,A_n)>0\}.
\end{align*}
Here convergence $(f_n,A_n)\ra(f,A)$ is in the product topology of $\F\times\A$.
\end{definition}

This says that a preference $\better$ is in the revealed subjective support of $\zeta$ if every choice that can be rationalized by $\better$ appears in the data encoded by $\zeta$, up to tie-breaking.\footnote{In more detail: $\better$ occurs in the data if for every choice pair $(f,A)$ either (1) $\rho(f,A)>0$ \emph{and} $\better\in N(A,f)$ or if (2) $\rho(f,A)=0$ and $\better\in N(A,f)$ then $\rho(f,A)=0$ only happens due to tie-breaking.} 


If the RUM has support on SEUs, the definition `picks out' the SEUs in the support of $\mu$ from Definition \ref{thm:defrum} up to positive affine transformations of the respective Bernoulli utilities.

\paragraph{Aside.} Another compact and suggestive way to write down the revealed subjective support of a SCF $\zeta$ is as follows. 

For a continuous preference $\better$ over $\F$ denote the set of choices it can rationalize as $R^{\better}$, that is                                                                                   
\[
R^{\better} =\{(f,A)\in \F\times\A: \better\in N(f,A) \}. 
\]

This is the set of choice data that are consistent with maximization of $\better$. 

The set of choices explained by the data represented by some SCF $\zeta$ is 

\[
N(\zeta) =\{(f,A): f\in A,\text{ }\exists (f_n,A_n)\ra (f,A)\text{ with }\zeta(f_n,A_n)>0\text{ for all }n\}.
\]

Then $RSSupp(\bar\rho)$ can be characterized as follows.

\[
RSSupp(\zeta) =\{\better: R^{\better}\subset N(\zeta) \}.
\]


%

\subsection{Axiomatization of aSCFs}

The following axiomatization of aSCFs is based on previous results about the axiomatization of SCFs in \cite{lu} and \cite{as}. 

Axioms 0-1 till 0-5 below are adaptations to our setting of aSCFs of the standard axioms from Theorem S.1 of \cite{lu}. They imply that an aSCF comes from an underlying RUM whose \emph{revealed subjective support} contains only SEUs. Axiom 0-6 is adapted from \cite{as} and ensures that there can only occur finitely many such SEUs. 
\paragraph{Standard Axioms in statewise form.}
For all $s\in S$ it holds
 \paragraph{Axiom 0-1: Statewise Monotonicity.}  $\rho(f,A,s)\geq \rho(f,B,s)$ for $A\subset B$. 
\paragraph{Axiom 0-2: Statewise Linearity.} $\rho(\lambda f+ (1-\lambda)g, \lambda A+ (1-\lambda)\{g\},s) = \rho(f,A,s)$ for any $A\in \A, g\in\F$ and $\lambda\in (0,1)$. 
\paragraph{Axiom 0-3: Statewise Extremeness.} $\rho(ext(A),A,s) = 1$ for all $A\in \A$.\footnote{Note that $\F$ has a mixture structure in the usual way. In particular, one can form $conv(A)$, the convex hull of $A$ for any menu $A$. Then $ext(A)$ is identified with the set of extremum points of $conv(A)$.} 
\paragraph{Axiom 0-4: Statewise Continuity.} $\A\ni A\mapsto\rho(\cdot, A|s)$ is continuous.\footnote{The image of the mapping is the space of simple lotteries on $\F$, equipped with the topology of weak convergence of probability measures.}
\paragraph{Axiom 0-5: State Independence.}

To explain this axiom we first introduce some terminology: a menu $A$ is called \emph{constant} if it contains only constant acts. Given a menu $A$ and a state $r\in S$ let $A(r) = \{f(r): f\in A\}$ be the constant menu containing all lotteries from acts in $A$ which happen at state $r$. 

Then \textbf{State Independence} says: Suppose $f(s_1) = f(s_2), A_1(s_1) = A_2(s_2)$ and $A_i(s) = \{f(s)\},s\neq s_i, i=1,2$. Then $\rho(f,A_1,s) = \rho(f,A_1\cup A_2,s)$. 

Intuitively, if an act $f$ yields the same payoff in states $s_1$ and $s_2$, payoffs of menu $A_1$ in $s_1$ are the same as those of menu $A_2$ in $s_2$ and acts in $A_i$ only differ in $s_i$ then the probability of choosing $f$ in $A_1$ is the same  as choosing $f$ in $A_1\cup A_2$, \emph{unless} the realization of the Bernoulli utility of the agent depends on whether $s_1$ or $s_2$ is realized. 
\paragraph{Axiom 0-6: Statewise Finiteness.}
There is $K>0$ such that for all $A\in \A$, there is $B\subset A$ with $|B|\leq K$ independent of $s$ such that for every $f\in A\setminus B$ there are sequences $f^n\ra^m f$ and $B^n\ra^m B$ with $\rho(f_n,\{f_n\}\cup B^n,s) = 0$.

\vspace{4mm}
To state the axiom of correct beliefs we define for a SEU pair $(q,u)$ where $p$ is the belief of the agent and $u$ her Bernoulli utility as $\pi_q(p,u) = p.$ That is, the projection to the belief used from the agent.
Furthermore, in the following $\rho(s|f,A)$ is the conditional probability of observing the realization of the objective state $s$ in the data conditional on the agent choosing $f$ from menu $A$.  

\paragraph{Axiom 0-7: Correct Interim Beliefs (CIB).} For all $f\in \F$ and $A\in\A$ with $\bar\rho(f,A)>0$ we have 
\begin{equation}\label{eq:CIB}
\rho(\cdot|f,A)\in \pi_{q}\left(conv\left(N(f,A)\cap RSSupp(\bar\rho)\right)\right)= conv\left(\pi_q(N(f,A)\cap RSSupp(\bar\rho))\right).
\end{equation}

The axiom says that the DGP of the objective state $s$ conditional on observed choice $(f,A)$ is a mixture of beliefs which correspond to some SEU that fulfill two natural conditions \emph{simultaneously}: 1) the SEU is contained in the revealed subjective support of the data and 2) the SEU rationalizes the choice $f$ from $A$.  






Incorrect beliefs can arise due to different reasons: the agent may observe objective signals with noise, she may have a misspecified prior or otherwise have \emph{subjectively} biased beliefs even though they average out to the correct prior. We exclude in this paper the case when incorrect beliefs originate from non-Bayesian updating with respect to \emph{any} prior.

In contrast to section 6 of \cite{lu} here the analyst gets information about the realization of the objective state and can glean out the true DGP from data. This allows her to make a direct comparison between the true DGP and the beliefs of the agent.\footnote{Moreover, in the dynamic model in Section \ref{sec:drseu} we assume that the agent is sophisticated and thus our model doesn't allow any \emph{prospective overconfindence/underconfidence} as in \cite{lu}.} Section 7 of \cite{lu} constructs a test of CIB based on test acts. His methods require non-stochastic taste whereas our axiom is robust to stochasticity of tastes.

Now we present a relaxation of the Correct Interim Beliefs Axiom which allows for incorrect beliefs but so that the incorrectness remains undetected by the agent ex-post. This is inconsequential in a static setting but has repercussions in the dynamic setting of Section \ref{sec:drseu} where we study an agent who \emph{passively learns} about objective states as well as her taste in every period. 

\paragraph{Axiom 0-7': No Unforeseen Contingencies (NUC)} For all $f\in \F$ and $A\in\A$ with $\bar\rho(f,A)>0$ it holds 
\[
supp\left(\rho(\cdot|f,A)\right)\subset \bigcup \{supp(q): q\in \pi_q(N(A,f)\cap RSSupp(\bar\rho))\}.
\]


Our first main result gives the axiomatization of aSCFs in a static setting. 

\begin{theorem}\label{thm:SREU}
The aSCF $\rho$ on $\A$ admits a R-SEU representation with CIB satisfied if and only if 
it satisfies Axioms 0-1 till 0-7. It admits a R-SEU representation with NUC satisfied if and only if 
it satisfies Axioms 0-1 till 0-6 together with Axiom 0-7'.
\end{theorem}

In the following whenever for an aSCF $\rho$ the Axioms 0-1- till 0-6 together with 0-7' are satisfied, we say \emph{Axiom 0 is satisfied for $\rho$.} 

\subsubsection{Informational Representation for aSCFs}\label{subsec:infoR}

We consider here the special case of Theorem \ref{thm:SREU} where all possible Bernoulli utilities in the representation are equal up to positive affine transformations of each other. This implies that stochasticity in choice only comes from randomness in beliefs. 

To facilitate analysis, we require the existence of a best constant act. This requirement is easily expressed in terms of stochastic choice. 
 
\paragraph{Axiom: Existence of a constant best act.} There exists a constant act $\bar f\in \F$ such that for every act $f\in \F$ it holds 

\[
f\neq \bar f\quad\imply\quad\rho(f,\{f,\bar f\}) = 0.
\]
\vspace{3mm}

The existence of a best constant act is assured for example if $Z$ consists of monetary prizes and the preferences of the agent over money are strictly increasing. Whenever this Axiom is satisfied, it becomes easier to eschew tie-breaking considerations when writing down other Axioms on data. 



The axiom on data which ensures that the agent has a deterministic taste is the following.\footnote{This is an adaptation of the \emph{C-Determinism} Axiom from the \cite{lu} who doesn't consider tie-breaking explicitly as we do.} 

\paragraph{Axiom: C-Determinism*.} For any menu $A$ consisting of constant acts it holds true
\[
\lim_{a\ra 1}\rho\left(af+(1-a)\bar f; A\setminus \{f\}\cup\{af+(1-a)\bar f\}\right)\in \{0,1\}.
\]

This says that except for possible stochastic tie-breaking, constant acts are chosen deterministically. On the other hand, if taste is stochastic then choice from constant menus should be stochastic, even after taking into account possible stochastic tie-breaking. Given this intuition the following characterization result is not surprising.

\begin{proposition}
[Informational Representation for aSCFs]\label{thm:inforepresentationascf}

Assume that an aSCF $\rho$ has a R-SEU representation with regular measure $\mu$.  
Assume that there exists a constant best act. 

Then the following are equivalent. 
\begin{enumerate}
\item For all $(q,u),(p,v)\in RSSupp(\bar\rho)$ $u$ is a positive affine transformation of $v$.
\item $\rho$ satisfies C-Determinism*.
\end{enumerate}

\end{proposition}


\section{Dynamic Random Subjective Expected Utility}\label{sec:drseu}

This section is devoted to the dynamic model. We introduce the general representation and two interesting specializations of it. After that, we give axioms for all three representations. 

\paragraph{Set up in the dynamic model.}


Let $Z$ be a \emph{finite} prize space, $\infty>T\ge 1$ and for each $t=0,\dots,T$ let $S_t$ be finite spaces of \emph{objective states}. The objective states evolve according to a DGP which cannot be influenced by the agent (passive learner situation).

Define recursively the spaces of \emph{consequences} for every period as follows. Let $X_T=Z$ and the set of acts $\F_T$ with a typical element $f_T:S_T\ra\Delta(Z)$. Let $\A_T$ be the collection of finite sets from $\F_T$. Then continue inductively by defining $X_t=Z\times \mathcal{A}_{t+1}$, where $\mathcal{A}_{t+1}$ is the collection of finite menus from $\F_{t+1}$. $\F_t$ is then the set of acts $f_t:S_t\ra\Delta(X_t)$.\footnote{Furthermore we denote in the following by $\A^c_{t}$ the collection of period$-t$ menus consisting of constant acts.}

Thus, an act $f_t$ at time $t<T$ gives for each possible objective state $s_t$ a lottery over current consumption and a continuation decision problem/menu. We denote $f_t^A$ the marginal act on menus $A_{t+1}$ and $f_t^Z$ the marginal act on $Z$ induced by $f_t$.


We assume in each period $(q_t,u_t)$ is private information of the agent whereas the realization of $s_t$ is observed by both the agent and the analyst. Thus stochasticity in choice comes from the information asymmetry between the agent and the analyst in the single-agent interpretation, whereas in the population interpretation the analyst is observing \emph{dynamic} data from a population of SEU agents whose preference characteristics are unknown. 




Visually the timeline is depicted in Figure \ref{fig:timeline}.

\begin{figure}[H]
\centering
\includegraphics[width=15cm]{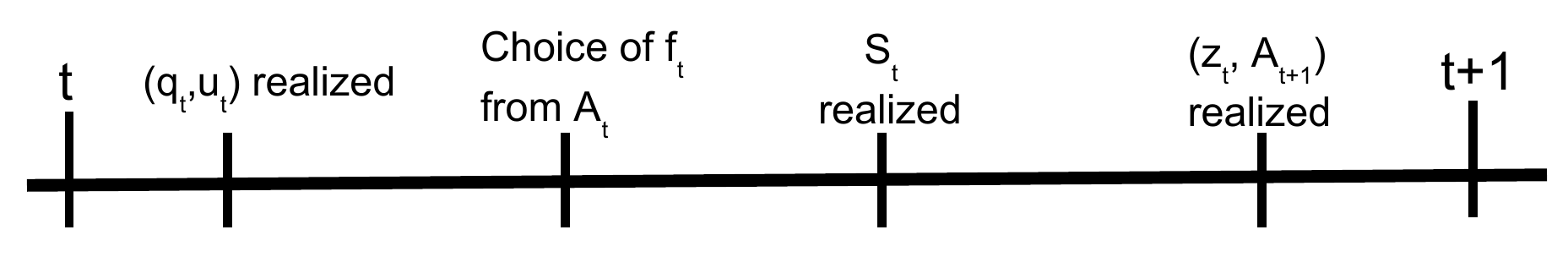}
\caption{Timeline for the dynamic setting.}
\label{fig:timeline}
\end{figure}
\vspace{4mm}









\paragraph{The observable in the dynamic setting.}

The analyst observes histories with a typical element $h^t$ as well as \emph{history-dependent} aSCFs $\rho_{t}(\cdot|h^{t-1})$. The collection of the former is denoted by $\h_t$ whereas of the latter simply by $\rho$ and called \emph{a dynamic augmented stochastic choice function (dynamic aSCF)}. These are described recursively as follows. For $t=0$ the analyst observes an aSCF $\rho_0$ as in Definition \ref{thm:defaSCFstatic}. The set $\h_0$ collects all histories $h^0=(f_0,A_0,s_0)\in \F\times\A_0\times S_0$ such that $\rho_0(h^0)>0$. For $h^0\in\h_0$ denote $\A_{1}(h^0): =supp(f_0^A)$ the set of period$-1$ menus that follow $h^0$ with positive probability. The construction is continued recursively: for any history $h^{t}\in\h_t$ there is an aSCF $\rho_{t+1}(\cdot|h^t)$ which can be used to define the set of possible continuation menus $\A_{t+1}(h^{t})$. The set of period$-(t+1)$ histories is then $\h_{t+1}: =\{(h^t,f_{t+1},A_{t+1},s_{t+1}):A_{t+1}\in \A_{t+1}(h^t), \rho_{t+1}(f_{t+1},A_{t+1},s_{t+1}|h^t)>0\}$.  

In simple words: histories are finite sequences of triplets $(f_i,A_i,s_i)$ with the interpretation that the data shows that with positive probability $f_i$ is chosen from menu $A_i$ and $s_i$ is the realized objective state in period $i$. Moreover, a history can only happen if the elements $(f_i,A_i,s_i)$ of its sequence happen \emph{successively} with positive probability starting from the `oldest' one $(f_0,A_0,s_0)$ to the most recent.


The data reflects \emph{limited observability} in the sense that $\rho_t$ is defined only conditional on histories which happen with positive probability in the data. We show below how this can be overcome. 

\subsection{Representations} 

We first define properties shared by all representations. The focus is on having properties which are tractable but still allow for a general enough representation.



\subsubsection{Simplicity, regularity and preference-based tie-breaking.}

Say that the triple $\mathbf{(\mathcal{F}_t,q_t,u_t,s_t)_{0\leq t\leq T}}$ is \emph{simple} w.r.t.\footnote{w.r.t. stands for \emph{with respect to}.} the probability space $(\Omega,\mathcal{F}^*,\mu)$ if 

\begin{enumerate}
\item each $\mathcal{F}_t$ is generated by a finite partition such that $\mu(\mathcal{F}_t(\omega))>0$ for all $\omega\in\Omega$. Here $\mathcal{F}_t(\omega)$ is the partition cell of $\mathcal{F}_t$ which contains $\omega$. 
\item the map $ (q_t,u_t,s_t):\Omega\ra \Delta(S_t)\times \R^{X_t}\times S_t$ has non-constant SEU $(q_t(\omega),u_t(\omega))$ for all $\omega$ and is adapted to the filtration $\mathcal{F}_t,t\leq T$. Moreover, whenever $\omega'\not\in\mathcal{F}_t(\omega)$ it holds $(q_t(\omega),u_t(\omega),s_t(\omega))\neq (q_t(\omega'),u_t(\omega'),s_t(\omega'))$.

\end{enumerate}


The \textbf{tiebreakers $(\hat q_t,\hat u_t)_{0\leq t\leq T}$} are \emph{regular} and \emph{preference-based}, i.e. 
\begin{enumerate}
\item $\mu(\omega\in\Omega: |M(A_t,\hat q_t,\hat u_t)|=1) = 1$ for all $A_t\in\mathcal{A}_t$. 
\item conditional on $\mathcal{F}_T(\omega)$ the sequence $(\hat q_1,\hat u_1),\dots,(\hat q_T,\hat u_T)$ is independent and 
\item $\mu((\hat q_t,\hat u_t)\in \cdot|\mathcal{F}_T(\omega)) = \mu((\hat q_t,\hat u_t)\in \cdot|q_l(\omega),u_l(\omega), l\leq t)$ for all $t$. 
\end{enumerate}

Simplicity and regularity are necessary for a parsimonious representation, whereas the \emph{preference-based} condition incorporated in C. ensures that the tie-breaking of the agent depends only on her realized SEU in the period at hand (and through it also on past history) but not on the realization of the objective state in the current period. 



We define for a triple $(f_k,A_k,s_k)$ the set
\[
C(f_k,A_k,s_k) = \{\omega\in \Omega: f_k\in M\left(M(A_k,q_k(\omega),u_k(\omega)),\hat q_k(\omega),\hat u_k(\omega)\right), s_k(\omega) = s_k\}.
\]
These are the states of the world which rationalize the observable $(f_k,A_k,s_k)$ in period $k$. Similarly one defines for a history $h^t=(A_0,f_0,s_0;\dots;A_t,f_t,s_t)$ the set of states of the world which rationalize the occurrence of the history. 

\[
C(h^t) = \cap_{l\leq t}C(A_l,f_l,s_l). 
\]

\subsubsection{The general representation.}

We are now ready to write down the most general representation of a dynamic aSCF. It doesn't impose any functional restrictions on the Bernoulli utilities of the agents and only a minimal restriction on the evolution of beliefs.

\begin{definition}\label{thm:filtrdefdrseu}
A Dynamic Random SEU representation (DR-SEU) of the dynamic aSCF $\rho$ is a tuple $\left(\Omega,\mathcal{F}^*,\mu,(\mathcal{F}_t,(q_t,u_t),s_t,(\hat q_t,\hat u_t))_{0\leq t\leq T}\right)$ such that 
\begin{enumerate}
\item $(\Omega,\mathcal{F}^*,\mu)$ is a finitely additive probability space, 
\item the filtration $(\mathcal{F}_t)\subset \mathcal{F}^*$ and the $\mathcal{F}_t-$adapted process $(q_t,u_t,s_t):\Omega\ra \Delta(S_t)\times\R^{X_t}\times S_t $ is simple, 
\item the $\mathcal{F}^*$-measurable tiebreaking process $(\hat q_t,\hat u_t):\Omega\ra\R^{X_t}$ is regular and preference-based and for all $f_t\in A_t, h^{t-1}\in \h_{t-1}(A_t)$,

\[
\rho_t(f_t,A_t,s_t|h^{t-1}) = \mu(C(f_t,A_t,s_t)|C(h^{t-1})). 
\]

\item Either 

D.1. The representation has correct interim beliefs (CIB):\\ $\mu(s_t\in\cdot|q_t) = q_t(\cdot)$ for all $t\in\{0,\dots,T\}$, 

or otherwise

D.2. The representation has no unforeseen contingencies (NUC):\\ $supp\left(\mu(s_t\in\cdot|q_t,u_t)\right) \subset supp\left(q_t(\cdot)\right)$.
\end{enumerate}

\end{definition}
Some explanations are in order. History $h^{t-1}$ happens with the probability $\mu(C(h^{t-1}))$: the state of the world has to be so that for each $l\leq t$ the realized subjective state/SEU $(q_l,u_l)$ picks $f_l$ from $A_l$, $f_l$ survives any possible tie-breaking and finally, in period $l$ the objective state $s_l$ is realized. 

Conditional on $C(h^{t-1})$ occurring, $f_t$ is chosen from $A_t$ only if the realized subjective state in period $t$ given by the pair $(q_t,u_t)$ is so that a SEU-maximizing choice from $A_t$ is $f_t$ and $f_t$ survives any possible tie-breaking. 

Note that the stochastic process of the objective and subjective states is unconstrained, except for the condition D: the agent uses the correct data-generating process conditional on her private information \emph{(correct interim beliefs)} or otherwise she respects the requirement of \emph{(no unforeseen contingencies)}, i.e. the agent never gets hard evidence that her belief process is misspecified. The only other requirement embodied in the definition is that the agent uses Bayes rule to update her beliefs.

\subsubsection{Two special cases: Evolving SEU vs. Gradual Learning.}

As noted before, the general representation doesn't include any behavioral restrictions on the evolution of the beliefs and tastes of the agent besides the SEU assumptions and that the agent remains Bayesian after every history with respect to her beliefs about the future evolution of tastes and objective states. In particular, her beliefs about the future SEU realizations may be incorrect. In this subsection we exclude this possibility.

\paragraph{Evolving SEU.} This specialization of DR-SEU captures a \emph{dynamically sophisticated} agent who \emph{correctly} takes into account the evolution of her future SEU preferences.\footnote{This model of sophisticated behavior still doesn't encompass all possible sophisticated behaviors allowed by the general DR-SEU representation -- see Example 3 concerning \cite{epstein} and \cite{epsteinetal} in subsection \ref{sec:axioms}.} There is an $\mathcal{F}_t-$adapted process of random EU-functionals $v_t,t=0,\dots,T$, the felicity functions, over instantaneous consumption lotteries $l\in \Delta(Z)$ and a discount factor $\delta>0$ such that $u_T = v_T$ and $u_t$ for $t\leq T$ is given by the following Bellman equation.

\begin{equation}\label{eq:evseu-filtr}
u_t(f_t(s_t)) = v_t(f_t^Z(s_t)) +\delta \E_{A_{t+1}\sim f_t^A(s_t),q_{t+1}\cdot u_{t+1}}\left[\max_{f_{t+1}\in A_{t+1}}(q_{t+1}\cdot u_{t+1})(f_{t+1})\middle|\mathcal{F}_t\right].
\end{equation}
Here the conditional expectation $\E[\cdot|\mathcal{F}_t]$ takes into account the randomness coming from the lottery $f_t^A(s_t)$ of the continuation problem as well as from the uncertainty about the SEU of the agent in period $t+1$. The agent makes the correct inference about the future SEU $q_{t+1}\cdot u_{t+1}$, given her current information in $\mathcal{F}_t$. 

\begin{definition}\label{thm:filtrdefevseu}
An Evolving SEU representation of the dynamic aSCF $\rho$ is a tuple \\$(\Omega,\mathcal{F}^*,\mu,(\mathcal{F}_t,q_t,u_t,s_t)_{0\leq t\leq T})$ such that 
\begin{enumerate}
\item $(\Omega,\mathcal{F}^*,\mu,(\mathcal{F}_t,q_t,u_t,s_t)_{0\leq t\leq T})$ is a DR-SEU representation.
\item \eqref{eq:evseu-filtr} holds true for the stochastic process of Bernoulli utilities $u_t, t=0,\dots,T$. 
\end{enumerate}
\end{definition}

If we assume there is only one period ($T=0$) then Evolving SEU collapses to the static model of section \ref{sec:rseu}. The same holds trivially true for the following special case of Evolving SEU. 

\paragraph{Gradual Learning.} This is a specialization of the Evolving SEU representation which captures an agent who is learning about her taste. This results in a martingale condition on the evolution of the felicities $v_t, t=0,\dots,T$. 

\begin{definition}\label{thm:filtrdefgradlearning}
A Gradual Learning (GL-SEU) representation of the dynamic augmented stochastic choice rule $\rho$ is a tuple $(\Omega,\mathcal{F}^*,\mu,(\mathcal{F}_t,q_t,u_t,s_t)_{0\leq t\leq T})$ such that 
\begin{enumerate}
\item $(\Omega,\mathcal{F}^*,\mu,(\mathcal{F}_t,q_t,u_t,s_t)_{0\leq t\leq T})$ is a Evolving-SEU representation.
\item There exists an EU-function $v$ for lotteries in $\Delta(Z)$ such that for all $t=0,\dots,T$ it holds
\begin{equation}
\label{eq:gl}
v_t = \E[v|\mathcal{F}_t]. 
\end{equation}
\end{enumerate}
\end{definition}
 As we show in the following subsection dynamic stochastic choice data are enough to distinguish the two special cases Evolving SEU and Gradual Learning even though the two models coincide in the static setting.\footnote{\cite{fis} showed the same insight in a setting of lotteries and without objective payoff-relevant states.}  

\subsection{Axiomatic Characterizations}\label{sec:axioms}
 The first axiomatization concerns the most general representation. 

\subsubsection{Axioms for DR-SEU}

Axioms for the general representation in Definition \ref{thm:filtrdefdrseu} can be classified in two groups. The first group identifies two types of \emph{observationally equivalent} histories. The second group comprises requiring Axiom 0 from the static setting after each history together with a technical axiom of \emph{history continuity}. 

\paragraph{Overcoming limited observability.}

Similar to \cite{fis} we characterize histories which are equivalent with respect to the information they reveal through two axioms: Contraction History Independence and Linear History Independence. This allows to overcome the limited observability problem.

Given a history $h^{t-1} = (A_0,f_0,s_0;\dots, A_{t-1},f_{t-1},s_{t-1})$ let $(h^{t-1}_{-k},(A'_k,f'_k,s'_k))$ be the history of the form 
$(A_0,f_0,s_0;\dots;A'_k,f'_k,s'_k;\dots;A_{t-1},f_{t-1},s_{t-1})$. That is, the history is changed only in period $k$.  

\begin{definition}\label{thm:CHI}
We say that $g^{t-1}\in\h^{k-1}$ is contraction equivalent to $h^{t-1}$ if for some $k$ we have $g^{t-1} = (h_{-k}^{t-1},(B_k,f_k,s_k))$ where $A_k\subset B_k$ and $\rho_k(f_k,A_k,s_k|h^{k-1}) = \rho_k(f_k,B_k,s_k|h^{k-1})$. 
\end{definition}

That is, when expanding the set of opportunities at a period $k$ but otherwise holding the history $h^{t-1}$ intact, the same stochastic choice results in the period of the expansion.

\paragraph{Axiom 1: Contraction History Independence (CHI)} For all $t\leq T,$ if $g^{t-1}\in\h_{t-1}(A_t)$ is contraction equivalent to $h^{t-1}\in\h_{t-1}(A_t)$ then for all $s_t\in S_t$
\[
\rho_t(\cdot,A_t,s_t|h^{t-1}) = \rho_t(\cdot,A_t,s_t|g^{t-1}).
\]

Intuitively, if the distribution of the preferences is stable, two contraction equivalent histories  should give the same stochastic choice in the future as well, all else equal. This is because in the Definition \ref{thm:CHI} above, elements from $B_k\setminus A_k$ were not attractive to any SEU in the underlying distribution of preferences which has induced either of the histories $h^{t-1}$ and $g^{t-1}$, and given the stability of the underlying distribution of preferences the content of private information revealed from the two histories $h^{t-1}$ and $g^{t-1}$ is the same. This implies that the continuation stochastic choice should be the same.


The other class of equivalent histories is the following. 

\begin{definition}\label{thm:LHI}
A finite set of histories $G^{t-1}\subset \h^{t-1}$ is linearly equivalent to $h^{t-1} = (A_0,f_0,s_0;\dots, A_{t-1},f_{t-1},s_{t-1})$ if 

\[
G^{t-1} = \{(h_{-k}^{t-1},(\lambda A_k+(1-\lambda)B_k,\lambda f_k +(1-\lambda)g_k,s_k)):g_k\in B_k\}.
\]
\end{definition}

That is, a history is changed only at a single period by having the revealed choice $f_k$ from $A_k$ mixed with all possible choices $g_k$ from a menu $B_k$.

One can calculate from the history-dependent aSCF, the probability choices \emph{conditional on a set of histories} $G^{t-1}$ by the formula 

\[
\rho(f_t,A_t,s_t|G^{t-1}) = \sum_{g^{t-1}\in G^{t-1}}\rho_t(f_t,A_t,s_t|g^{t-1})\cdot \frac{\rho(g^{t-1})}{\sum_{h^{t-1}\in G^{t-1}}\rho(h^{t-1})}.
\]

\paragraph{Axiom 2: Linear History Independence (LHI)} For all $t\leq T$ if $G^{t-1}\subset \h_{t-1}(A_t)$ is linearly equivalent to $h^{t-1}\in \h_{t-1}(A_t)$, then $\rho_t(f_t,A_t,s_t|h^{t-1}) = \rho_t(f_t,A_t,s_t|G^{t-1})$.
\vspace{4mm}\\
Intuitively, if we have a set of histories $G^{t-1}$ linearly equivalent to history $h^{t-1}$ with the mixing happening in period $k$, because of SEU-properties, $f_k$ is optimal from $A_k$ if and only if a mixture of the type $\lambda f_k+(1-\lambda)g_k$ with some $g_k$ is optimal from the mixed menu $\lambda \{f_k\}+(1-\lambda)B_k$. Therefore, the mixing doesn't reveal anything new regarding the private information of the agent and so continuation stochastic choice should be the same.

Now let Axioms 1 and 2 hold for the observable and assume the menu $A_t$ is not possible with positive probability after history $h^{t-1}$. Define
\[
\rho^{h^{t-1}}(f_t,A_t,s_t): = \rho_t(f_t,A_t,s_t|\lambda h^{t-1}+(1-\lambda)d^{t-1}),
\]
for some history $d^{t-1} = (g_k,\{g_k\},s_k)_{0\le k\le t-1}$ which leads to menu $A_t$ with probability one. LHI ensures that the construction is well-defined and coincides with $\rho_t(f_t,A_t,s_t|h^{t-1})$ whenever $A_t\in \A_t(h^{t-1})$. Note here that histories of the type $d^{t-1}$ don't reveal anything about the private information of the agent. They should be interpreted as tools for the analyst to obtain variation in the data, much needed for identification of the underlying parameters.

\paragraph{History-Dependent R-SEU and History Continuity.}

We model agents who in every period are SEU but have private information about their preferences. Therefore, the data need to satisfy Axiom 0 from the static setting. This is the content of the next Axiom. 

\paragraph{Axiom 3: R-SEU in every period} For all $t\leq T$  and $h^{t-1}$, each of the history-dependent aSCFs $\rho_t(\cdot|h^{t-1})$ satisfies Axiom 0 from the static setting, i.e. it has a R-SEU representation.

The last axiom needed to characterize DR-SEU is a technical form of Continuity. The following definition gives our concept of continuity for histories and is adapted from \cite{fis}.

\begin{definition}\label{thm:defmixconv}
1) For a sequence of acts $f_n$ say that $f_n$ converges in mixture to the act $f$, written as $f_n\ra^{m} f$, if there exists $h\in \F$ and $\alpha_n\ra 0$ with $f_n = \alpha_nh+ (1-\alpha_n)f$.

2) For a sequence of menus $(B^n)_n\subset \A$ say that $B_n$ converges in mixture to the act $f$, written $B^n\ra^m f$, if there exists $B\in \A$ and $\alpha_n$ with $B^n = \alpha_nB + (1-\alpha_n)\{f\}$. 

3) For a sequence of menus $(A^n)_n\subset\A$ say that $A_n$ converges in mixture to the menu $A$, written $A^n\ra^m A$, if for each $f\in A$ there is a sequence $(B_f^n)_n\subset \mathcal{A}$ such that $B_f^n\ra^m\{f\}$ and $A^n = \cup_{f\in A}B_f^n$.
\end{definition}

We next define menus and histories without ties, a concept we also come across later. 

\begin{definition}\label{thm:menuwoties} For any $0\leq t\leq T$ and $h^{t-1}\in\h_{t-1}$ the set of period $t-$menus without ties conditional on $h^{t-1}$ is denoted by $\A^*_t(h^{t-1})$ and consists of all $A_t\in\A_t$ such that for any $f_t\in A_t$ and any sequences $f_t^n\ra^mf_t,s_t\in S_t$ and $B_t^n\ra^mA_t\setminus\{f_t\}$ we have 
\[
\lim_n\rho_t(f_t^n,B_t^n\cup\{f_t^n\},s_t|h^{t-1}) = \rho_t(f_t,A_t,s_t).
\]
For $t=0$ we write $\A_0^*:=\A_0^*(h^{t-1})$. The set of period $t$ histories without ties is $\h_{t}^*:=\{h^t = (A_0,f_0,s_0;\dots;A_t,f_t,s_t)\in \h_{t}: A_k\in\A_t^*(h^{k-1}),\text{ for all }k\leq t\}$. 
\end{definition}

Intuitively, a menu $A_t$ without ties is so that no matter the SEU of the agent, she never needs to perform tie-breaking. Therefore the menu can be perturbed in \emph{any} direction and the probabilities of observing the perturbed act $f_t^n$ chosen from the perturbed menu $B_t^n$ converge to the probability of observing $f_t$ chosen from $A_t$. A history without ties is so that every menu occurring in it is without ties.


The technical Continuity axiom reads then as follows. 

\paragraph{Axiom 4: History Continuity} For all $t\leq T, A_{t},f_{t}$ and $h^{t-1}\in\h_{t-1}$,
\[
\rho_{t}(f_{t},A_{t},s_{t}|h^{t-1})\in co\{\lim_{n}\rho_{t+1}(f_{t+1},A_{t-1},s_{t}|h^{t-1,n}):h^{t,n}\ra^m h^t,h^{t-1,n}\in \h_{t-1}^*\}.
\]

Whenever a history $h^{t-1}$ is perturbed slightly, the change is in choices and decision problems as the objective states $s_k,k\le t-1$ come from a finite set. If the perturbation comes from menus without ties so that the agent doesn't need to perform tie-breaking along the path of the history, the probabilities of observing $f_t$ chosen from $A_t$ as well as $s_t$ realized should change continuously with the history.

\begin{theorem}\label{thm:drseufiltrthm}
For a dynamic aSCF $\rho$ the Axioms 1-4 are equivalent to the existence of a DR-SEU representation.  
\end{theorem}

If we add Existence of a Best Act and C-Determinism* from subsection \ref{subsec:infoR} to Axiom 0, we get a characterization of the special case of DR-SEU representation where the agent knows her Bernoulli utility $u_t$ for certain in every period. That is, she is learning only about the objective states. 

\begin{proposition}
[Informational Representation for aSCFs]\label{thm:infordynamic}

Assume that a dynamic aSCF $\rho$ has a DR-SEU representation with regular measure $\mu$.  
Assume that there exists a constant best prize. 

Then the following are equivalent after every history $h^t$ observed with positive probability. 
\begin{enumerate}
\item For all $(q_t,u_t),(p_t,v_t)\in RSSupp(\bar\rho_t(\cdot|h^t))$ $u$ is a positive affine transformation of $v_t$.
\item $\rho_t(\cdot|h^t)$ satisfies C-Determinism*.
\end{enumerate}

\end{proposition}

\subsubsection{Evolving SEU}


\paragraph{History-dependent revealed preference.} 

Stochastic choice coupled with the SEU assumption imposes enough structure on data to allow the identification of a \emph{history-dependent} preference relation $\better_{h_t}$ on acts. Intuitively, if the `tail' of the history $h^t$ is $(f_t,A_t,s_t)$, the SEU draw $(q_t,u_t)$ in period $t$ has to rationalize the choice of $f_t$ from $A_t$. For every pair of acts $g_t,r_t$ we can then define $g_t\better_{h^t} r_t$ if $g_t$ is weakly better than $r_t$ for every possible draw of SEU from $N(f_t,A_t)$ that happens with positive probability under the respective DR-SEU representation. Note that this implies that $\better_{h^t}$ is potentially incomplete. The following definition adds tie-breaking considerations to the intuition we just explained.

\begin{definition}\label{thm:histpref}
For each $t\leq T-1$ and $h^t=(h^{t-1},A_t,f_t,s_t)\in\h_{t}$ we define the relation $\better_{h^t}$ on $\F_t$ as follows: 
For any $g_t,g'_t\in \F_t$  we have $g_t\better_{h^t} r_t$ if there exist sequences in $\F_t$ with $g_t^n\ra^m g_t$ and $r_t^n\ra^m r_t$ such that 
\[
\rho_t\left(\frac{1}{2}f_t+\frac{1}{2}r_t^n, \frac{1}{2}A_t+\frac{1}{2}\{g_t^n,r_t^n\},s_t\middle|h^{t-1}\right) = 0,\text{ for all }n.
\]
Finally, let $\sim_{h^t}, \sbetter_{h^t}$ be the indifference and strict part of $\better_{h^t}$.
\end{definition}

Because of Axiom 0 in DR-SEU, specifically the \emph{no unforeseen contingencies (NUC)} assumption, the preference $\better_{h^t}$ doesn't depend on the realization of the period$-t$ objective state $s_t$ as long as that state has positive probability under $h^{t-1}$. 

We now put the additional axioms characterizing Evolving SEU on $\better_{h^t}$. 

\paragraph{Axiom 4: Separability.}

For all $t\leq T-1, g^t, r^t\in \F_t$ we have $g_t\sim_{h^t} r_t$ whenever $g_t^A(s_t) =^d r_t^A(s_t)$ and $g_t^Z(s_t) =^d r_t^Z(s_t)$ for all $s_t\in S_t$.
\vspace{3mm}\\
This says that whenever the marginal distributions over the current prize lottery and continuation menu of two acts after a history $h^t$ are the same then the two acts are indifferent under the revealed preference after the history. It ensures that Bernoulli utility $u_t$ has the form 

\begin{equation}\label{eq:helpevseu}
u_t(z_t,A_{t+1}) = v_t(z_t)+\delta V_t(A_{t+1}).
\end{equation}

Axiom 4 allows the definition of a history-dependent menu preference over continuation menus.

\begin{definition}\label{thm:filtrmenupref}
Fix a $z_t\in Z$. Take a $h^t\in \h_t$ and define an ex-post menu preference $\better_{h^t}$ over $\A_{t+1}$ by
\[
A_{t+1}\better_{h^t}B_{t+1},\text{ if  }\delta_{(z_t,A_{t+1})}\better_{h^t}\delta_{(z_t,B_{t+1})}.
\]
\end{definition}
\vspace{2mm}
We now add other menu preference axioms to shape the menu preference $V$ from \eqref{eq:helpevseu} into the form needed for \eqref{eq:evseu-filtr}. The next three Axioms are standard.

\paragraph{Axiom 5: Monotonicity.} Whenever $A_{t+1}\subseteq B_{t+1}$ it holds $B_{t+1}\better_{h^t}A_{t+1}$.
\paragraph{Axiom 6: Indifference to Timing.} For any $A_{t+1},B_{t+1}$ and $\alpha\in (0,1)$ we have 
\[
\alpha A_{t+1}+(1-\alpha)B_{t+1}\sim_{h^t}\alpha A_{t+1} + (1-\alpha)B_{t+1}.
\]
\paragraph{Axiom 7: Menu Non-Degeneracy.} There exists $A_{t+1},B_{t+1}$ such that $\delta_{(z_t,B_{t+1})}\better_{h^t}\delta_{(z_t,A_{t+1})}$ for all $z_t$.  
\vspace{3mm}\\
Before stating the next axiom, we introduce an operation on menus which produces for every menu a constant menu containing all the lotteries in its acts.
Formally, in a setting with AA-acts from $\F$ for a menu $A\subset \F$ define the menu of constant acts from $\bar A$ as follows.
\[
\bar A = \{g\in \F: g\text{ constant act with } g(s) = f(s')\text{ for some }f\in A, s,s'\in S\}. 
\]

The following axiom ensures that the menu preference $\better_{h^t}$ of Definition \ref{thm:filtrmenupref} can be represented by Expected Utility preferences with stochastic but state-independent Bernoulli utilities. 

\paragraph{Axiom 8: Weak Dominance.} For any $A_{t+1}\in \A_{t+1}$ it holds $\bar A_{t+1}\better_{h^t}A_{t+1}$.
\vspace{3mm}\\
Intuitively, from the perspective of the end of period $t$ and compared to the menu $A_{t+1}$, the menu $\bar A_{t+1}$ offers insurance w.r.t. the stochasticity of both beliefs and tastes as ex-post in $t+1$ the agent can choose her best lottery from any act in $A_{t+1}$ whereas in $A_{t+1}$ which lottery the agent ultimately faces depends on the realization of the objective state $s_{t+1}$.  

\paragraph{Menu Finiteness (technical).}

Next we define what it means for a menu preference to be finite. This is a technical property we need for tractability.

\begin{definition}\label{thm:menufinite}
For $\better$ a menu preference over some set of prizes $X$ say that it satisfies \emph{Finiteness} if there exists $K\in \N$ such that for menu $A$ there exists $B\subset A$ with $|B|\leq K$ and so that $B\sim A$.
\end{definition}

\paragraph{Axiom 9: Finiteness of Menu preference} For all $h^t\in \h_t$, the menu preference on $\A_{t+1}$ derived from $\better_{h^t}$ satisfies \emph{Finiteness} as in Definition \ref{thm:menufinite}. 
\vspace{4mm}\\
Finally, we add the sophistication axiom which ensures that the agent correctly predicts her future beliefs and tastes. Intuitively, if enlarging the menu $A_{t+1}$ to $B_{t+1}$ is valuable for the agent just after the realization of history $h^t$ and her beliefs about the future evolution of her preferences are correct, this is because there are possible draws of SEUs in period $t+1$ for which elements in $B_{t+1}\setminus A_{t+1}$ are optimal. This should be then reflected in the $h^t-$dependent stochastic choice from $B_{t+1}$. 


\paragraph{Axiom 10: Sophistication} 

For all $t\leq T-1$, $h^t\in\h_{t}$ and $A_{t+1}\subset B_{t+1}\in \A_{t+1}^*(h^t)$, the following are equivalent
\begin{enumerate}
\item $\rho_{t+1}(f_{t+1},B_{t+1},s_{t+1}|h^t)>0$ for some $f_{t+1}\in B_{t+1}\setminus A_{t+1}$ and some $s_{t+1}\in S_{t+1}$.
\item $B_{t+1}\sbetter_{h^t}A_{t+1}$. 
\end{enumerate}



\begin{theorem}\label{thm:evseufiltrthm}
For a dynamic aSCF $\rho$ satisfying a DR-SEU representation the Axioms 4-10 are equivalent to the existence of an Evolving SEU representation. 
\end{theorem}

Next, we note down a special cases of the Evolving SEU representation which can be used to model data from a population of agents with deterministic but heterogeneous tastes who are learning about payoff-relevant objective states. Thus, uncertainty about taste resolves in the first period, i.e. after an agent from the population is `drawn', but there is persistent uncertainty about payoff-relevant objective states. 

\paragraph{Example 3: Stochastic taste only in period zero.}
If we replace Axiom 8 with the following Strong Dominance axiom\footnote{This is what \cite{dlst} call Dominance in their main theorem.} then we get a version of Evolving SEU, where tastes are stochastic only in $t=0$ and the profile of future tastes is completely determined after every period-0 history.

\paragraph{Axiom 8': Strong Dominance}For all $0\leq t\leq T-1$ and $h^t\in \h_t$ we have: \\
If $f_{t+1}\in A_{t+1}$ and $\{f_{t+1}(s_{t+1})\}\better_{h^t}\{g_{t+1}(s_{t+1})\}$ for all $s_{t+1}\in S_{t+1}$ then \\$A_{t+1}\sim_{h^t} A_{t+1}\cup\{g_{t+1}\}$. 

Intuitively, if the Bernoulli utility is deterministic and if an act is better than another uniformly across all states, adding the dominated act to a menu which contains the dominating act doesn't make the menu more valuable. 


\begin{proposition}\label{thm:evseudomfiltrthm}
For a dynamic aSCF $\rho$ satisfying a DR-SEU representation the Axioms 4-7,8',9 and 10 are equivalent to the existence of an Evolving SEU representation where stochasticity of tastes is resolved at the end of period $0$.  
\end{proposition}

Finally, we note a special case of DR-SEU involving a sophisticated agent but which doesn't have an Evolving SEU representation.

\paragraph{Example 4.} \cite{epstein} and \cite{epsteinetal} consider a sophisticated agent who experiences temptation in beliefs and therefore updates her beliefs about objective states in a subjective way not necessarily conforming to Bayesian updating with respect to the true data-generating process. The ex-post choice versions of these models are special cases of DR-SEU and satisfy C-Determinism*, but they violate Axiom 5 (Monotonicity), which is necessary for an Evolving SEU representation.\footnote{The model in \cite{epsteinetal} features infinite horizon so the statement above holds for its finite horizon version.} 

\subsubsection{Gradual Learning}

Gradual Learning imposes additional restrictions on the evolution of Bernoulli utilities of an Evolving SEU representation: the agent is learning about a fixed taste. 

To explain the three additional Axioms which lead to the Gradual Learning representation we introduce some notation.

For some $t\leq T-1$ and given a sequence $l_t,\dots,l_T\in \Delta(Z)$ of consumption lotteries, let the stream of lotteries $(l_t,\dots,l_T)\in \Delta(X_t)\subset \F_t$ be the period-$t$ lottery that at every period $\tau\geq t$ yields consumption according to $l_{\tau}$. Formally, for any consumption lottery $l\in\Delta(Z)$ and menu of constant acts $A_{t+1}\in \A^c_{t+1}$ define $(l,A_{t+1})\in \Delta(X_{t+1})$ to be the lottery which has stochastic consumption now and fixed continuation with probability one.\footnote{This is similar to the definition in section 4.3 of \cite{fis}. } Then $(l_t,\dots,l_T) = (l_t,A_{t+1})\in \Delta(X_t)$ is defined recursively from period $T$ backwards by $A_T = \{l_T\}\in\A_T$ and $A_s = \{(l_s,A_{s+1})\}\in\A_s$ for all $s=t+1,\dots,T-1$. We write $(l_t,\dots,l_{\tau},m,\dots,m)$ if $l_{t+1} = \dots =l_T$ for some $m\in \Delta(Z)$ and $\tau\geq t$. 

\paragraph{Axiom 11: Stationary Preference over Lotteries [FIS].} For all $t\leq T-1, l,m,n\in \Delta(Z)$ and $h^t$ we have 
\[
(l,n,\dots,n)\sbetter_{h^t}(m,n,\dots,n)\text{ if and only if }(n,l,\dots,n)\sbetter_{h^t}(n,m,n,\dots,n). 
\]
Intuitively, if and only if the felicity $v_t$ today is just the average of the future felicity $v_{t+1}$ tomorrow, it holds true from today's perspective that postponing the \emph{choice} between two lotteries by a period results in the same ranking as for the case that the choice is made immediately. 

For the second axiom, just as in \cite{fis} for lotteries $l,m\in\Delta(Z)$, we say they are \emph{$h^t$-non-indifferent} if $(l,n\dots,n)\not\sim_{h^t}(m,n,\dots,n)$ for some $n\in \Delta(Z)$. 

Moreover, to avoid tautologies we require a non-degeneracy condition. 

\paragraph{Condition 1: Consumption Non-degeneracy}\label{cond} For all $t\leq T-1$ and $h^t$, there exists $h^t-$non-indifferent $l,m\in \Delta(Z)$. 

\paragraph{Axiom 12: Constant Intertemporal Trade-off [FIS].} For all $t,\tau\leq T-1,$ if $l,m$ are $h^t-$non-indifferent and $\hat l,\hat m$ are $g^{\tau}$-non-indifferent, then for all $\alpha\in [0,1]$ and $n\in \Delta(Z)$:

\begin{align*}
(l,m,n,\dots,n)\sim_{h^t}&(\alpha l+(1-\alpha)m,\alpha l + (1-\alpha)m,n,\dots,n)\\
\equivalent\\
(\hat l,\hat m,n,\dots,n)\sim_{g^{\tau}}&(\alpha \hat l+(1-\alpha)\hat m,\alpha \hat l + (1-\alpha)\hat m,n,\dots,n). 
\end{align*}
This ensures that the discounting factor $\delta$ from the Evolving SEU representation is unique. 

Finally, we note down the classical axiom which gives $\delta<1$.
\paragraph{Axiom 13: Impatience [FIS].} For all $t\leq T-1, h^t$ and $l,m,n\in \Delta(Z)$,\\if $(l,n,\dots,n)\sbetter_{h^t}(m,n,\dots,n)$, then $(l,m,n,\dots,n)\sbetter_{h^t}(m,l,n,\dots,n)$.
\vspace{4mm}\\The characterization result for Gradual Learning is then as follows.

\begin{theorem}\label{thm:glfiltrthm}
Assume the aSCF $\rho$ satisfies an Evolving SEU model and assume \hyperref[cond]{Condition 1} is satisfied. Then Axioms 11-13 are equivalent to the existence of a Gradual Learning representation for $\rho$. 
\end{theorem}

\subsection{Uniqueness}

The following Proposition proved in Section 4 of the online appendix shows that all three representations are unique up to positive affine transformations of the Bernoulli utilities the agent uses to evaluate lotteries over the respective consequence spaces $X_t$ as well as up to relabeling of the states of the world $\omega$ and of the objective states $s_t$. The characterization of uniqueness is a prerequisite for the comparative static exercises of Section \ref{sec:comp}. The results mirror closely the identification in \cite{fis} adapted to our more general setting with agents who hold (possibly incorrect) beliefs about payoff-relevant states.

\begin{proposition}\label{thm:uniquenessfiltr}
1) Suppose that a dynamic aSCF $\rho$ admits two DR-SEU representations $\left(\Omega,\mathcal{F}^*,\mu,(\mathcal{F}_t,(q_t,u_t),s_t,(\hat q_t,\hat u_t))_{0\leq t\leq T}\right)$ and $\left(\Omega',\mathcal{F}'^*,\mu',(\mathcal{F}'_t,(q'_t,u'_t),s'_t,(\hat q'_t,\hat u'_t))_{0\leq t\leq T}\right)$.

Then there exists a bijection $\phi_t:\mathcal{F}_t\ra\mathcal{F}'_t$ and $\mathcal{F}_t$-measurable functions $\alpha_t:\Omega\ra\R_{++}$ and $\beta_t:\Omega\ra\R$ such that for all $\omega\in\Omega$:

\begin{enumerate}[(i)]
\item $\mu(\mathcal{F}_0(\omega)) = \mu'(\phi_0(\mathcal{F}_0(\omega)))$ and $\mu(\mathcal{F}_t(\omega)|\mathcal{F}_{t-1}(\omega)) = \mu'(\phi_t(\mathcal{F}_t(\omega))|\phi_t(\mathcal{F}_{t-1}(\omega)))$ if $t\geq 1$;
\item $q'_t\equiv q_t$ for all $t\geq 1$, $u_t(\omega) = \alpha_t(\omega)u'_t(\omega')+\beta_t(\omega)$ whenever $\omega'\in \phi_t(\mathcal{F}_t(\omega))$;
\item $\mu((\hat q_t,\hat u_t)\in B_t(\omega)|\mathcal{F}_t(\omega)) = \mu'((\hat q'_t,\hat u'_t)\in \phi_t(B_t(\omega))|\mathcal{F'}_t(\phi_t(\omega)))$ for any $B_t(\omega) = \{(p_t,v_t)\in \Delta(S_t)\times \R^{X_t}: f_t\in M(M(A_t,(q_t(\omega),q_t(\omega)),p_t,v_t))\}$ for some $f_t\in A_t,$ $A_t\in\A_t$. 
\end{enumerate}

2) If $\rho$ admits two Evolving-SEU representations then in addition to (i)-(iii) above we have 
\begin{enumerate}[(i)]
  \setcounter{enumi}{3}
\item $\alpha_t(\omega) = \alpha_0(\omega)\left(\frac{\hat\delta}{\delta}\right)^t,$ for all $\omega\in \Omega$ and $t\geq 0$;
\item $v_t(\omega) = \alpha_t(\omega)v'_t(\omega')+\gamma_t(\omega)$ whenever $\omega'\in\phi_t(\mathcal{F}_t(\omega))$, where $\gamma_T(\omega) = \beta_T(\omega)$ and $\gamma_t(\omega) = \beta_t(\omega) -\delta\E[\beta_{t+1}|\mathcal{F}_t(\omega)]$ if $t\leq T-1$.
\end{enumerate}

3) If $\rho$ has two Gradual Learning Representations and satisfies \hyperref[cond]{Condition 1}, then in addition to (i)-(v) the following holds

\begin{enumerate}[(i)]
  \setcounter{enumi}{5}
  \item $\delta = \delta'$
  \item $\beta_t(\omega) = \frac{1-\delta^{T-t+1}}{1-\delta}\E[\beta_T|\mathcal{F}_t(\omega)]$.
  \end{enumerate}

\end{proposition}

1) shows that agent's choices uniquely identify the evolution of her private information in both relevant dimensions: tastes and beliefs. The lack of identification for the Bernoulli utility functions $u_t$ is unavoidable. Intuitively, when one rescales the Bernoulli utilities by a factor which depends only on information up to time $t$, the sets of maximal elements $M(A_t;q_t,u_t)$ don't change. 

2) shows that the Evolving SEU model allows for stronger identification of the Bernoulli utilities. The scaling factor of Bernoulli utilities needs to be measurable with respect to the information available at $t=0$. This is because in the Evolving SEU model the utility of the continuation problem enters \emph{cardinally} into the overall utility of choosing an act from a menu. One can then use the same information, namely that available in period $t=0$, to build a measuring rod with which utilities can be compared across periods. Obviously, the scaling factor $\alpha_t$ still depends on the state of the world $\omega$. In a population interpretation of the observable aSCF this means that different agents may use different information available at $t=0$ to compare utils intertemporally.

3) shows that the Gradual Learning model improves on the identification properties of the Evolving SEU model because the discount factor is identified uniquely. This is a consequence of the Constant Intertemporal Trade-Off Axiom. Under that Axiom any possible scaling of the Bernoulli utilities in addition to depending on time $t=0$ information only, has to be constant over time.

\section{Comparative Statics Results}\label{sec:comp}

This section offers simple comparative statics results under varying assumptions about the representations of the observable aSCF. The characterizations are simple because aSCFs represent very rich data sources.

\subsection{A measure of belief biasedness}

If the analyst doesn't observe anything about the realization of objective states, it is impossible to discuss correctness of beliefs of the agents. Most of the canonical models of behavior based only on menu choice as an observable, as in \cite{dlst} and \cite{ks} and many others, as well as models of stochastic choice without observable objective states as in \cite{lu} cannot address questions of belief biasedness. In this part we illustrate what is possible if the observable of the analyst consists of aSCFs.

For simplicity we assume there are best and worst prizes which coincide for all agents considered: that is, constant acts $\underline{f}, \bar f$ such that for every aSCF $\rho$ considered it holds:
\[
\text{ for every } f\neq \underline f\text{ we have }\bar\rho(f,\{\underline{f},f\}) = 1\text{ and for every }f\neq \bar f\text{ we have }\bar\rho(f,\{\bar f,f\}) = 0.
\]
Moreover, for simplicity we assume the agents have \emph{the same} non-stochastic taste $u$ and focus on comparative statics related to beliefs.\footnote{Formally speaking all aSCF/SCF-s in this subsection satisfy C-determinism* -- choice is stochastic because beliefs of an agent are stochastic, besides possible randomness coming from tie-breaking. 
In this setting all the machinery of \cite{lu}, esp. the related test acts can be used (see online appendix). The conditions on the SCFs which imply that the taste of distinct agents are the same are available upon request.}

We assume there is an underlying state of the world $\omega$ coming from a finite set $\Omega$. For example in Example 2 $\omega$ may encode gender or ethnicity. An analyst observes two agents $i=1,2$ who are interested in the realization of an objective payoff-relevant state $s\in S$. A state of the world $\omega$ goes hand in hand with a set of beliefs about the possible realizations of $s$ for each agent and a true data-generating-process (DGP). The analyst observes the aSCFs of the agents which are assumed to have the following form.

\begin{equation}\label{eq:helpcorr}
\rho_i(f,A,s) = \sum_{\omega\in\Omega}\mu(\omega,s)\tau^i_{q_i(\omega),u}(f,A),\quad i=1,2.
\end{equation}
Here $\mu\in \Delta(\omega\times S)$ and the tie-breakers $\tau^i_{q_i(\omega),u}$ depend only on the realized SEU $(q_i(\omega),u)$ of agent $i$.  

We assume $\mu$ is either known by the analyst (e.g. an experiment in a lab) or the analyst gleans it from the data $\rho_i$ using Theorem \ref{thm:SREU}.  



Now assume the analyst fixes a direction $q(\omega)\in\Delta(S)$ for possible biases for every $\omega\in\Omega$ and is interested in finding out how biased, if at all, the beliefs of the agents are in the direction $\{q(\omega)\}_{\omega\in\Omega}$. The analyst might think that a possible bias for $\omega$ corresponds to some `extreme' $q(\omega)\neq \mu(\cdot|\omega)$.\footnote{For example, if $\Omega$ encodes gender and the true DGP is that $\mu(\cdot|\omega)$ is independent of $\omega$, a possible extreme bias might be to assume that for $\omega=male$, $q(\omega)$ is `tilted' towards more favorable realizations of the objective state $s$ whereas for $\omega = female$, $q(\omega)$ is `tilted' towards more unfavorable realizations of the objective state $s$. As Example 1 illustrates, this might be the case with employment data depending on the vocation and job properties.}

A natural way in terms of the aSCF to say that an agent is biased in the direction $\{q(\omega)\}_{\omega\in\Omega}$ and that, say, agent 1 has \emph{uniformly less biased beliefs} than agent 2 is to require the following in terms of the representation.



\begin{definition}\label{thm:morecorrdef}
1) Agent i's beliefs are biased toward the direction $q:=\{q(\omega)\}_{\omega\in\Omega}$ if and only if 
there exists a vector of weights $\{a(\omega)\}_{\omega\in \Omega}\in [0,1]^{\Omega}$ such that the following holds 

$$q_i(\omega) = a_i(\omega)q(\omega)+ (1-a_i(\omega))\mu(\cdot|\omega)\text{ for some } a(\omega)\in [0,1].$$

2) Agent 1's beliefs are uniformly less biased toward  $q$ than agent 2's beliefs if and only if 
it holds for every $\omega\in \Omega$ that $0\leq a_1(\omega)\leq a_2(\omega)\leq 1$.
\end{definition}

Figure \ref{fig:morecorr} helps describe the definition.

\begin{figure}[H]
\centering
\includegraphics[width=8cm]{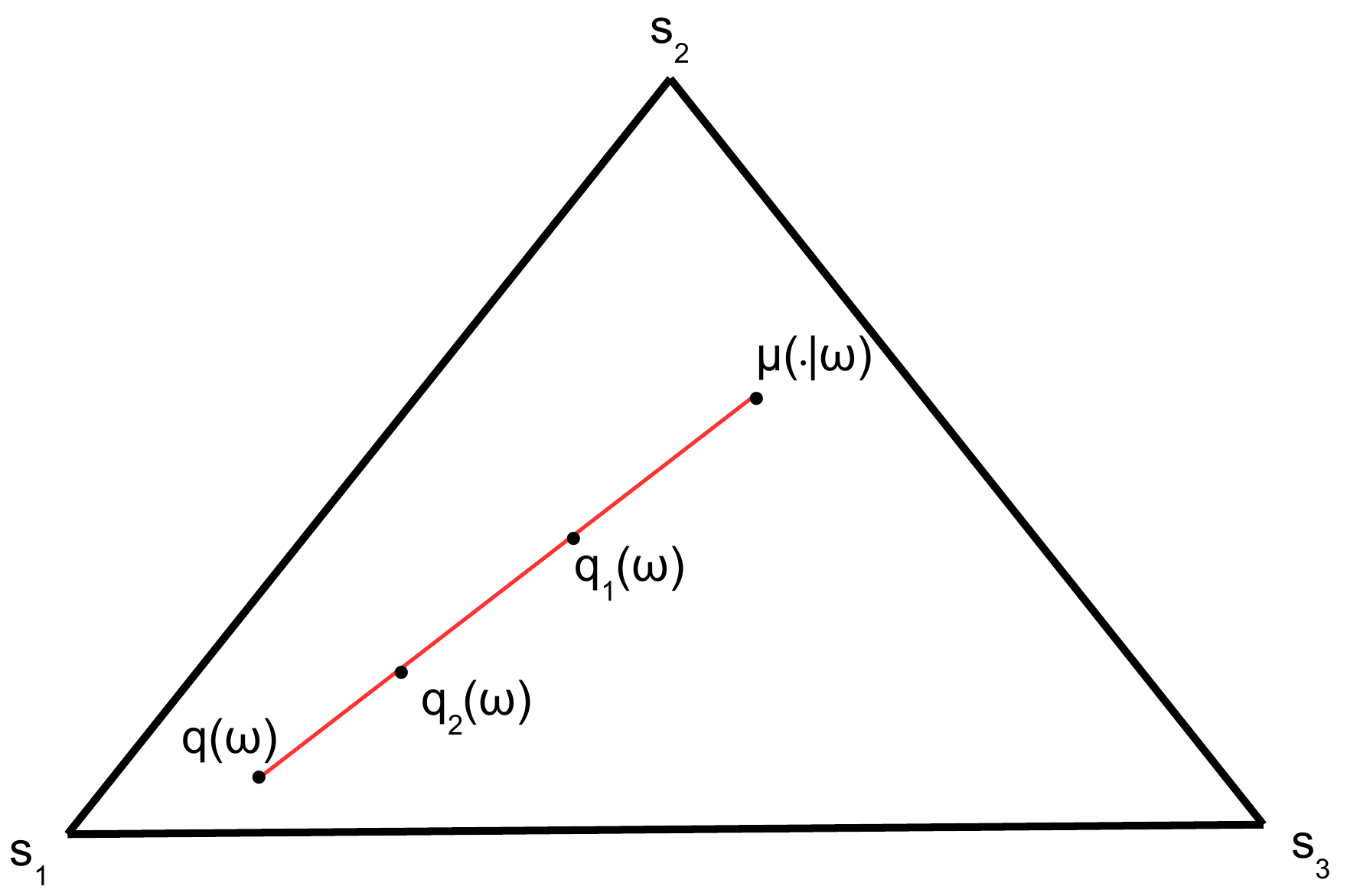}
\caption{In state $\omega$ agent 1 has beliefs more aligned to true DGP than agent 2.}
\label{fig:morecorr}
\end{figure}

The associated menu preference approach from \cite{lu} provides a way to identify the weights of the bias in some direction $q$. 

\begin{definition}[\cite{lu}]\label{thm:inducedmenupref}
Given $\bar\rho$, let the associated menu preference $\better_{\bar\rho}$ be given by the utility function on menus $V_{\bar\rho}:\A\ra[0,1]$ with 
\[
V_{\bar\rho}(A) = \int_0^1 \bar\rho(A,A\cup\{\alpha\underline{f}+(1-\alpha)\bar f\})da.
\]
\end{definition}

For a fixed weight in $\alpha\in[0,1]$ the value $\bar\rho(A,A\cup\{\alpha\underline{f}+(1-\alpha)\bar f\})$ gives the probability that an element of $A$ beats the act $\alpha\underline{f}+(1-\alpha)\bar f$, that is, the probability that the agent prefers items out of the menu $A$ instead of the \emph{test act} with weight $\alpha$ on the worst prize. Intuitively speaking, a menu is more valuable in the associated menu preference of a SCF if \emph{in the aggregate} its elements are more preferred than test acts $\alpha\underline{f}+(1-\alpha)\bar f$. \cite{lu} shows that, up to tie-breaking considerations, every stochastic choice function as $\bar\rho$ can be characterized through its associated menu preference $V_{\bar\rho}$. Thus, except for tie-breaking, $\bar\rho$ contains no more information about the agent than $V_{\bar\rho}$ does.

Given the direction of bias $q$ define for every weight of biases $a:\Omega\ra [0,1]$ the associated menu preference where the agent gives weight $a(\omega)$ to the belief $q(\omega)$ whenever the state of the world $\Omega$ is realized.

\[
V_a(A) = \int_{\Omega} \max_{f\in A}\left[a(\omega)q(\omega)+(1-a(\omega))\mu(\cdot|\omega)\right]\cdot(u\circ f)\mu(d\omega). 
\]

This gives a map $\psi_q:[0,1]^{\Omega}\ra\{\text{menu preferences}\}$.\footnote{The image of this map can naturally be identified with value functions of menu preferences.} Intuitively, one can interpret any element $a\in [0,1]^{\Omega}$ as a vector of degrees of biasedness towards $q$. 

Note that the construction of the map $\psi_q$ comes directly from the data: the aSCF-s $\rho_i,i=1,2$ give $\mu(\cdot|\omega)$ (or the analyst knows this already) and the analyst picks the bias vector $q$. 
Once can show that once a bias direction $q$ is fixed, every weight vector $a$ defines a \emph{unique} menu preference $V_a$.

This allows the following characterization of the degree of belief-biasedness in direction $q$ in terms of observables/data. Here, recall that the induced menu preference from the stochastic choice function $\bar\rho$ is also completely constructed from stochastic choice data. 

\begin{proposition}\label{thm:morecorrchar}
Assume that the two aSCF $\rho_i,i=1,2$ are as in \eqref{eq:helpcorr} and consider a vector of biases $q\in \Delta(S)^{\Omega}$. It holds:

\begin{enumerate}
\item Agent $i$'s beliefs are uniformly biased toward the direction $q$ with degree $a\in [0,1]^{\Omega}$ if and only if $$\psi^{-1}_q(V_{\bar\rho_i}) = a,$$ i.e. if and only if $a$ is the image under $\psi_q$ of the menu preference induced from stochastic choice. 
\item Agent 1's beliefs are uniformly less biased toward the direction $q$ than agent 2's beliefs if and only if $$\psi^{-1}_q(V_{\bar\rho_1})\leq \psi^{-1}_q(V_{\bar\rho_2}).$$
\end{enumerate}

\end{proposition}




Note that by varying $q$, an analyst can use the induced menu preference of $\bar\rho_i$ (from Definition \ref{thm:inducedmenupref}) to identify the \emph{actual} bias \emph{direction} of an agent whenever her aSCF doesn't satisfy the Axiom of Correct Interim Beliefs from Definition \ref{thm:filtras-r-seu}. 

\paragraph{Example 1 continued.} 
In the context of Example 1 from the Introduction, subsection \ref{sec:disc} this Proposition states that stochastic choice data are enough for the analyst to identify the incorrect beliefs $\hat q_i, i=1,2$. Namely, assume directions for the biases $q(s_0') = (1,0)$ and $q(s_0'') = (0,1)$. These correspond to the `extreme' beliefs that a candidate with $s_0=s_0'$ will always deliver outcome $s_1=g$ and a candidate $s_0=s_0''$ will always deliver outcome $s_1=b$. The Proposition delivers then $a(s_0') = 2\hat q_1-1$ and $a(s_0'') = 1-2\hat q_2$ so that whenever $a:S_0\ra[0,1]$ is identified from data the analyst can recover the incorrect beliefs $\hat q_i, i=1,2$.
\vspace{3mm}\\
An alternative to the vector of weights $a\in [0,1]^{\Omega}$ on biases is to require instead a uniform weight $a\in [0,1]$ on biases which is independent of the realization of the characteristic $\omega$. The conditions on the induced menu preferences identifying the bias $a$ are then simpler than in Proposition \ref{thm:morecorrdef}.\footnote{Defining the menu preference of an \emph{unbiased} agent (a counterfactual) and of a \emph{fully biased} agent, the condition of biasedness is that the induced menu preference of the agent is a convex combination of the menu preferences of the unbiased and fully biased agent and that $a$ corresponds to the weight on the biased agent.}  Nevertheless, in applications, the bias weights will usually differ according to the realization of the characteristic $\omega$. For example, one might expect in some cases the agent to use the correct conditional DGP $\mu(\cdot|\omega)$ and in other cases of realized $\omega$-s to use a very biased belief much closer to an `extreme' $q(\omega)\neq\mu(\cdot|\omega)$. Therefore, here we have focused on the concept of Definition \ref{thm:morecorrdef} which allows for this additional flexibility.

\subsection{The speed of learning about taste}

In this subsection we consider agents in a dynamic setting ($T\ge 1$) whose stochastic choice data satisfy the Gradual Learning model and discuss measures across agents of the speed of learning about taste. We assume for all agents considered in this subsection that at time $t=0$ their taste is not deterministic. 
Formally, we require the following conditions on any aSCF of this section. 

\paragraph{Assumptions} For all aSCFs in this subsection it holds true: 
\begin{enumerate}
\item $\rho$ satisfies a Gradual Learning (GL) representation with $T\geq 1$ and sequence of felicities $v_t,t\in\{0,\dots,T\}$.  
\item $\bar\rho_0$ doesn't satisfy C-Determinism*.
\end{enumerate}
\vspace{4mm}
B. ascertains that there is non-trivial learning about taste for an agent. 
On the other hand, due to Sophistication (assumed as part of A.), if an agent learns her future taste at the end of a period $t$, her taste remains deterministic in all future periods. 

Recall that the preferences $\better_{h^t}$ on continuation menus $\A_{t+1}$ for some history $h^t\in\h_t$ from Definition \ref{thm:filtrmenupref} are derived solely from stochastic choice data. If for an agent her uncertainty about future taste is resolved after a history $h^t$ the derived menu preference on $\A_{t+1}$ derived from $\better_{h^t}$ will satisfy Strong Dominance. On the other hand, Strong Dominance will be violated for $\better_{h^t}$ whenever an agent's uncertainty about future taste doesn't get resolved after history $h^t$. 
The same holds if instead of looking at whether Strong Dominance is satisfied we look at whether C-Determinism* is satisfied.

This suggests a simple way to define the \emph{speed of learning about taste} of an agent who satisfies a Gradual Learning model as well as an equally simple way to rank such agents according to their speed of learning about taste.

\begin{definition}
\label{thm:speedoflearningdef}
1) Say that an agent learns her future taste after history $h^t$ if her derived menu preference on $\A_{t+1}$ from $\better_{h^t}$ satisfies Strong Dominance or equivalently, if $\rho_{t+1}(\cdot|h^t)$ satisfies C-Determinism*.\footnote{Equivalence holds under the assumption that the data satisfy the GL representation.} 
\\2) Say that an agent becomes certain of her future taste at time $t$ if she learns her future taste after every history $h^t\in\h_t$. 
\\3) Say that agent 1 learns her taste faster than agent 2 if the following implication holds true for every $t\le T-1$:
\[
\text{agent 2 becomes certain of her taste at }t\quad\imply\quad \text{agent 1 becomes certain of her taste at }t.
\]
\end{definition}

The characterization of these concepts in terms of the GL representation (Definition \ref{thm:filtrdefgradlearning}) is as follows.

\begin{proposition}
\label{thm:speedoflearningchar}
1) Suppose an agent has a GL representation with probability space $(\Omega, \mathcal{F}^*,\mu)$. Then an agent learns her future taste after history $h^t$ if and only if conditional on $C(h^t)$ her felicity is deterministic, i.e.
\[
v_{t+1}\text{ is a constant function on }C(h^t).
\]
2) Suppose an agent has a GL representation with underlying probability space $(\Omega, \mathcal{F}^*,\mu)$. An agent becomes certain of her future taste at time $t$ if and only if her felicity at time $t$ is independent of the state of the world $\omega$, i.e.
\[
v_{t+1}\text{ is a constant function on all of }\Omega.
\]
3) Suppose two agents $i=1,2$ have GL representations with underlying probability space $(\Omega, \mathcal{F}^*,\mu)$ but otherwise may have different filtrations $\{\mathcal{F}^i_t\}_{t\leq T}$ and different evolution of SEUs $\{(q^i_t,u^i_t)\}_{t\leq T}$ for $i=1,2$. Then agent 1 learns her taste faster than agent 2 if and only if the following implication holds true for every $t\leq T-1$:
\[
v^2_{t+1}\text{ is a constant function on all of }\Omega\quad\imply\quad v^1_{t+1}\text{ is a constant function on all of }\Omega.
\]
\end{proposition}

\paragraph{Example 4.} 

Assume that we have two investors $i=1,2$ facing \emph{the same} market conditions whose CARA Bernoulli utility over monetary outcomes has the form $x\mapsto 1-e^{-\gamma_i x}$ where $\gamma_i$ is random according to a discrete distribution taking positive values from a finite set $\Gamma\subset [1,+\infty)$. In every period each investor decides whether to invest in a risky project $f$, whose outcome is strongly dependent on market conditions (objective state $s_t\in \R_+$ drawn anew each period) through $f(s_t)\sim \sqrt{s_t}\times Uniform\{-1,1\} + s_t$ or to pick investments $h(\alpha)$ whose $s_t$-independent outcome satisfies $h(\alpha)\sim \sqrt{\alpha} \times Uniform\{-1,1\}+\alpha$. Then according to the above Proposition an analyst has two ways of telling who of the two investors has learned her parameter $\gamma_i$ the earliest. If she only has data on choices from menus containing only acts of the type $h(\alpha)$ she finds the first time when the choice of each investor on such menus becomes deterministic. If she only has data of choice among menus, an indicator that investor 1 learned her preference parameter earlier is that she starts preferring menus where $f$ is present to menus where $f$ isn't present earlier in time than investor 2 does. 


\section{Conclusion}\label{sec:conc}

We have introduced a dynamic stochastic choice model general enough to encompass situations where a subjective expected utility agent has both stochastic taste as well as stochastic beliefs about the realization of objective payoff-relevant states. Under the assumption that the analyst has access to data which reveal the agent's history-dependent choices as well as the sequence of realizations of objective states we have characterized axiomatically the case when the analyst can uncover the otherwise arbitrary evolution of the private information of the agent. 

The assumed richness of the data allows the analyst to test whether the agent is using correctly specified beliefs about objective states conditional on her private information and if not, to determine the bias of the agent as well as to compare different agents according to their biasedness of beliefs. We have also characterized special cases of the general representation, Evolving SEU and Gradual Learning, which would have been otherwise indistinguishable in the static setting. Finally, in the case of Gradual Learning, we have shown how an analyst is able to detect from data that the agent has stopped learning about her taste and that therefore the randomness in choice only comes from randomness in beliefs. 

Information acquisition is outside the scope of this model and constitutes the natural next step in research. 
E.g. we shouldn't expect the student in Example 2 not to try and actively learn early about her final job market outcome. So it natural to expect Indifference to Timing to be violated; if an agent tries to actively learn about future tastes by spending resources after history $h^t$ we should expect her to satisfy instead the weaker condition: 
\[
\text{if }A_{t+1}\sim_{h_{t+1}}B_{t+1}\quad\text{ then }\quad\alpha A_{t+1}+(1-\alpha)B_{t+1}\worse_{h^t}A_{t+1}.
\]
That is, since contingent planning costs utility, the agent is averse to it whenever she is ex-ante indifferent between two decision problems. Introducing information acquisition in this framework would also allow a better study of misspecified learning. 

Other directions to pursue are as follows. We haven't considered consumption dependence as \cite{fis} do in their DREU model of stochastic taste only.\footnote{This is an easy extension left to the interested reader.} Developing `systems' of DR-SEUs coming from agents in strategic situations is also left for future research, as is characterizing meaningful relaxations of the Sophistication assumption in the Evolving SEU model. 

Finally, on another perspective, this paper is about identification and not inference. In applications data sets are naturally finite. We leave for future research characterizations of stochastic dynamic behavior when data sets are finite.


\newpage

\begin{appendices}

The Appendix is organized as follows.  Appendix \ref{sec:RSEU} is devoted to the proof of Theorem \ref{thm:SREU}. Appendix \ref{sec:asdynamic} describes the Ahn-Sarver representations in the dynamic setting. These are more convenient for proofs and their equivalence to the Filtration-based representations from the main text of the paper is proved in the online appendix. Appendix \ref{sec:sephist} proves the existence of so-called \emph{separating histories}. These are an essential tool in the proof of the main characterization theorems. Most of Appendix \ref{sec:proofmain} is devoted to the proof of Theorem \ref{thm:drseufiltrthm}, the rest of it to the proofs of Section \ref{sec:comp}. The proof of the rest of the characterization theorems is in the online appendix. Besides the rest of the auxiliary results, the latter also contains most of the technical work needed to extend the menu choice literature to the setting of SEUs, add explicit tie-breaking to \cite{lu} and beliefs about objective payoff-relevant states to \cite{as}.

\section{Random Subjective Expected Utility with observable objective states (AS-version)}\label{sec:RSEU}

\subsection{Separation property for acts - static setting}

We prove a \emph{separation property} for menus of acts, similar to Lemma 1 in \cite{as} (separation property for lotteries). 

We start with a trivial remark which will be used extensively in the following.

\begin{remark}
1) A SEU preference encoded by $(q,u)$ is constant (i.e. consists of only indifferences) if and only if $u$ is constant. 

2) Two SEU representations $(q,u)$ and $(q',u')$ represent the same SEU preference if and only if $q=q'$ and $u\approx u'$. 
\end{remark}
The separation property for acts is as follows.

\begin{lemma}[Separation property in the AA setting]\label{thm:separationproperty}
Let $Z'$ be any set (possibly infinite) and let $\{(q_k,u_k):k=1,\dots,K\}\subset \Delta(S)\times \R^{Z'} $ be a set of \emph{pairwise distinct} SEU representations s.t. $u_k$ is non-constant for all $k=1,\dots, K$. Then there is a collection of acts $\{f_k:k=1,\dots,K\}\subset \F$ s.t. $q_k\cdot u_k(f_k)>q_k\cdot u_k(f_l)$ for any distinct $l,k\in \{1,\dots,K\}$. 
\end{lemma}

\begin{proof}
We divide the proof in three steps.

\textbf{Step 1.} Assume first that $u_k\not\approx u_l$ for all $l\neq k$ and that $Z'$ is finite. Then we are in the setting of Lemma 13 from \cite{fis} and can use a menu of constant acts to realize the separation property required.

\textbf{Step 2.} Assume now that $u_k\approx u_l$ for all $l\neq k$ and that $Z'$ is finite. W.l.o.g. we can assume that $u_k=u_l=u$ and that $im(u) = [0,1]$. Note that in this case it also holds $q_k\neq q_l$ for all $l\neq k$. It is enough in this case to solve the following problem: 
\[
\text{(P)\quad For all }k\text{  find  }p_k\in\Delta(S)\text{ s.t. }q_k\cdot p_k>q_k\cdot p_l,\quad l\neq k.
\]
Now we are again in the setting of Lemma 13 in \cite{fis}, if we take as Bernoulli utilities the $q_k$-s. Formally, it follows $q_k\not\approx q_l$ whenever $S$ has more than one element as one can check using uniqueness result in the classical vNM Theorem. Thus, Lemma Lemma 13 in \cite{fis} gives probability distributions $p_k, k=1,\dots K$ satisfying (P). Now, we can easily construct the acts needed by the formula $u(f_k(s)) = p_k(s), s\in S, k=1,\dots,K$. 
Note that this trick works because $\Delta(S)\subset [0,1]^S$. 

\textbf{Step 3.} Assume now that we are in the general case $(q_k,u_k)\not\approx (q_l,u_l),l\neq k$.  There exists a finite $Z\subset Z'$ s.t. all $u_k$ are non-constant in $\Delta(Z)$. We are going to choose acts $f:S\ra\Delta(Z)$. Assume w.l.o.g. that for all $k$ we have $im(u_k)\subseteq [0,1]$. Divide the Bernoulli utilities $u_k$ in classes $r=1,\dots R\leq K$ s.t. if $l,k$ are so that $u_k\approx u_l$ they belong to the same class. Within the same class, normalize the Bernoulli utilities to be equal. Thus, we can rewrite the SEU preferences given as 
\[
\{(q_{rl},u_r):r=1,\dots R, l=1\dots, K_r\}.
\]
Now pick constant acts $h_r,r=1\dots, R$ as in \textbf{Step 1} with $u_r(h_r)>u_r(h_{r'}),r\neq r'$. Pick also within each group $r\in \{1,\dots,R\}$ acts $f_{rl}, l=1,\dots, K_r$ with image in $\Delta(Z)$ s.t. $q_{rl}\cdot u_r(f_{rl})>q_{rl}\cdot u_r(f_{rl'}), l\neq l'$. We claim that the separating acts we are after can be taken of the form 

\[
\lambda f_{rl}+(1-\lambda)h_r,\quad r=1,\dots,R; l=1,\dots,K_r
\]

whenever $\lambda>0$ small enough. 

We need to show that there exists $\lambda\in (0,1)$ with
\[
(P1)\qquad  q_{rl}\cdot u_r(\lambda f_{rl}+(1-\lambda)h_r)>q_{rl}\cdot u_r(\lambda f_{r'l'}+(1-\lambda)h_r'),\text{ whenever }(r,l)\neq (r',l').
\]

Consider first the case $r=r'$. Then $l\neq l'$ and (P1) is true for all $\lambda$ by linearity of the Bernoulli functions and the choice of $f_{rl}$. 

Consider then the case $r\neq r'$. Given that $u_r({h_r})>u_r(h_{r'})$ and the linearity of the Bernoulli functions, for a fixed pair of tuples $(r,l)\neq (r',l')$ (P1) becomes true whenever $\lambda$ is small enough for that pair. This gives a positive upper bound on $\lambda$. Since the number of pairs $(r,l)$ is finite, overall there exists a $1>\lambda>0$ for which (P1) is satisfied for all distinct pairs $(r,l)\neq (r',l')$.

\end{proof}

\subsection{Proof for the Axiomatization of aSCFs (AS-version)}

Pick an element $y^*\in X$ and set $U =\{u\in \R^X: u(y^*) = 0\}$.

We first define the AS-version (Ahn-Sarver version) of the representation.

\begin{definition}\label{thm:AS-RSEUaSCFdef}
1) Let $\rho$ be an aSCF for acts in $\F$ over $\Delta(X)$ where $X$ is a separable metric space and $S$, the set of objective states is finite. 

We say that $\rho$ admits an \emph{AS-version R-SEU  representation} if there is a triple $$(SubS,\mu,\{((q,u),\tau_{q,u}):(q,u)\in SubS\})$$ such that 
\begin{enumerate}
\item $SubS$ is a finite subjective state space of distinct and non-constant SEUs and $\mu$ is a probability measure on $SubS\times S$. 
\item For each $(q,u)\in SubS$ the tie-breaking rule $\tau_{q,u}$ is a regular sigma-additive probability measure on $\Delta(S)\times U$ endowed with the respective product Borel sigma-Algebra. 
\item For all $f\in \F$, $A\in \A$ and $s\in S$ we have 

\begin{equation}\label{eq:asreu}
\rho(f,A,s) = \sum_{(q,u)\in SubS}\mu(q,u,s)\tau_{q,u}(f,A),
\end{equation}

where $\tau_{q,u}(f,A) := \tau_{q,u}\left(\{(p,w)\in \Delta(S)\times U: f\in M(M(A;u,q);w,p)\}\right)$.
\end{enumerate}

2) We say that the AS-version R-SEU representation has \emph{no unforeseen contingencies} if $supp(\mu(\cdot|q,u))\subseteq supp(q)$ for all $(q,u)\in SubS$. 

3) We say that the AS-version R-SEU representation has \emph{correct interim beliefs} if $\mu(\cdot|q,u) = q(\cdot)$ for all $(q,u)\in SubS$. 
\end{definition}

The next Theorem gives the axiomatization of aSCFs which have an AS-version R-SEU representation. 

\begin{theorem}\label{thm:aASSREU}
The aSCF $\rho$ on $\A$ admits an AS-version R-SEU representation
if and only if it satisfies 
\begin{enumerate}
\item Statewise Monotonicity
\item Statewise Linearity
\item Statewise Extremeness
\item Statewise Continuity
\item Statewise State Independence
\item Statewise Finiteness
\end{enumerate}
Moreover, it additionally has a \emph{No Unforeseen Contingencies} representation if and only if it additionally satisfies No Unforeseen Contingencies. Finally, it has a \emph{Correct Interim Beliefs} representation if and only if it additionally satisfies Correct Interim Beliefs.
\end{theorem}


\begin{proof}[Proof of Theorem \ref{thm:aASSREU}.]
\textbf{Necessity.} Checking this is routine. In particular, one checks easily that $RSSupp(\bar\rho) = supp(\mu)$. 
\vspace{2mm}\\
\textbf{Sufficiency.} We prove this in several steps.

\emph{Step 1.} We construct the SCFs $\bar\rho$ from $\rho$ as well as $\rho(\cdot,\cdot|s)$ for all $s\in S$. Due to the axioms on $\rho$ all of $\bar\rho$ as well as $\rho(\cdot,\cdot|s), s\in S$ satisfy all axioms from Theorem 1 in the online appendix. In particular, we have the following representations: for all $f\in A,A\in \A$

\begin{equation}\label{eq:helpascf1}
\bar\rho(f,A) = \sum_{(q,u)\in SubS}\psi(q,u)\tau_{q,u}(f,A)
\end{equation}

and 

\begin{equation}\label{eq:helpascf2}
\rho(f,A|s) = \sum_{(q,u)\in SubS(s)}\psi^s(q,u)\tau^s_{q,u}(f,A).
\end{equation}

with appropriate probability measures $\psi$ and $\psi^s$ on finite sets of SEUs.

\emph{Step 2.} Due to simple probability accounting it holds

\begin{equation}\label{eq:helpascf3}
\bar\rho(f,A) = \sum_{s\in S}\rho(f,A|s)\rho(s). 
\end{equation}
If it were true that $supp(\psi^s)\not\subseteq supp(\psi)$ for some $s\in S$ then by use of separating menus as constructed in Lemma \ref{thm:separationproperty} one could come to a contradiction to \eqref{eq:helpascf3}. The same kind of contradiction argument and use of Lemma \ref{thm:separationproperty} leads to exclusion of the case $supp(\psi)\setminus\cup_{s\in S}supp(\psi^s)\neq \emptyset$. In all we have established 
\[
supp(\psi)=\cup_{s\in S}supp(\psi^s).
\]
In particular, we can extend w.l.o.g. $\psi^s$ for all $s$ to all of $supp(\psi)$ by setting it to zero outside of $supp(\psi^s)$. 

\emph{Step 3.} By a similar mixing argument as in Proposition 2 in the online appendix (see step 3 there) one can easily show that whenever $(q,u)\in supp(\psi)\cap supp(\psi^s)$ we have $\tau_{q,u}^s=\tau_{q,u}$. In particular, we can write the representations for $\rho(\cdot,\cdot|s)$ as 
\begin{equation}\label{eq:helpascf4}
\rho(f,A|s) = \sum_{(q,u)\in SubS}\psi^s(q,u)\tau_{q,u}(f,A).
\end{equation}
By plugging \eqref{eq:helpascf4} in \eqref{eq:helpascf3}, rearranging and using the uniqueness result for the AS-representation of $\rho$ from Proposition 2 in the online appendix we get 

\begin{equation}\label{eq:helpascf5}
\psi(q,u) = \sum_{s\in S}\psi^s(q,u)\rho(s),\quad (q,u)\in supp(\psi). 
\end{equation}

By setting $\mu(q,u,s) = \psi^s(q,u)\rho(s)$ we define a probability measure over $SubS\times S$ whose marginal over $SubS$ is full support and which satisfies \eqref{eq:asreu}.

\emph{Step 4.} Take a separating menu $\bar A = \{f(q,u):(q,u)\in supp(\psi)\}$ for $supp(\psi)$. We show that the following property (P) gives us the representation for correct interim beliefs.

\[
(P)\qquad \rho(\cdot|f(q,u),\bar A) = q(\cdot),\quad (q,u)\in supp(\psi). 
\]

\textbf{Claim.} (P) implies the representation with correct interim beliefs. 

\begin{proof}[Proof of Claim.]
For the menu $\bar A$ and each $(q,u)\in supp(\psi)$ we have 

\begin{equation}\label{eq:helpmehere}
\psi^s(q,u) =\rho(f(q,u),\bar A|s) =  \frac{\rho(f(q,u),\bar A,s)}{\rho(s)} = \frac{\rho(s|f(q,u),\bar A)\bar\rho(f(q,u),\bar A)}{\rho(s)} = \frac{q(s)\psi(q,u)}{\rho(s)}. 
\end{equation}
Here, only in the last equality we have used (P) and the definition and representation of $\bar\rho$ from Theorem 1 in the online appendix. We write this as the identity $$(!)\quad \rho(s)\psi^s(q,u) = q(s)\psi(q,u).$$ Summing (!) w.r.t. $(q,u)$ we get the identity (!!) $\rho(s) = \sum_{(q,u)\in SubS}\psi(q,u)q(s)$ for all $s\in S$ and thus a unique solution for $\psi^s$ in \eqref{eq:helpmehere}. It is then trivial to see that the representation holds because of \eqref{eq:helpascf4} and (!!). 
\end{proof}

\emph{Step 5.} In this step we show that (P) is implied by Correct Interim Beliefs. 

Denote in general for each $q\in \Delta(S)$ such that $(q,u)\in supp(\mu)$ for some $u$ $\rho(\cdot|f(q,u),\bar A) = \hat q(q,u)(\cdot)$.

Suppose by contradiction that there exists some $(q,u)\in supp(\mu)$ with $\rho(\cdot|f(q,u),\bar A) \neq q(\cdot)$. If it holds for some $\hat u$ that $(\hat q(q,u),\hat u)\in supp(\mu) = RSSupp(\bar\rho)$ then we know that $(\hat q(q,u),\hat u)\not\in N(\bar A,f(q,u))\cap RSSupp(\bar\rho) = \{(q,u)\}$ as $\bar A$ is separating for $RSSupp(\bar\rho)$ and $\hat q(q,u)\neq q$, which implies $(q,u)\not\approx (\hat q(q,u),\hat u)$. But clearly $|N(\bar A,f(q,u))\cap RSSupp(\bar\rho)| = |\{(q,u)\}|=1$.

Overall it follows that Correlated Interim Belief axiom is violated at the choice data $(f(q,u),\bar A)$.

\emph{Step 6.} We show that the following property (P!) gives us the representation for no unforeseen contingencies.

\[
(P!)\qquad supp(\rho(\cdot|f(q,u),\bar A)) \subset supp(q(\cdot)),\quad (q,u)\in supp(\psi). 
\]

\textbf{Claim.} (P!) implies the representation with unforeseen contingencies.

We look at \eqref{eq:helpmehere}, but leave out the final equality. The Claim follows immediately.  

\emph{Step 7.} In this step we show that (P!) is implied by No Unforeseen Contingencies.

Suppose by contradiction that there exists some $(q,u)\in supp(\psi)=RSSupp(\bar\rho)$ with $supp(\rho(\cdot|f(q,u),\bar A))\not\subseteq supp(q(\cdot))$. Pick again a separating menu for $RSSupp(\bar\rho)$ and note that $|N(\bar A,f(q,u))\cap RSSupp(\bar\rho)| = |\{(q,u)\}|=1$. Overall it follows that the No Unforeseen Contingencies axiom is violated for the choice data $(f(q,u),\bar A)$.

\end{proof}

We note down uniqueness.\footnote{The online appendix shows equivalence between AS-based representations and Filtration-based representations.}

\begin{proposition}\label{thm:aSCFuniqueness}
The AS-version REU-representation for an aSCF $\rho$ is essentially unique in the sense that for each two representations the only degree of freedom is positive affine transformations of the Bernoulli utilities of elements in the support of the measures over SEUs. 
\end{proposition}

\begin{proof}
For the case of CIB this follows directly from Proposition 2 in the online appendix applied to the SCF corresponding to the aSCF. 

For the case of NUC, if there are two different representations for $\rho $ with respective measures $\mu,\mu'$ it follows from Proposition 2 in the online appendix that the marginals are equal: $\sum_{s}\mu(q,u,s) = \sum_{s}\mu'(q,u,s)$ for all $(q,u,s)$. In particular, up to equivalence classes of positive affine transformations of the Bernoulli utility functions the support of these two marginals in $\Delta(S)\times\R^{X}$ is equal for the two measures. Assume then w.l.o.g. the same normalization for both supports. Taking now a separating menu $\bar A$ for the SEUs in the support of the two measures $\mu,\mu'$, we have from the representation property that 
\[
\rho(f(q,u),\bar A,s) = \mu(q,u,s)=\mu'(q,u,s)\text{ for all }s.
\]
This concludes the proof. 
\end{proof}

\begin{proof}[Proof for Proposition \ref{thm:inforepresentationascf}]
\textbf{Sufficiency.}
Define the SCF on $\Delta(X)$ by the formula\footnote{Here a slight abuse of notation as we haven't written down the isomorphism between constant menus of acts and menus of lotteries, but the context gives clarity.} 

\[
\tau(f,A) = \rho(f,A),\quad A\text{ is menu of constant acts}.
\]
Note that Theorem 1 in the online appendix gives with some slight abuse of notation
\[
\tau(f,A) = \sum_{(q,u)\in SubS\text{ for some }q}\mu(q,u)\tau_{q,u}(\{(p,w)\in \Delta(S)\times U: f\in M(M(A;q,u);p,w)\}).
\]
Since the beliefs play no role in the decision of the agent (all acts are constant), one can rewrite this as 
\[
\tau(f,A) = \sum_{u\in\pi_u\left(SubS\right)}\mu(u)\tau'_{u}(\{w\in U: f\in M(M(A;u);w)\}),
\]
where $\mu(u) = \sum_{q: (q,u)\in SubS}\mu(q,u)>0$ and $\tau'_{u} = \sum_{q: (q,u)\in SubS}\frac{\mu(q,u)}{\mu(u)}\tau_{q,u}$. Note that $\tau'_{u}$ is a regular tie-breaker for lotteries.\footnote{Here, the $w$ breaking ties from $M(A,u)$ is drawn as follows: first draw a $(q,u)$ where $(q,u)$ has probability $\frac{\mu(q,u)}{\mu(u)}$ and then, draw (conditionally independently across the $(q,u)$-s) $w$ according to the marginal of $\tau_{q,u}$ on $U$. This works because the tie-breakers are \emph{preference-based}.}

Obviously this gives an S-based REU representation as in Theorem 4 of \cite{fis}.  
C-Determinism* implies then directly that $\tau$ has only one state in the sense of the S-based representation from \cite{fis}.\footnote{Otherwise one arrives easily at a contradiction through separating lotteries to either $\mu(u)>0$ for all $u$ or to the C-Determinism* Axiom.} In particular, $u\approx v$ for all $u,v\in U$ such that $(q,u),(p,v)\in supp(\mu)$ for some $q,p\in \Delta(S)$.

\textbf{Necessity.} Consider a menu of constant acts $A$. Then for all $(q,u), (p,u)\in supp(\mu)$ we have $M(A;u,q)= M(A;v,u) = :M(A,q)$, so that by a small abuse of notation which uses the fact that the menu $A$ is constant we can write 

\[
\rho(f,A) = \sum_{(q,u)\in SubS}\mu(q)\tau_{q}(\{w\in U: f\in M(M(A;u);w)\}).
\]
The existence of a best constant act $\bar f$ means $u(\bar f)>u(f)$ whenever $f\neq \bar f$ and $f$ also constant. 

Note now that for each $g\in A, g\neq f$ we have for either $u(af+(1-a)\bar f)>u(g)$ or $u(af+(1-a)\bar f)<u(g)$ for all $a<1$ near enough to $a$. It follows that 

\[
\rho(af+(1-a)\bar f;A\setminus\{f\}\cup\{af+(1-a)\bar f\}) \in \{0,1\},\quad\text{for all }a<1\text{ near enough to 1.}
\]
Thus C-Determinism* is satisfied.  
\end{proof}

We skip writing down a statement and proof of a Proposition connecting AS-version representations with the representations in Definition \ref{thm:filtras-r-seu} (filtration form) since it will be subsumed in the more general arguments in Section 4 of the online appendix.



\section{AS-Based Representations for the dynamic setting}\label{sec:asdynamic}

The proofs in this appendix are done in the AS-version of the representations. Here we explain what these are. The online appendix then establishes the equivalence between the two types of representations. 

\subsection{Dynamic Random Subjective Expected Utility (DR-SEU)}
\begin{definition}\label{thm:drseudef}
We say that a history-dependent family of aSCF $\rho = (\rho_0,\dots,\rho_T)$ has a DR-SEU representation if there exists
\begin{itemize}
\item a finite objective state space $S$ and a collection of partitions $S_t,t=1\dots,T$ of $S$ such that $S_t$ is a refinement of $S_{t-1}$.
\item a finite collection of states of the world $\Theta_t, t=0,\dots,T$ (an element is of the type $(q_t,u_t,s_t)\in\Delta(S_t)\times\R^{X_t}\times S_t$). The sequence $\Theta_t,t\leq T$ has a partitional structure and there are no repetitions: each element $(q_t,u_t,s_t)$ is indexed by the predecessors $(q_0,u_0,s_0;\dots;q_{t-1},u_{t-1},s_{t-1})$.\footnote{This means that there can be repetitions in terms of the SEUs $(q_t,u_t)$ but whenever this happens a different $s_t$ is realized.} Moreover we have the restriction that $s_{k}\in supp (q_{k})$. 

\item a collection of probability kernels $$\psi_k:\Theta_{k-1} \ra \Delta(\Theta_k)$$ for $k=0,\dots,T$ \footnote{With the obvious conventions for $k=0$.} with a typical element in the image written as $\psi_k^{q_{k-1},u_{k-1},s_{k-1}}$. In particular, the probability that $(q_k,u_k,s_k)$ is realized after $\theta_{k-1}$ occurs is $\psi_k^{\theta_{k-1}}(q_k,u_k,s_k)$.

\item a sequence of tie-breakers: for all $t=0,\dots,T$ a regular probability measure $\tau_{(q_t,u_t)}$ over $\Delta(S_t)\times\R^{X_t}$, for all $(q_t,u_t) = \pi_{qu}(\theta_t)$ for some $\theta_t\in\Theta_t$. 

\end{itemize}
 such that the following two conditions hold. 

\textbf{DR-SEU 1}
\begin{enumerate}[(a)]
\item every $(q_t,u_t)\in \pi_{qu}\left(supp(\psi_t^{\theta_{t-1}})\right)$ represents a non-constant SEU preference.
\item $supp(\psi_t^{\theta_{t-1}}) \cap supp(\psi_t^{\theta'_{t-1}})=\emptyset$ whenever $\theta_{t-1}\neq \theta'_{t-1}$, both in $\Theta_{t-1}$.\footnote{This implies, that whenever $\pi_s(\theta_{t-1})=\pi_s(\theta_{t-1})$ and two elements $\theta_t\in supp(\psi_t^{\theta'_{t-1}}),\theta_t'\in supp(\psi_t^{\theta'_{t-1}})$ with $\pi_{qs}(\theta_t) = \pi_{qs}(\theta'_t)$ we must have $u_t\neq u'_t$.}
\item $\cup_{\theta_{t-1}}supp (\psi_t^{\theta_{t-1}}) = \Theta_t$. 
\item either (correct interim beliefs) The kernels $\psi$ satisfy $\psi_k^{\theta_{k-1}}(s_k|q_k,u_k) = q_k(s_k)$ 
\vspace{3mm}\\
or otherwise (no unforeseen contingencies) $supp\left(\psi_k^{\theta_{k-1}}(\cdot|q_k,u_k)\right) \subset supp(q_k)$.
\end{enumerate}

\textbf{DR-SEU 2} 

The SCF $\rho_t$ after a history $h^{t-1} = (A_0,f_0,s_0;\dots,A_{t-1},f_{t-1},s_{t-1})$ is given by
\\
\resizebox{1.1 \textwidth}{!} {
$
\rho_t(s_t,f_t,A_t|h^{t-1}) = \frac{\sum_{\pi_s(\theta_0,\dots,\theta_t)=(s_0,\dots,s_t)}\left[\prod_{k=0}^{t-1}\psi_k^{\theta_{k-1}}(\theta_k)\tau_{\pi_{qu}(\theta_k)}(f_k,A_k)\right]\cdot \psi_t^{\theta_{t-1}}(\theta_t)\tau_{\pi_{qu}(\theta_t)}(f_t,A_t)}{\sum_{\pi_s(\theta_0,\dots,\theta_{t-1})=(s_0,\dots,s_{t-1})}\left[\prod_{k=0}^{t-1}\psi_k^{\theta_{k-1}}(\theta_k)\tau_{\pi_{qu}(\theta_k)}(f_k,A_k)\right]}.
$
}
\end{definition}

 

\subsection{Evolving Subjective Utility (Evolving SEU)}

The Evolving Subjective Expected Utility representation is a special case of DR-SEU.

In the pre-choice situation in period $t$ when the agent knows $(q_t,u_t) = \pi_{qu}(\theta_t)$ and satisfies the Evolving SEU representation she evaluates acts according to the following SEU functional

\begin{equation}\label{eq:evseu-as}
\E_{q_t}[u_t(f_t)] = \E_{s_t\sim q_t}[u_t(f_t(s_t))] = \E_{s_t\sim q_t}[v_t(f_t^Z(s_t))] + \delta V_{t}^{\pi_{qu}(\theta_t)}(f_t^A). 
\end{equation}
Here $V_{t}^{\pi_{qu}(\theta_t)}(f_t^A)$ is defined in two steps. First we define 
\begin{equation}
\label{eq:dlst}
V_t^{\theta_t}(A_{t+1})= \int\max_{f_{t+1}\in A_{t+1}}\E_{q_{t+1}}[u_{t+1}(f_{t+1})]d\psi_{t+1}^{\theta_t}(q_{t+1},u_{t+1}).
\end{equation}
This gives the value of a menu when the agent knows the menu, but not the SEU with which it will evaluate the acts. This is the situation just after $(z_t,A_{t+1})$ is known to the agent at the end of period $t$. 

A moment before, i.e. when the agent doesn't know $s_t$ yet the value of $f^A_{t}$ is given by 
\begin{equation}
\label{eq:extdlst}
V_t^{\pi_{qu}(\theta_t)}(f^A_{t}):= \sum_{s_t}\sum_{A_{t+1}\in suppf_t^A(s_t)}q_t(s_t)f_t^A(s_t)(A_{t+1})V_t^{\theta_t}(A_{t+1}) =: \sum_{s_t}q_t(s_t)V_t^{\theta_t}(f_t^A(s_t)).
\end{equation}

Note that the uncertainty that is integrated out in \eqref{eq:extdlst} is the \emph{objective} one concerning $s_t$ and that we have used equation \eqref{eq:dlst} to define the extension of $V_t^{\theta_t}$ to lotteries over menus.\footnote{I.e. agent is Expected Utility w.r.t. lotteries over menus.} 

We can rewrite this in integral form as follows.
\begin{equation*}
V_t^{\pi_{qu}(\theta_t)}(f^A_{t}) = \int\max_{f_{t+1}\in A_{t+1}}\E_{q_{t+1}}[u_{t+1}(f_{t+1})]d\psi_{t+1}^{\pi_{qu}(\theta_t)}(q_{t+1},u_{t+1}),
\end{equation*}

where $\psi_{t+1}^{\pi_{qu}(\theta_t)}(q_{t+1},u_{t+1}) := \sum_{s_t}q_t(s_t)\psi_{t+1}^{\theta_t}(q_{t+1},u_{t+1}) = \sum_{s_t,s_{t+1}}q_t(s_t)\psi_{t+1}^{\theta_t}(q_{t+1},u_{t+1},s_{t+1})$.




\subsection{Gradual SEU-Learning.}\label{sec:gradualas}

Gradual SEU-learning is the case of Evolving SEU with the additional requirement that her sequence of expected utility functionals from consumption $v_t, t=0,\dots T$ form a Martingale. 
In the following we use the projection $\pi_v$, which for a $u_t$ as in \eqref{eq:evseu-as} gives the corresponding $v_t$. 

Normalize $v_t(\bar p) = 0$ for all $t$ where $\bar p$ is the uniform lottery over $Z$. This is possible because $Z$ is assumed to be finite for the dynamic setting. After a $\theta_t = (q_t,u_t,s_t)$ it has to hold for the sequence $\pi_v({\theta_t})$ from the Evolving SEU representation

\begin{equation}\label{eq:glseu-as}
\pi_v({\theta_t}) = \frac{1}{\delta}\sum_{(q_{t+1},u_{t+1})\in \pi_{qu}(\Theta_{t+1}) }\psi_{t+1}^{\theta_t}(q_{t+1},u_{t+1})\cdot \pi_v(u_{t+1}) = \frac{1}{\delta}\E_{}[\pi_v(\theta_{t+1})|\theta_t].
\end{equation}

\section{Separating histories}\label{sec:sephist}

We first define histories \emph{consistent} with a state $\theta_t$. Then we define \emph{separating histories} for a fixed state $\theta_t$. The main result of this section establishes the existence of separating histories (Lemma \ref{thm:sephistexistence}). 

Let us assume that we have an aSCF $\rho$ which satisfies DR-SEU 1. We define the predecessor of a state $\theta$ as $pred(\theta_t) = (\theta_0,\dots,\theta_{t-1})$.

\begin{definition}\label{thm:defpred}
For a state $\theta_t = (q_t,u_t,s_t)$ denote by $pred(\theta_t) = (\theta_0,\dots,\theta_{t-1})$ the unique predecessor of $\theta_t$ from $\prod_{i=0}^{t-1}\Theta_i$.
\end{definition}
The concept is well-defined because of DR-SEU 1 (a)-(b). 


\begin{definition}\label{thm:defconsistent}
Given a history $h^t=(A_0,f_0,s_0;A_1,f_1,s_1;\dots; A_t,f_t,s_t)$ say that $\theta_t$ is consistent with $h^t$ if for the unique predecessor of $\theta_t$, given by $(\theta_0,\dots,\theta_{t-1})$ we have 

\[
\prod_{k=0}^t\tau_{\pi_{qu}(\theta_k)}(f_k,A_k)\cdot \psi_k^{\theta_{k-1}}(\theta_k)>0.
\]
Here we use the convention $\psi_0^{\theta_{-1}}:= \psi_0$. 
\end{definition}

Note that multiple states $\theta_t$ can be consistent with the same history $h^t$. 

Define 
\[
\resizebox{1.15\hsize}{!}{$QU_{\theta_k}(A_{k+1},f_{k+1},s_{k+1}) = \{(q_{k+1},u_{k+1}): (q_{k+1},u_{k+1}, s_{k+1})\in supp(\psi_{k+1}^{\theta_k})\text{ and }f_{k+1}\in M(A_{k+1};q_{k+1},u_{k+1})\}.$}
\]
This is the set of SEU-s $(q_{k+1},u_{k+1})$ occurring right after $\theta_k$ which can rationalize the data $(A_{k+1},f_{k+1},s_{k+1})$.

For time $t=0$ define 
\[
QU_0(A_0,f_0,s_0) = \{(q_{0},u_{0}): (q_{0},u_{0},s_0)\in supp(\psi_{-1}^{\theta_0})\text{ and }f_{0}\in M(A_{0};q_{0},u_{0})\}.
\]

We prove first the following Lemma.

\begin{lemma}[Pendant to Lemma 1 in \cite{fis}]\label{thm:consistentpred}
Fix any $\theta_t$ and its predecessor $(\theta_0,\dots,\theta_{t-1})$. Suppose $h^t = (B_0,g_0,s_0;\dots;B_t,g_t,s_t)$ satisfies $QU_{\theta_{k-1}}(B_k,g_k,s_k)=\{\pi_{qu}(\theta_k)\}$. Then for all $k=0,\dots,t$, only $\theta_k$ in $\Theta_k$  can be consistent with $h^k$. 
\end{lemma}

\begin{proof}
Fix any $l=0,\dots,t$ and consider $\theta'_l\in \Theta_l\setminus\{\theta_l\}$ with $pred(
\theta'_l) = (\theta'_0,\dots,\theta'_{l-1})$. 

Let $k\leq l$ be smallest such that $\theta'_k\neq \theta_k$. Then $\pi_{qu}(\theta'_k)\in \pi_{qu}(supp(\psi_{k}^{\theta_{k-1}}))$. So $QU_{\theta_{k-1}}(B_k,g_k,s_k)=\{(q_k,u_k)\}$ (which is assumed) implies either (A) $(q_k,u_k)\neq (q'_k,u'_k)$ or (B) $(q_k,u_k)= (q'_k,u'_k), s_k\neq s'_k$ (otherwise contradiction to $\theta_k\neq\theta'_k$).

In the case of (B) the definition of the QU-sets implies then that $s'_k\not\in supp(q_k)$, i.e. $q'_k(s'_k)=0$. In the case of (A) the definition of the QU-sets implies $g_k\not\in M(B_k;q'_k,u'_k)$, i.e. $\tau_{q'_k,u'_k}(g_k,B_k)=0$. Overall we have that $\theta'_l$ is not consistent with $h^l$.

Next we show that $\theta_l$ is consistent with $h^l$. Note that from the definition of histories w.r.t. to some aSCF it follows that $\rho(g_l,B_l|h^l)>0$. DR-SEU 2 then implies 
\[
\sum_{\pi_{qu}(\theta_0,\dots,\theta_t)\in \times_{i\leq l}SEU_{i}}\left[\prod_{k=0}^{l-1}\psi_k^{\theta_{k-1}}(\theta_k)\tau_{\pi_{qu}(\theta_k)}(f_k,A_k)\right]\cdot \psi_l^{\theta_{l-1}}(\theta_l)\tau_{\pi_{qu}(\theta_l)}(f_l,A_l)q_l(s_l)>0.
\]
If it happens that $pred(\theta_l)\neq (\theta_0,\dots,\theta_{l-1})$ then $\left[\prod_{k=0}^{l-1}\psi_k^{\theta_{k-1}}(\theta_k)\right]\cdot \psi_l^{\theta_{l-1}}(\theta_l) =0$ just by the definition of DR-SEU 1. If otherwise $pred(\theta_l)= (\theta_0,\dots,\theta_{l-1})$ but $\theta_l\neq \theta'_l$ then we showed above that $\left[\prod_{k=0}^{l-1}q_k(\pi_s(\theta_k))\tau_{\pi_{qu}(\theta_k)}(f_k,A_k)\right]\cdot \tau_{\pi_{qu}(\theta_l)}(f_l,A_l)q_l(s_l)$ = 0.

\end{proof}

\begin{definition}\label{thm:sephistdef}
A separating history for $\theta_t$ with $pred(\theta_t) = (\theta_0,\dots,\theta_{t-1})$ is a history $h^t = (B_0,g_0,s_0;\dots;B_t,g_t,s_t)\in \h^*_t$ such that $QU_{\theta_{k-1}}(B_k,g_k,s_k) = \{\pi_{qu}(\theta_k)\}$ for all $k\leq t$. For the case $k=0$ we abuse notation and write $QU_{-1}(B_0,g_0,s_0) = QU_0(B_0,g_0,s_0)$. 
\end{definition}


\begin{remark}\label{thm:sephistproperties}
1) Let $A_t\in\A_{t}$ arbitrary. After introducing LHI below, one sees easily, that when mixing a separating history for $\theta_t$ with a deterministic history such that it has the same projection on objective states as $h^{t-1}$ one can assume that $h^{t-1}$ is so that $A_t\in \A^*_{t}(h^{t-1})$. In particular separating histories are not unique. 

2) By definition, $\theta_t$ is the only state in $\Theta_{t-1}$ consistent with $h^{t-1}$ if $h^{t-1}$ is a separating history for $\theta_{t-1}$.
\end{remark}


Write 
\[
\D_{t-1} = \{d^{t-1}\in\h_{t-1}: d^{t-1} = (\{f_0\},f_0,s_0;\dots \{f_{t-1}\},f_{t-1},s_{t-1}), f_i\in \F_i\},
\]
for the set of histories such that the menu is degenerate in each period and look at its subset 
\[
\D\mc_{t-1} =\{d^{t-1}\in\h_{t-1}: d^{t-1} = (\{h_0\},h_0,s_0;\dots \{h_{t-1}\},h_{t-1},s_{t-1}),h_i,i\leq t-1\text{ are constant acts}\}.
\]
The latter consists of deterministic histories where the agent faces only constant acts and thus objective states don't matter. 

Note that given a menu $A_t\not\in \A_t(h^{t-1})$ we can always choose a $h^{t-1}\in\D\mc_{t-1}$ with $A_t\in supp(h^A_{t-1})$. Then we can define the extended aSCF as follows.

\begin{definition}\label{thm:extendedascf} For a history $h^{t-1}\in \h_{t-1}$, $A_t\in \A_t$ and $s_t\in S_t$ define
\[
\rho_t^{h^{t-1}}(\cdot,A_t,s_t) = \rho_t(\cdot,A_t,s_t|\lambda h^{t-1}+(1-\lambda)d^{t-1}),
\]
for some $\lambda\in (0,1]$, where $d^{t-1}\in \D\mc_{t-1}$ is so that $\lambda h^{t-1}+(1-\lambda)d^{t-1}\in \h_{t-1}(A_t)$.\footnote{For this to hold it suffices that $A_t\in supp(h_{t-1}^A)$.} 
\end{definition}


We prove the extension is well-defined. 

\begin{lemma}\label{thm:extascfwelldefined}
Suppose that $\rho$ satisfies LHI. Fix $t\ge 1, A_t\in\A_t, h^{t-1} = (A_0,f_0,s_0;\dots, A_{t-1},f_{t-1},s_{t-1})\in\h_{t-1}$ and $(\lambda_0,\dots,\lambda_{t-1}),(\hat\lambda_0,\dots,\hat\lambda_{t-1})\in (0,1]^t$. 

Suppose $d^{t-1} = (h_0,\{h_0\},s_0;\dots;h_{t-1},\{h_{t-1}\},s_{t-1}), \hat d^{t-1} = (\hat h_0,\{\hat h_0\},s_0;\dots;\hat h_{t-1},\{\hat h_{t-1}\},s_{t-1}) \in \D\mc_{t-1}(A_t)$. Then we have 
\[
\rho_t(\cdot,A_t,s_t|\lambda h^{t-1}+(1-\lambda)d^{t-1}) = \rho_t(\cdot,A_t,s_t|\hat\lambda \hat h^{t-1}+(1-\hat\lambda)\hat d^{t-1}).\footnote{Here the mixture operation for histories is valid 
for every pair of histories which share the same sub-history of objective states -- the mixture operation only acts on the sub-history of acts and menus. Recall that mixture of menus is defined through the Minkowski sum.  
}
\]
In particular, $\rho_t^{h^{t-1}}$ is well-defined. 
\end{lemma}

\begin{proof}

Let $k = \max\{n=0,\dots,t-1:h_n\neq \hat h_n\}$. 

Suppose that $k=-1$. This means that $d^{t-1}=\hat d^{t-1}$. If $\lambda_i>\hat\lambda_i$ for $i=0,\dots,t-1$ then the $i-$th entry of $\lambda h^{t-1}+(1-\lambda) d^{t-1}$ can be rewritten as an appropriate mixture of the $i-$th entry of $\hat\lambda \hat h^{t-1}+(1-\hat\lambda)\hat d^{t-1}$ and $(A_i,f_i,s_i)$. If on the other hand $\lambda_i\le \hat\lambda_i$ for $i=0,\dots,t-1$ then the $i-$th entry of $\lambda h^{t-1}+(1-\lambda) d^{t-1}$ can be rewritten as an appropriate mixture of the $i-$th entry of $\hat\lambda \hat h^{t-1}+(1-\hat\lambda)\hat d^{t-1}$ and $(A_i,f_i,s_i)$. Starting from $i=0$ and using LHI and working our way up the index $i=0,\dots,t-1$ we see that the aSCF is unaffected by replacing each entry of $\hat\lambda h^{t-1}+(1-\hat\lambda) \hat d^{t-1}$ with its corresponding entry from $\lambda h^{t-1}+(1-\lambda) d^{t-1}$. This shows the result for the case $k=-1$. 

Assume now for the induction step that the statement is true for all $k\leq m-1$ for some $0\leq m\leq t-1$. We show that the claim still holds for $k=m$.\footnote{The argument is the same as in the proof of Lemma 15 in \cite{fis}. It is based on the fact that when mixing a history $h^{t-1}$ with a degenerate history from $\D_{t-1}$, then the sets of maximizers $N(A_i,f_i)$ doesn't change.} 
Define the following objects.
\begin{align*}
&B_m = \frac{1}{2}A_m+\frac{1}{2}\{h_m\},\quad \hat B_m = \frac{1}{2}A_m+\frac{1}{2}\{\hat h_m\},\quad r_m =\frac{1}{2}f_m+\frac{1}{2}h_m,\quad \hat r_m =\frac{1}{2}f_m+\frac{1}{2}\hat h_m\\
& g_n = \frac{1}{2}h_n+\frac{1}{2}l_n,\quad \hat g_n = \frac{1}{2}\hat h_n+\frac{1}{2}l_n, 
\end{align*}
for soon to be specified $l_n,n =1,\dots, t-1$. Namely, define $l_n$ recursively so that they satisfy 
\[
\lambda_n A_n+(1-\lambda_n)\{g_n\}, \hat\lambda_n A_n+(1-\hat\lambda_n)\{\hat g_n\}, \frac{1}{2}A_n+\frac{1}{2}\{h_n\}, \frac{1}{2}A_n+\frac{1}{2}\hat\{g_n\},\{g_n\}\in supp(l_{n-1}^A).
\]
Finally augment the constant act $l_{m-1}$ so that 
\[
\frac{2}{3}B_m + \frac{1}{3}\{\hat g_m\},\frac{2}{3}\hat B_m + \frac{1}{3}\{ g_m\}, \frac{1}{2}\{g_m\}+\frac{1}{2}\{\hat g_m\}\in supp(l_{m-1}^A). 
\]

 Denote $c^{t-1}:=(g_n,\{g_n\},s_n)_{n=0}^{t-1}$ and $\hat c^{t-1}:=(\hat g_n,\{\hat g_n\},s_n)_{n=0}^{t-1}$ both in $\D\mc_{t-1}$. Note that we have $\lambda h^{t-1}+(1-\lambda)c^{t-1}, \hat\lambda h^{t-1}+(1-\hat\lambda)\hat c^{t-1}\in \h_{t-1}(A_t)$ by construction. Also, the last entry at which $c^{t-1}$ and $\hat c^{t-1}$ differ is $m$. Thus by repeated application of LHI we can replace $\lambda h^{t-1}+(1-\lambda)d^{t-1}$ by $\lambda h^{t-1}+(1-\lambda)c^{t-1}$ and $\hat \lambda h^{t-1}+(1-\hat\lambda)\hat d^{t-1}$ by $\hat\lambda h^{t-1}+(1-\hat\lambda)\hat c^{t-1}$. $c^{t-1},\hat c^{t-1}$ and also satisfy the following relations. 
 
 \begin{align*}
 (a):\quad & \frac{1}{2}h^{t-1}+\frac{1}{2}d^{t-1},\frac{1}{2}h^{t-1}+\frac{1}{2}\hat d^{t-1}\in \h_{t-1}(A_t),\\
 (b):\quad & \frac{2}{3}B_m+\frac{1}{3}\{\hat h_m\},\{\frac{1}{2}h_m+ \frac{1}{2}\hat\h_m\}\in supp(h_{m-1}^A),\\
 (c):\quad & \frac{2}{3}\hat B_m+\frac{1}{3}\{h_m\},\{\frac{1}{2}h_m+ \frac{1}{2}\hat\h_m\}\in supp(\hat h_{m-1}^A).
 \end{align*}
These imply immediately
 
 \begin{align*}
 (d):\quad &\left(\frac{2}{3}B_m+\frac{1}{3}\{\hat g_m\}, \frac{2}{3}r_m+\frac{1}{3}\{\hat g_m\}\right)=\left(\frac{2}{3}\hat B_m+\frac{1}{3}\{g_m\}, \frac{2}{3}\hat r_m+\frac{1}{3}\{g_m\}\right)\\& = \left(\frac{1}{3}A_m+\frac{2}{3}\{\frac{1}{2}h_m+\frac{1}{2}\hat h_m\},\frac{1}{3}f_m+\frac{2}{3}(\frac{1}{2}h_m+\frac{1}{2}\hat h_m)\right).
 \end{align*}
 Now (a)-(c) imply that the histories
 \[
 \left((\frac{1}{2}h^{t-1}+\frac{1}{2}c^{t-1})_{-m},\left(\frac{2}{3}B_m+\frac{1}{3}\{\hat g_m\}, \frac{2}{3}r_m+\frac{1}{3}\{\hat g_m\},s_m\right)\right)
 \]
 and 
 \[
 \left((\frac{1}{2}h^{t-1}+\frac{1}{2}c^{t-1})_{-m},\left(\frac{2}{3}\hat B_m+\frac{1}{3}\{g_m\}, \frac{2}{3}\hat r_m+\frac{1}{3}\{ g_m\},s_m\right)\right)
 \]
 are in $\h_{t-1}(A_t)$. Moreover, (d) implies that the first history is an entry-wise mixture of $h^{t-1}$ with $e^{t-1} = (c^{t-1}_{-m},\{\frac{1}{2}h_m+\frac{1}{2}\hat h_m)\},\frac{1}{2}h_m+\frac{1}{2}\hat h_m,s_m)$, whereas the second is an entry-wise mixture of $\hat c^{t-1}$ with $\hat e^{t-1} = (\hat d^{t-1}_{-m},\{\frac{1}{2}h_m+\frac{1}{2}\hat h_m)\},\frac{1}{2}h_m+\frac{1}{2}\hat h_m,s_m)$. 
 
 The base case of the induction ($k=-1$) gives 
 \[
 \rho_t(\cdot;A_t,s_t|\lambda h^{t-1}+(1-\lambda)c^{t-1}) = \rho_t(\cdot;A_t,s_t|\frac{1}{2} h^{t-1}+\frac{1}{2}^{t-1}c^{t-1})
 \]
 and 
 \[
 \rho_t(\cdot;A_t,s_t|\hat \lambda h^{t-1}+(1-\hat\lambda)\hat c^{t-1}) = \rho_t(\cdot;A_t,s_t|\frac{1}{2} h^{t-1}+\frac{1}{2}^{t-1}\hat c^{t-1}).
 \]
 But note that the entry where $e^{t-1}, \hat e^{t-1}$ first differ is strictly less than $m$. Hence applying the inductive hypothesis we have 
 \begin{align*}
 &\rho_t\left(\cdot;A_t,s_t\middle|\left((\frac{1}{2}h^{t-1}+\frac{1}{2}c^{t-1})_{-m},\left(\frac{2}{3}B_m+\frac{1}{3}\{\hat g_m\}, \frac{2}{3}r_m+\frac{1}{3}\{\hat g_m\},s_m\right)\right)\right)
 =\\& \rho_t\left(\cdot;A_t,s_t\middle|\left((\frac{1}{2}h^{t-1}+\frac{1}{2}\hat c^{t-1})_{-m},\left(\frac{2}{3}\hat B_m+\frac{1}{3}\{ g_m\}, \frac{2}{3}\hat r_m+\frac{1}{3}\{ g_m\},s_m\right)\right)\right).
 \end{align*}
 Combining this together with the implication from the base case we get the result. 
\end{proof}


In the next Lemma we show that the extended aSCF satisfies the formula in DR-SEU 2. 

\begin{lemma}\label{thm:drseu2formulaextended}
Suppose that we have an aSCF $\rho$ which has a DR-SEU representation as in Definition \ref{thm:drseudef} till some period $T\in \N$. Then the extended version of $\rho$ as in Definition \ref{thm:extendedascf} will satisfy DR-SEU 2, i.e. for all $t\leq T, \forall f'_t, A'_t$ and $h^{t-1} = (A_0,f_0,s_0;\dots;A_{t-1},f_{t-1},s_{t-1})$ and $f_t,A_t$ we have 
\\
\resizebox{1.1 \textwidth}{!} {
$
\rho_t(s_t,f_t,A_t|h^{t-1}) = \frac{\sum_{\pi_s(\theta_0,\dots,\theta_t)=(s_0,\dots,s_t)}\left[\prod_{k=0}^{t-1}\psi_k^{\theta_{k-1}}(\theta_k)\tau_{\pi_{qu}(\theta_k)}(f_k,A_k)\right]\cdot \psi_t^{\theta_{t-1}}(\theta_t)\tau_{\pi_{qu}(\theta_t)}(f_t,A_t)q_t(s_t)}{\sum_{\pi_s(\theta_0,\dots,\theta_{t-1})=(s_0,\dots,s_{t-1})}\prod_{k=0}^{t-1}\psi_k^{\theta_{k-1}}(\theta_k)\tau_{\pi_{qu}(\theta_k)}(f_k,A_k)}.
$
}

\end{lemma}

\begin{proof}
If $h^{t-1}\in \h_{t-1}(A_t)$ then the claim follows directly from DR-SEU2. Assume thus that $h^{t-1}\not\in \h_{t-1}(A_t)$ and take $d^{t-1} = (\{h_0\},h_0,s_0;\dots;\{h_{t-1}\},h_{t-1},s_{t-1})\in \D\mc_{t-1}$ with $d^{t-1}\in\h_{t-1}(A_t)$ and compatible with the sub-history of objective states so that according to Definition \ref{thm:extendedascf} we can define for some $\lambda\in (0,1)$
\[
\rho_t(f_t,A_t,s_t|h^{t-1}): = \rho_t(f_t,A_t,s_t|\lambda h^{t-1}+(1-\lambda)d^{t-1}). 
\]
Note that 

(1) the formula depends on the menus and acts chosen only through the tiebreakers $\tau$. 

(2) $d^{t-1}\in \D^{t-1}$ implies that for all $s\leq t$
\begin{align*}
f_s\in M(M(A_s;q_s,u_s), p_s,w_s)\equivalent \lambda f_s+(1-\lambda)h_s\in M\left(M(\lambda A_s+(1-\lambda)\{h_s\};q_s,u_s), p_s,w_s\right). 
\end{align*}
1) and 2) imply immediately that for all $s\leq t$
\[
\tau_{q_s,u_s}(f_s,A_s) = \tau_{q_s,u_s}\left(\lambda f_s+(1-\lambda)h_s,\lambda A_s+(1-\lambda)\{h_s\}\right).
\]
From here the result follows from applying DR-SEU 2 to the history $\lambda h^{t-1}+(1-\lambda)d^{t-1}$.
\end{proof}

We define $\Theta(h^{t-1})\subset\Theta_{t-1}$ as the set of states $\theta_{t-1}$ consistent with $h^{t-1}$ in the sense of Definition \ref{thm:defconsistent}. 

\begin{lemma}\label{thm:charmenuwoties}[Pendant to Lemma 14 in \cite{fis}]
Fix $t\in\{0,\dots,T\}$ and suppose that we have a DR-SEU representation up to time $t$. Take any $h^{t-1} = (A_0,f_0,s_0;\dots;(A_{t-1},f_{t-1},s_{t-1})\in\h_{t-1}$ and $A_t\in\A_t$. Then the following are equivalent.
\begin{enumerate}
\item $A_t\in\A^*_t(h^{t-1})$.
\item For each $\theta_{t-1}\in \Theta(h^{t-1})$ and $(q_t,u_t)\in \pi_{qu}\left(
supp(\psi_t^{\theta_{t-1}})\right)$ we have $|M(A_t;q_t,u_t)|=1$.
\end{enumerate}
\end{lemma}

\begin{proof}
\emph{From A. to B.:} We prove the contrapositive. Suppose that there is $\theta_{t-1}\in \Theta(h^{t-1})$ and $(q_t,u_t)\in \pi_{qu}(supp(\psi_t^{\theta_{t-1}}))$ with $|M(A_t;q_t,u_t)|>1$. Pick any $f_t\in M(A_t;q_t,u_t)$ with $\tau_{q_t,u_t}(f_t,A_t)>0$. Since $u_t$ is non-constant by DR-SEU 1, we can find lotteries $\Delta(X_t)$ with $u_t(\underline r)<u_t(\bar r)$. Fix a sequence $\alpha_n\in (0,1)$ with $\alpha_n\ra 0$ and let $f_t^n = \alpha_n \delta_{\underline{r}} + (1-\alpha_n)f_t$ as well as $\underline{g}_t^n = \alpha_n \delta_{\underline{r}} + (1-\alpha_n)g_t$ and $\bar{g}_t^n = \alpha_n \delta_{\bar{r}} + (1-\alpha_n)g_t$ for all $g_t\in A_t\setminus \{f_t\}$. Let $\underline{B}_t^n=\{\underline{g}_t^n:g_t\in A_t\setminus\{f_t\}\}$ and $\bar{B}_t^n=\{\bar{g}_t^n:g_t\in A_t\setminus\{f_t\}\}$. Finally let $B_t^n = \underline{B}_t^n\cup \bar{B}_t^n$. Then we have $B_t^n\ra^m A_t\setminus\{f_t\}$ and $f_t^n\ra^m f_t$. Furthermore, since $|M(A_t;q_t,u_t)|>1$ we can pick $g_t\in A_t\setminus \{f_t\}$ such that $q_t\cdot u_t(\bar{g}_t^n)> q_t\cdot u_t(f_t^n)$. This implies $\tau_{q_t,u_t}\left(f_t^n,B_t^n\cup\{f_t^n\}\right) = 0$. 

Furthermore, note that for $(q'_t,u'_t)\in \pi_{qu}(\Theta_t)\setminus\{(q_t,u_t)\}$ we always have 
\[
N(M(A_t;q'_t,u'_t);f_t) = N(M(\underline{B}^n_t\cup\{f_t^n\};q'_t,u'_t);f^n_t)\supseteq N(M(B^n_t\cup\{f_t^n\};q'_t,u'_t);f^n_t),
\]
which implies $\tau_{q'_t,u'_t}(f_t,A_t)\ge \tau_{q'_t,u'_t}(f_t^n,B^n_t\cup\{f_t^n\})$ for all $n$. Letting $pred(\theta_{t-1}) = (\theta_0,\dots,\theta_{t-2})$ Lemma \ref{thm:drseu2formulaextended} implies that for all $n$ and all $s_t\in S_t$\footnote{Note that we need Lemma \ref{thm:drseu2formulaextended} here because the history $h^{t-1}$ is not assured to lead to $B^n_t\cup\{f_t^n\}$ with positive probability.}
\begin{align*}
&\rho_{t}(f_t,A_t,s_t|h^{t-1}) - \rho_{t}(f^n_t,B^n_t\cup\{f_t^n\},s_t|h^{t-1}) = \\
&  \resizebox{1.0\hsize}{!}{$\frac{\sum_{\pi_s(\theta'_0,\dots,\theta'_t) = (s_0',\dots,s_t')}\left[\prod_{k=0}^{t-1}\psi_k^{\theta'_{k-1}}(\theta'_k)\tau_{\pi_{qu}(\theta'_k)}(f_k,A_k)\right]\cdot \psi_t^{\theta'_{t-1}}(\theta'_t)\left(\tau_{\pi_{qu}(\theta'_t)}(f_t,A_t)-\tau_{\pi_{qu}(\theta'_t)}(f_t^n,B_t^n\cup\{f_t^n\})\right)}{\sum_{\pi_s(\theta'_0,\dots,\theta'_{t-1}) = (s_0',\dots,s_{t-1}')}\prod_{k=0}^{t-1}\psi_k^{\theta'_{k-1}}(\theta'_k)\tau_{\pi_{qu}(\theta'_t)}(f_t,A_t)}$}\\
& \ge \frac{\prod_{k=0}^{t-1}\psi_k^{\theta_{k-1}}(\theta_k)\tau_{\pi_{qu}(\theta_t)}(f_t,A_t)}{\sum_{\pi_s(\theta'_0,\dots,\theta'_{t-1})=(s_0',\dots,s_{t-1}')}\prod_{k=0}^{t-1}\psi_k^{\theta'_{k-1}}(\theta'_k)\tau_{\pi_{qu}(\theta'_t)}(f_t,A_t)}>0.
\end{align*}
The last line doesn't depend on $n$ so we get

\[
\limsup_{n\ra\infty}\rho_{t}(f^n_t,B^n_t\cup\{f_t^n\},s_t|h^{t-1})<\rho_{t}(f_t,A_t,s_t|h^{t-1}).
\]
By Definition \ref{thm:menuwoties} we have $A_t\not\in \A_{t}^*(h^{t-1})$.

\emph{From B. to A.:} Suppose $A_t$ satisfies B. Consider any $f_t\in A_t, f_t^n\ra^m f_t, B_t^n\ra^m A_t\setminus\{f_t\}$. Consider a $\theta_{t-1}\in\Theta(h_{t-1})$ and $(q_t,u_t)\in \pi_{qu}(supp(\psi_t^{\theta_{t-1}}))$. 

By 2. we either have $M(A_t;q_t,u_t) = \{f_t\}$ or $f_t\not\in M(A_t;q_t,u_t)$. In the former case $q_t\cdot u_t(f_t)>q_t\cdot u_t(g_t)$ for all $A_t\ni g_t\neq f_t$. By linearity we have $q_t\cdot u_t(f^n_t)>q_t\cdot u_t(g^n_t)$ for all $g_t^n\in B_t^n$ for all $n$ large enough. 

This implies $\tau_{q_t,u_t}(f_t,A_t) = \lim_n \tau_{q_t,u_t}(f^n_t,B^n_t\cup\{f_t^n\}) = 1$. In the case that $f_t\not\in M(A_t;q_t,u_t)$ we have similarly $q_t\cdot u_t(f_t)<q_t\cdot u_t(g_t)$ for some $A_t\ni g_t\neq f_t$. But then linearity implies $\tau_{q_t,u_t}(f_t,A_t) = \lim_n \tau_{q_t,u_t}(f^n_t,B^n_t\cup\{f_t^n\}) = 0$. 

Overall, for all $\theta_{t-1}\in \Theta(h^{t-1})$ and $(q_t,u_t)\in \pi_{qu}(supp(\psi_t^{\theta_{t-1}}))$ it holds $\tau_{q_t,u_t}(f_t,A_t) = \lim_n \tau_{q_t,u_t}(f^n_t,B^n_t\cup\{f_t^n\})$. By looking at the formula in Lemma \ref{thm:drseu2formulaextended} we see that this implies for all $s_t\in S_t$ and all $n$ large enough
\[
\rho_t(f_t^n,B_t^n\cup\{f_t^n\},s_t|h^{t-1}) = \rho_t(f_t,A_t,s_t|h^{t-1}).
\]
This finishes the proof.
\end{proof}

Before continuing, we register the piece of notation for an arbitrary $f_t\in\F_t$: $supp^Z(f_t) := \cup_{q\in supp(f_t)}supp(q)$.

\begin{lemma}\label{thm:strongapprmenuwoties}[Pendant to Lemma 17 in \cite{fis}.]
Suppose we have a DR-SEU representation till time $T$. Fix any $\theta_{t-1}\in \Theta_{t-1}$, separating history $h^{t-1}$ for $\theta_{t-1}$ and $A_t\in \A_t$. Then there exists a sequence $A_t^n\ra^m A_t$ with $A_t^n\in \A^*_{t}(h^{t-1})$.\footnote{Note that because of Remark \ref{thm:sephistproperties} this is w.l.o.g.} Moreover, given a $(q'_t,u'_t)\in \pi_{qu}(supp(\psi_t^{\theta_{t-1}}))$ and $f_t\in M(A_t;q_t,u_t)$ we can ensure in this construction that there is $f_t^n(q'_t,u'_t)\in A_t^n$ with $f_t^n(q'_t,u'_t)\ra^m f_t$ such that $QU_{\theta_{t-1}}(A_t^n,f_t^n(q'_t,u'_t),s_t)=\{(q'_t,u'_t)\}$ for all $s_t\in supp(q'_t)$.
\end{lemma}

\begin{proof}
Let $QU(\theta_{t-1}):=\pi_{qu}(supp(\psi_t^{\theta_{t-1}}))$. By Definition \ref{thm:drseudef} there exists a finite set $Y_t\subseteq X_t$ such that (i) for any $(q_t,u_t)\in QU(\theta_{t-1})$, $u_t$ is non-constant over $Y_t$; (ii) for any distinct $(q_t,u_t)\neq (q'_t,u'_t)$, both in $supp(\psi_t^{\theta_{t-1}})$, $(q_t,u_t)\neq (q'_t,u'_t)$ on $\F_t(Y_t)$ \footnote{Recall this denotes the set of acts whose images are contained in $\Delta(Y_t)$.} and (iii) $\cup_{f_t\in A_t}supp^Z(f_t)\subseteq Y_t$.  

By (i) and (ii) and Lemma \ref{thm:separationproperty} we can find a separating menu $C_t = \{f_t(q_t,u_t):(q_t,u_t)\in QU(\theta_{t-1})\}$, i.e. such that for all $(q_t,u_t)\in QU(\theta_{t-1})$ we have $M(C_t;q_t,u_t) =\{f_t(q_t,u_t)\}$. 

Pick $z(q_t,u_t)\in argmax_{y\in Y_t}u_t(y)$ for all $(q_t,u_t)\in QU(\theta_{t-1})$, write by a small abuse of notation again $z(q_t,u_t)$ for the constant act paying out $z(q_t,u_t)$ with probability one at each state of the world and define the constant act $b_t = \frac{1}{|Y_t|}\sum_{y\in Y_t}\delta_{y}\in \Delta(Y_t)$. Again, we denote by $b_t$ with a small abuse of notation the constant act which pays out the lottery $b_t$ in each state of the world. 

By (i) we have $q_t\cdot u_t(z(q_t,u_t))>q_t\cdot u_t(b_t)$ for all $(q_t,u_t)\in QU(\theta_{t-1})$. If we then define $\hat f_t(q_t,u_t)=\alpha f_t(q_t,u_t) + (1-\alpha)z(q_t,u_t)$ we still have $q_t\cdot u_t(\hat f_t(q_t,u_t))>q_t\cdot u_t(b_t)$ if we choose $\alpha\in(0,1)$ small enough. This is because of the `finiteness' of all the data going into the problem. Note also, that if we define $\hat C_t= \{\hat f_t(q_t,u_t):(q_t,u_t)\in QU(\theta_{t-1})\}$ we still have $M(\hat C_t;q_t,u_t) =\{\hat f_t(q_t,u_t)\}$. 

Now pick for each $(q_t,u_t)\in QU(\theta_{t-1})$ a $f_t(q_t,u_t)\in M(A_t;q_t,u_t)$. To also prove the `moreover' part, pick $f_t(q_t,u_t)$ as required in the `moreover' part. Fix a sequence $\epsilon_n\in(0,1)$ going to zero. For each $n$ and $(q_t,u_t)\in QU(\theta_{t-1}):=supp(\psi_t^{\theta_{t-1}})$ let $f_t^n(q_t,u_t) = (1-\epsilon_n)f_t(q_t,u_t) + \epsilon_n\hat f_t(q_t,u_t)$. Moreover, for each $g_t\in A_t$ define $g_t^n = (1-\epsilon_n)g_t+\epsilon_n b_t$. Finally, take 
\[
A_t^n = \{f_t^n(q_t,u_t):(q_t,u_t)\in QU(\theta_{t-1})\}\cup \{g_t^n: g_t\in A_t\}. 
\]
Note that $A_t^n\ra^m A_t$. Finally, note that by construction we have $M(A_t^n;q_t,u_t) = \{f_t^n(q_t,u_t)\}$. 

Since by Remark \ref{thm:sephistproperties}, part 2) $\theta_{t-1}$ is the only state consistent with $h^{t-1}$ Lemma \ref{thm:charmenuwoties} and the construction here imply $A_t^n\in \A^*_t(h_{t-1})$, as required. The last required property, i.e. $QU_{\theta_{t-1}}(A_t^n,f_t^n(q_t,u_t),s_t)=\{(q_t,u_t)\}$ for any $s_t\in supp(q_t)$ is true by construction. 
\end{proof}

The next result proves the existence of separating histories.  
\begin{lemma}\label{thm:sephistexistence}[Pendant to Lemma 2 in \cite{fis}.]
For any $\theta_t\in \Theta_t$ with $pred(\theta_t) = (\theta'_0,\dots,\theta'_{t-1})$ there always exists a separating history. 
\end{lemma}

\begin{proof}
By Lemma \ref{thm:separationproperty} and DR-SEU 1 we can construct for $\Theta_0$ a menu $B_0 = \{f_0^{\pi_{qu}(\theta_0)}:\theta_0\in \Theta_0\}\in \A_0$ such that $QU_0(B_0,f_0^{\pi_{qu}(\theta_0)},\pi_s(\theta_0)) = \{\pi_{qu}(\theta_0)\}$ for all $\theta_0\in\Theta_0$. Proceeding inductively, again using Lemma \ref{thm:separationproperty} and DR-SEU 1, we can find a menu $B_k(\theta_{k-1})=\{f_k^{\pi_{qu}(\theta_{k})}:\pi_{qu}(\theta_{k})\in \pi_{qu}(supp(\psi_k^{\theta_{k-1}}))\}$ for all $\theta_{k-1}\in \Theta_{k-1}$ such that (!) $QU_{\theta_{k-1}}(B_k(\theta_{k-1}),f_k^{\pi_{qu}(\theta_{k})},\pi_s(\theta_k)) = \{\pi_{qu}(\theta_k)\}$ for all $\pi_{qu}(\theta_{k})\in \pi_{qu}(supp(\psi_k^{\theta_{k-1}}))$. 

Moreover, we can assume that $B_{k+1}(\theta_k)\in supp^A(f_k^{\pi_{qu}(\theta_k)})$ for all $k=0,\dots,t-1$ and $\theta_k\in \Theta_k$ by mixing each $f_k^{\pi_{qu} (\theta_k)}$ with the constant act delivering $(z,B_{k+1}(\theta_k))$ for a $z\in Z$ fixed throughout. If the mixing puts small enough probability on the constant act in question, then (!) is preserved. 

This implies in particular that 
$h^t:=(B_0,f_0^{\theta'_0},s'_0;\dots; B_t(\theta'_{t-1}), f_t^{\pi_{qu}(\theta_t)},\pi_s(\theta_t))\in \h_{t}$. Moreover, since $QU_{\theta'_{k-1}}(B_k(\theta'_{k-1}),f_k^{\pi_{qu}(\theta'_k)},\pi_s(\theta'_k)) = \{\pi_{qu}(\theta'_k)\}$, it follows by Lemma \ref{thm:consistentpred} that only the state $\theta'_k$ is consistent with $h^k$ for $k=0,\dots,t$. Additionally, by construction for all $(q_k,u_k)\in \pi_{qu}(supp(\psi_k^{\theta'_{k-1}}))$ we have $M(B_k(\theta'_{k-1});q_k,u_k) = \{f_k^{q_k,u_k}\}$. Hence, by Lemma \ref{thm:charmenuwoties} we have $B_k(\theta'_{k-1})\in A_k^*(h^{k-1})$. Since this holds for all $k$ we have overall $h^{t}\in \h_{t}^*$. In summary it follows that $h^t$ is a separating history for $\theta_t$. 

\end{proof}

\section{Proof of the main result in the dynamic setting}\label{sec:proofmain}

Here we prove the representation theorem in its AS-version for DR-SEU. The proofs for the special cases Evolving SEU and Gradual Learning are in the online appendix. 

\subsection{Proof for DR-SEU}

\subsubsection{Sufficiency} We proceed by induction on $t\leq T$. First consider $t=0$. Because of the axioms and $X_0$ being a separable metric space we have the existence of an AS-version R-SEU representation for $\rho$ on $\h^0$. Depending on the version of the representation we are looking at, i.e. whether CIB or NUC is satisfied, we also have the respective property for the representation at time $t=0$. Set $SEU_0 = \{\pi_{qu}(\theta_0):\theta_0\in\Theta_0\}$.

Suppose next that we have the representation for all $t'\leq t$. We now construct the representation for $t+1$.

To this end, pick a subjective state $\theta_t\in\Theta_t$ and pick an arbitrary separating history $h^t(\theta_t)$ for $\theta_t$. This exists by Lemma \ref{thm:sephistexistence}. Define
\[
\rho^{\theta_t}_{t+1}(\cdot,A_{t+1},s_{t+1}) = \rho(\cdot,A_{t+1},s_{t+1}|h^t(\theta_t)).
\]
Here we use for the right-hand side the extended aSCF, which is well-defined as per Lemma \ref{thm:extendedascf}. As per axioms we get a representation 

\begin{equation}\label{eq:th1axiom0}
\rho^{\theta_t}_{t+1}(f_{t+1},A_{t+1},s_{t+1}) = \sum_{(q_{t+1},u_{t+1})\in SEU_{t+1}^{\theta_t}}\psi_{t+1}^{\theta_t}(q_{t+1},u_{t+1},s_{t+1})\tau_{(q_{t+1},u_{t+1})}(f_{t+1},A_{t+1}). 
\end{equation}

Again, depending on the respective property required by the axioms on beliefs, CIB or NUC, the kernel $\psi_{t+1}^{\theta_t}$ satisfies the respectively required property in DR-SEU 1.

We set $SEU_{t+1} = \sqcup_{\theta_t} SEU_{t+1}^{\theta_t}$ and define $\Theta_{t+1}$ accordingly by the collection of all $(q_{t+1},u_{t+1},s_{t+1})$ such that $(q_{t+1},u_{t+1})\in SEU_{t+1}$ and $s_{t+1}\in supp(q_{t+1})$.\footnote{The symbol $\sqcup$ means we join them into a union of disjoint sets, i.e. if a SEU $(q,u)$ appears in the support of two distinct $\theta_t,\theta'_t$ then we count it twice.} We extend the measures $\psi^{\theta_t}_{t+1}$ to all of $SEU_{t+1}$ by setting them to zero outside of $SEU_{t+1}^{\theta_t}$.

We see that DR-SEU 1 is satisfied by Definition. 

With this definition we can rewrite \eqref{eq:th1axiom0} as 

\begin{equation*}
\rho^{\theta_t}_{t+1}(f_{t+1},A_{t+1},s_{t+1}) = \sum_{\theta_{t+1}\in \Theta_{t+1}}\psi_{t+1}^{\theta_t}(\theta_{t+1})\tau_{\pi_{qu}(\theta_{t+1})}(f_{t+1},A_{t+1}). 
\end{equation*}

Before showing DR-SEU 2, we show that the definition of $\rho_{t+1}^{\theta_t}$ doesn't depend on the particular separating history for $\theta_t$ picked in its definition.

\begin{lemma}\label{thm:fislm3}
Fix any $\theta_t\in\Theta_t$ with $pred(\theta_t) = (\theta_0,\dots,\theta_{t-1})$. Suppose $h^t = (f_0,A_0,s_0;\dots;f_t,A_t,s_t)\in\h_t$ satisfies $QU_{\theta_{k-1}}(A_k,f_k,s_k) = \{\pi_{qu}(\theta_k)\}$ for all $k=0,1,\dots,t$. Then for any $A_{t+1}\in\A_{t+1}$ and $s_{t+1}\in S_{t+1}$ it holds  $\rho_{t+1}(\cdot,A_{t+1},s_{t+1}|h^t) = \rho_{t+1}^{\theta_t}(\cdot,A_{t+1},s_{t+1})$. 
\end{lemma}

\begin{proof}
\emph{Step 1.} Let $\tilde h^t = (\tilde f_0,\tilde A_0,\tilde s_0;\dots;\tilde f_t,\tilde A_t,\tilde s_t)\in\h_t$ denote the separating history for $\theta_t$ used to define $\rho_{t+1}^{\theta_t}$. We first prove the Lemma under the assumption that $h^t\in \h^*_t$, i.e. that $h^t$ is itself a separating history for $\theta_t$. Note that since $h^t,\tilde h^t\in\h^*_t$ and $QU_{\theta_{k-1}}(A_k,f_k,s_k) = QU_{\theta_{k-1}}(\tilde A_k,\tilde f_k,\tilde s_k)=\{(q_k,u_k)\}$ Lemma \ref{thm:charmenuwoties} implies that $M(A_k,q_k,u_k) = \{f_k\}$ and $M(\tilde A_k,q_k,u_k) = \{\tilde f_k\}$. 

Pick lotteries $(r_0,\dots,r_t)\in\Delta(X_0)\times \dots\times \Delta(X_t)$ such that $A_{t+1}\in supp(r_t^A)$ and so that for all $k=0,\dots,t-1$ it holds
\[
\{B_{k+1},\tilde B_{k+1},B_{k+1}\cup \tilde B_{k+1}\}\subset supp(r_k^A),
\]
where $B_l = \frac{1}{3}A_l+\frac{1}{3}\{\tilde f_l\}+\frac{1}{3}\{r_l\}$ and $\tilde B_l = \frac{1}{3}\tilde A_l+\frac{1}{3}\{ f_l\}+\frac{1}{3}\{r_l\}$ for $l=0,\dots,t$. Here we have identified lotteries with their respective constant acts. Define also the mixture act $g_l = \frac{1}{3}f_l+\frac{1}{3}\tilde f_l+\frac{1}{3}r_l$. 

Linearity of SEU functionals implies
\[
QU_{\theta_{k-1}}(B_{k},g_k,s_k) = QU_{\theta_{k-1}}(\tilde B_{k},g_k,\tilde s_k) = QU_{\theta_{k-1}}(\tilde B_{k} \cup B_{k},g_k,\tilde s_k) = \{(q_k,u_k)\}.\footnote{Note that in the last equality it is irrelevant whether we write $\tilde s_k$ or $s_k$ because of the argument in the first paragraph of the first step of the proof.} 
\]
We also have 
\[
M(B_k,q_k,u_k) = M(\tilde B_k,q_k,u_k) =M(\tilde B_k\cup B_k,q_k,u_k) = \{g_k\}. 
\]

This implies that for all $k=0,\dots,t$ and $(q'_k,u'_k)\in \pi_{qu}\left(supp(\psi^{\theta_{k-1}}_{k-1})\right)$ we have 
\[
\tau_{q'_k,u'_k}(g_k,B_k)=\tau_{q'_k,u'_k}(g_k,\tilde B_k) = \tau_{q'_k,u'_k}(g_k,\tilde B_k\cup B_k) = \begin{cases}
1, \text{ if }\pi_{qu}(\theta_k)= \pi_{qu}(\theta'_k)\\
0, \text{ if }\pi_{qu}(\theta_k)\neq \pi_{qu}(\theta'_k).
\end{cases}
\]
By DR-SEU 2 of the inductive hypothesis it follows for all $k=0,\dots, t-1$ that 
\begin{align*}
\begin{split}
&\psi_{t}^{\theta_{t-1}}(q_t,u_t,s_t) = \rho_t(g_t,\tilde B_t,s_t|\tilde B_0,g_0,s_0;\dots;\tilde B_{t-1},g_{t-1},s_{t-1})\\
& =  \rho_t(g_t,B_t,s_t|B_0,g_0,s_0,\dots,B_{t-1},g_{t-1},s_{t-1}) \\
& = \rho_t(g_t,\tilde B_t\cup B_t,s_t|\tilde B_0,g_0,s_0,\dots,\tilde B_{k}\cup B_{k},g_{k},s_{k},\dots,\tilde B_{t-1}\cup B_{t-1},g_{t-1},s_{t-1})\\
&= \rho_t(g_t,\tilde B_t\cup B_t,s_t| B_0,g_0,s_0,\dots,\tilde B_{k}\cup B_{k},g_{k},s_{k},\dots,\tilde B_{t-1}\cup B_{t-1},g_{t-1},s_{t-1}).
\end{split}
\end{align*}
Note that in these relations we could have replaced everywhere $s_k$ with $\tilde s_k$, since both are in the support of $q_k$ by the definition of the operator $QU_{\theta_{k-1}}$. 

Since all the histories considered above are compatible  with $A_{t+1}$ we apply CHI recursively to get 
\begin{align}
\begin{split}
&\rho_{t+1}(\cdot,A_{t+1},s_{t+1}|B_0, g_0,s_0;\dots;B_t, g_t,s_t) = \rho_{t+1}(\cdot,A_{t+1},s_{t+1}|\tilde B_0\cup B_0, g_0,s_0;\dots;\tilde B_t\cup B_t, g_t,s_t)\label{eq:ref11}\\& = \rho_{t+1}(\cdot,A_{t+1},s_{t+1}|\tilde B_0, g_0,s_0;\dots;\tilde B_t, g_t,s_t).
\end{split}
\end{align}
Here $s_{t+1}\in S_{t+1}$ is arbitrary. 
Use LHI and Lemma \ref{thm:extascfwelldefined} (well-definiteness of the extended aSCF) to get 
 \begin{align}
 \begin{split}
 &\rho_{t+1}(\cdot, A_{t+1},s_{t+1}|h^t) = \rho_{t+1}(\cdot, A_{t+1},s_{t+1}|B_0,g_0,s_0;\dots;B_t,g_t,s_t),\label{eq:ref12}\\
 &\rho_{t+1}(\cdot, A_{t+1},s_{t+1}|\tilde h^t) = \rho_{t+1}(\cdot, A_{t+1},\tilde s_{t+1}|\tilde B_0,g_0,\tilde s_0;\dots;\tilde B_t,g_t,\tilde s_t).
 \end{split}
 \end{align}
Finally, we put \eqref{eq:ref11} and \eqref{eq:ref12} together to get 
\[
\rho_{t+1}(\cdot, A_{t+1},s_{t+1}|h^t) = \rho_{t+1}(\cdot, A_{t+1},s_{t+1}|\tilde h^t).
\]
This establishes the proof for the case that $h^t\in \h^*_t$. 

\emph{Step 2.} Now suppose that $h^t\not\in \h^{t*}$. Take any sequence of (valid) histories $h^{t,n}\in \h^{t*}$ with $h^{t,n}\ra^m h^t$ with $h^{t,n} = (A_0^n,f_0^n,s_0^n;\dots;A_t^n,f_t^n,s_t^n)$ for each $n$. Existence is ensured by the Axiom of History Continuity. 

\textbf{Claim.} For all large $n$ we have $QU_{\theta_{k-1}}(A_k^n,f_k^n,s_k^n) = \{\pi_{qu}(\theta_k)\}$ for all $k=0,\dots,t$. 

\emph{Proof of Claim.}
Take some subsequence $(h^{t,n_l})_{l\geq 1}$ of $(h^{t,n})_{n\geq 1}$. We have $\rho_k(f_k^{n_l},A_k^{n_l},s_k^{n_l}|h^{k-1,n_l})>0$ for all $k=0,\dots,t$ by the definition of histories. Assume that by DR-SEU 2 for $k\leq t$ we can find $\theta'_{t,n_l}\in \Theta_t$ with $pred(\theta'_{t,n_l}) = (\theta'_{0,n_l},\dots,\theta'_{t-1,n_l})$ and $(\theta'_{0,n_l},\dots,\theta'_{t,n_l})\neq (\theta_{0,n_l},\dots,\theta_{t,n_l})$ such that $\pi_{qu}(\theta'_{k,n_l})\in QU_{\theta'_{k-1,n_l}}(f_k^{n_l},A_k^{n_l},s_k^{n_l})$ for all $k=0,\dots,t$. Since $S_0\times\dots \times S_t$ is finite, by choosing an appropriate subsequence we can assume $(\theta'_{0,n_l},\dots,\theta'_{t,n_l}) = (\theta'_{0},\dots,\theta'_{t})\neq (\theta_{0},\dots,\theta_{t})$ for all $l$. Pick the smallest $k$ such that $\theta'_k\neq\theta_k$ and pick any $g_k\in A_k$. Since $A_k^{n_l}\ra^m A_k$ we can find $g_k^{n_l}\in A_k^{n_l}$ with $g_k^{n_l}\ra^m g_k$. Since we have for all $l$ that $\pi_{qu}(\theta'_k)\in QU_{\theta'_{k-1}}(f_k^{n_l},A_k^{n_l},s_k^{n_l})$, so $\pi_{qu}(\theta'_k)(f_k^{n_l})\geq \pi_{qu}(\theta'_k)(g_k^{n_l})$ and thus also $\pi_{qu}(\theta'_k)(f_k)\geq \pi_{qu}(\theta'_k)(g_k)$ by linearity of the SEU represented by $\pi_{qu}(\theta'_k)$. 

Moreover, by choice of $k$ we have $\pi_{qu}(\theta'_k)\in \pi_{qu}(supp(\psi_{k-1}^{\theta'_{k-1}})) = \pi_{qu}(supp(\psi_{k-1}^{\theta_{k-1}}))$. But the fact that $QU_{\theta_{k-1}}(f_k,A_k,s_k) = \{\pi_{qu}(\theta_k)\}$ implies that $\pi_{qu}(\theta'_k) = \pi_{qu}(\theta_k)$ for all $k$. We have thus shown that each subsequence $(h^{t,n_l})_{l\geq 1}$ of $(h^{t,n})_{n\geq 1}$ has a subsequence with the property required by the claim. A simple argument by contradiction now establishes the claim. 

\emph{End of Proof of Claim.}

The Claim establishes that for all large enough $n$, $h^{t,n}$ satisfies the assumption of the Lemma. Since $h^{t,n}\in \h_{t}^*$, Step 1 then shows that $\rho_{t+1}(f_{t+1},A_{t+1},s_{t+1}|h^{t,n}) = \rho_{t+1}^{\theta_t}(f_{t+1},A_{t+1},s_{t+1})$ for all large enough $n$ and all $f_{t+1},s_{t+1}$. History Continuity now allows to close the argument and prove that 
\[
\rho_{t+1}(f_{t+1},A_{t+1},s_{t+1}|h^t) = \rho_{t+1}^{\theta_t}(f_{t+1},A_{t+1},s_{t+1}).
\]
\end{proof}

As a next step we establish that $\rho_{t+1}(\cdot|h^t)$ is a weighted average of the $\rho_{t+1}^{\theta_t}$ for $\theta_t$ consistent with $h^t$.  

\begin{lemma}\label{thm:lemma4fis}[Pendant of Lemma 4 in \cite{fis}]
For any $f_{t+1}\in A_{t+1}$ and $h^{t} = (A_0,f_0,s_0;\dots;A_t,f_t,s_t)\in \h_t(A_{t+1})$ we have \\
\resizebox{1.1 \textwidth}{!} {
$
\rho_{t+1}(f_{t+1},A_{t+1},s_{t+1}|h^t) = \frac{\sum_{\pi_s(\theta_0,\theta_1,\dots,\theta_t)=(s_0,\dots,s_t)}\prod_{k=0}^t\psi_k^{\theta_{k-1}}(\theta_k)\tau_{\pi_{qu}(\theta_k)}(f_k,A_k)\cdot\rho_{t+1}^{\theta_t}(f_{t+1},A_{t+1},s_{t+1}) }{\sum_{\pi_s(\theta_0,\dots,\theta_t)=(s_0,\dots,s_t)}\prod_{k=0}^t\psi_k^{\theta_{k-1}}(\theta_k)\tau_{\pi_{qu}(\theta_k)}(f_k,A_k)}.
$
}
\end{lemma}

\begin{proof}
Let $(\theta_t^1,\dots,\theta_t^m)$ be the set of elements from $\Theta_t$ that are consistent with history $h^t$, as defined in Definition \ref{thm:defconsistent}. For each $j=1,\dots,m$ let $\hat h^t(j) = (B_0^j,f_0^j,s_0;\dots;B_t^j,f_t^j,s_t)$ be a separating history for $\theta_t^j$. Note that such a history exists because under $\theta_t^j$ and its predecessors the `right' sub-history of objective states $(s_0,\dots,s_t)$ has positive probability.

We can assume w.l.o.g. that for each $k=1,\dots,t$ in all objective states $s_{t-1}$ there is a positive probability (albeit possibly small) for $(z,\frac{1}{2}A_k+\frac{1}{2}B_k^j)$ for some $z$. This can be achieved by mixing with constant acts. Thus, w.l.o.g. we can ensure that 
$h^t(j) := \frac{1}{2}h^t+\frac{1}{2}\hat h^t(j)\in \h_t(A_{t+1})$. 

Note first that it holds for all $j=1,\dots,m$ 
\begin{equation}\label{eq:th1help0}
\rho(h^t(j)) = \prod_{k=0}^t \psi_k^{\theta^j_{k-1}}(\theta_k^j)\tau_{\pi_{qu}(\theta_k^j)}(f_k,A_k).
\end{equation}

This follows from the following calculation. 
\begin{align*}
\rho(h^t(j)) & = \prod_{k=0}^t\rho_k(\frac{1}{2}f_k+\frac{1}{2}f_k^j;\frac{1}{2}A_k+\frac{1}{2}B_k^j,s_k|h^k(j))\\
& = \sum_{(\theta'_0,\dots,\theta'_t)}\prod_{k=0}^t\psi_k^{\theta'_{k-1}}(\theta'_k)\tau_{\pi_{qu}(\theta'_k)}(\frac{1}{2}f_k+\frac{1}{2}f_k^j,\frac{1}{2}A_k+\frac{1}{2}B_k^j)\\
& = \prod_{k=0}^t \psi_k^{\theta^j_{k-1}}(\theta^j_k)\tau_{\pi_{qu}(\theta^j_k)}(\frac{1}{2}f_k+\frac{1}{2}f_k^j,\frac{1}{2}A_k+\frac{1}{2}B_k^j)\\
& = \prod_{k=0}^t \psi_k^{\theta^j_{k-1}}(\theta^j_k)\tau_{\pi_{qu}(\theta^j_k)}(f_k,A_k).
\end{align*}

Here the second equality follows from DR-SEU2 and the inductive hypothesis for Sufficiency. The final two equalities follow from the fact that $h^t(j)$ is a separating history for $\theta^j_t$ (see Lemma \ref{thm:charmenuwoties}). Since $\theta_t^j$ is consistent with $h^t$ it follows \\$\psi_k^{\theta_{k-1}}(\theta_k)\cdot\tau_{\pi_{qu}(\theta^j_k)}(f_k,A_k)>0$ for all $k=0,\dots,t$ and therefore also: 

for every $\pi_{qu}(\theta'_k)\in \pi_{qu}(supp(\psi_k^{\theta_{k-1}^j}))$, $\tau_{\pi_{qu}(\theta'_k)}(\frac{1}{2}f_k+\frac{1}{2}f_k^j,\frac{1}{2}A_k+\frac{1}{2}B_k)>0$ if and only if $\pi_{qu}(\theta'_k) = \pi_{qu}(\theta^j_k)$. This yields the third equality above. 

Define now $H^t = \{h^t(j):j=1,\dots,m\}\subset \h_t(A_{t+1})$. By repeated application of LHI we have that 
\begin{equation}\label{eq:th1help1}
\rho_{t+1}(f_{t+1},A_{t+1},s_{t+1}|h^t) = \rho_{t+1}(f_{t+1},A_{t+1},s_{t+1}|H^t).
\end{equation}

Moreover, we have that 
\begin{align}
\begin{split}\label{eq:th1help2}
&\rho_{t+1}(f_{t+1},A_{t+1},s_{t+1}|H^t) = \frac{\sum_{j=1}^m \rho(h^t(j))\rho_{t+1}(f_{t+1},A_{t+1},s_{t+1}|h^t(j))}{\sum_{j=1}^m \rho(h^t(j))}\\
& = \frac{\sum_{j=1}^m \prod_{k=0}^t\psi_k^{\theta_{k-1}^j}(\theta_k^j)\tau_{\pi_{qu}(\theta_k^j)}(f_k,A_k)\cdot \rho_{t+1}(f_{t+1},A_{t+1},s_{t+1}|h^t(j))}{\sum_{j=1}^m \prod_{k=0}^t\psi_k^{\theta_{k-1}^j}(\theta_k^j)\tau_{\pi_{qu}(\theta_k^j)}(f_k,A_k)}\\
& = \frac{\sum_{j=1}^m \prod_{k=0}^t\psi_k^{\theta_{k-1}^j}(\theta_k^j)\tau_{\pi_{qu}(\theta_k^j)}(f_k,A_k)\rho^{\theta_t^j}_{t+1}(f_{t+1},A_{t+1},s_{t+1})}{\sum_{j=1}^m \prod_{k=0}^t\psi_k^{\theta_{k-1}^j}(\theta_k^j)\tau_{\pi_{qu}(\theta_k^j)}(f_k,A_k)}\\
& = \frac{\sum_{\pi_s(\theta_0,\dots,\theta_t)=(s_0,\dots,s_t)}\prod_{k=0}^t\psi_k^{\theta_{k-1}}(\theta_k)\tau_{\pi_{qu}(\theta_k)}(f_k,A_k)\cdot\rho^{\theta_t}_{t+1}(f_{t+1},A_{t+1},s_{t+1})}{\sum_{\pi_s(\theta_0,\dots,\theta_t)=(s_0,\dots,s_t)}\prod_{k=0}^t\psi_k^{\theta_{k-1}}(\theta_k)\tau_{\pi_{qu}(\theta_k)}(f_k,A_k)}.
\end{split}
\end{align}

Here the first equality holds by definition of choice conditional on a set of histories. The second follows from \eqref{eq:th1help0}. Note that $h^t(j)$, being a separating history for $\theta_t^j$ and consistent with $h^t$, implies $QU_{\theta_k^j}(\frac{1}{2}f_k+\frac{1}{2}f_k^j,\frac{1}{2}A_k+\frac{1}{2}B_k^j,s_k) = \{\pi_{qu}(\theta_k^j)\}$ for each $k$. Hence, Lemma \ref{thm:fislm3} implies that $\rho_{t+1}(f_{t+1},A_{t+1},s_{t+1}|h^t(j)) = \rho_{t+1}^{\theta_t^j}(f_{t+1},A_{t+1},s_{t+1})$. This yields the third equality. 

Finally, note that if $(\theta_0,\dots,\theta_t)\in \Theta_0\times\dots\times\Theta_t$ has $(\theta_0,\dots,\theta_t)\neq (\theta^j_0,\dots,\theta^j_t)$ for all $j$, then either $\theta_t\not\in\{\theta_t^j:j=1,\dots,m\}$ or $\theta_t = \theta_t^j$ for some $j$ but $pred(\theta_t^j)\neq (\theta_0,\dots,\theta_{t-1})$. In either case we have 
$\prod_{k=0}^t\psi_k^{\theta_{k-1}}(\theta_k)\tau_{\pi_{qu}(\theta_k)}(f_k,A_k) = 0$ by the inductive step up to $t$. This justifies the last equality in \eqref{eq:th1help2}.

Combining \eqref{eq:th1help1} and \eqref{eq:th1help2}, we obtain the desired conclusion. 
\end{proof}

We show that our construction satisfies DR-SEU2 at step $t+1$ as well. We recall the representation in \eqref{eq:th1axiom0} and combine it with Lemma \ref{thm:lemma4fis} to get for any $h^t = (A_0,f_0,s_0;\dots;A_t,f_t,s_t)\in \h_t(A_{t+1})$ 

\begin{align*}
\rho_{t+1}(f_{t+1},A_{t+1},s_{t+1}|h^t) =  
\end{align*}
\resizebox{1.1\hsize}{!}{%
        $
= \frac{\sum_{\pi_s(\theta_0,\dots,\theta_t)=(s_0,\dots,s_t)}\prod_{k=0}^t\psi_k^{\theta_{k-1}}(\theta_k)\tau_{\pi_{qu}(\theta_k)}(f_k,A_k)\cdot \left(\sum_{\theta_{t+1}}\psi_{t+1}^{\theta_t}(\theta_{t+1})\tau_{\pi_{qu}(\theta_{t+1})}(f_{t+1},A_{t+1})\right)}{\sum_{\pi_s(\theta_0,\dots,\theta_t)=(s_0,\dots,s_t)}\prod_{k=0}^t\psi_k^{\theta_{k-1}}(\theta_k)\tau_{\pi_{qu}(\theta_k)}(f_k,A_k)}$%
        }
\begin{align*}
= \frac{\sum_{\pi_s(\theta_0,\dots,\theta_{t+1})=(s_0,\dots,s_{t+1})}\prod_{k=0}^{t+1}\psi_k^{\theta_{k-1}}(\theta_k)\tau_{\pi_{qu}(\theta_k)}(f_k,A_k)}{\sum_{\pi_s(\theta_0,\dots,\theta_t)=(s_0,\dots,s_t)}\prod_{k=0}^{t}\psi_k^{\theta_{k-1}}(\theta_k)\tau_{\pi_{qu}(\theta_k)}(f_k,A_k)}.
\end{align*}

\subsubsection{Necessity}

Suppose that $\rho$ admits a DR-SEU representation as in Definition \ref{thm:drseudef}. From the representation in 
DR-SEU 2 and from Lemma \ref{thm:extascfwelldefined} we have that for a fixed $h^t\in \h_t$ the static aSCF rule $\rho_{t}(\cdot|h^t)$ satisfies the static axioms.

\textbf{Claim 1.} $\rho$ satisfies CHI.

\begin{proof}
Take any $h^{t-1} = (h_{-k}^{t-1}, (A_k,f_k,s_k))$ and $\hat h^{t-1} = (h_{-k}^{t-1}, (B_k,f_k,s_k))$ with $A_k\subseteq B_k$ and $\rho_k(f_k,A_k,s_k|h^{k-1}) = \rho_k(f_k,B_k,s_k|h^{k-1})$. From DR-SEU 2 this implies 

\begin{align}
\begin{split}\label{eq:necth1help1}
&\sum_{(\theta_0,\dots,\theta_k)}\left(\prod_{l=0}^{k}\psi_l^{\theta_{l-1}}(\theta_l)\tau_{\pi_{qu}(\theta_l)}(f_l,A_l)\right) \\&= \sum_{(\theta_0,\dots,\theta_k)}\left(\prod_{l=0}^{k}\psi_l^{\theta_{l-1}}(\theta_l)\tau_{\pi_{qu}(\theta_l)}(f_l,B_l)\right).
\end{split}
\end{align}

It follows from $\tau_{\pi_{qu}(\theta_l)}(f_k,A_k)\leq \tau_{\pi_{qu}(\theta_l)}(f_k,B_k)$ that equality in \eqref{eq:necth1help1} can hold if and only if $\tau_{\pi_{qu}(\theta_l)}(f_k,A_k) =  \tau_{\pi_{qu}(\theta_l)}(f_k,B_k)$ whenever $\theta_k$ is consistent with $h_k$. This implies immediately due to DR-SEU 2 that 
\[
\rho_t(\cdot|h^{t-1}) = \rho_t(\cdot|\hat h^{t-1}).
\]
\end{proof}

\textbf{Claim 2.} $\rho$ satisfies LHI. 

\begin{proof}
Take any $A_t,s_t$ and $h^{t-1} = (A_0,f_0,s_0;\dots;A_{t-1},f_{t-1},s_{t-1})\in \h_{t-1}(A_t)$ and $H^{t-1}\subseteq \h_{t-1}(A_t)$ of the form $H^{t-1} = \{h_{-k}^{t-1}, (\lambda A_k+(1-\lambda)B_k,\lambda f_k+(1-\lambda)g_k,s_k)): g_k\in B_k\}$ for some $k<t, \lambda\in (0,1)$ and $B_k =\{g_k^j: j=1,\dots,m\}\in \A_k$. Let $\tilde A_k = \lambda A_k+(1-\lambda)B_k$ and for each $j=1,\dots,m$ let $\tilde f_k^j = \lambda f_k+(1-\lambda)g_k^j$ and $\tilde h^{t-1}(j) = (h_{-k}^{t-1},(\tilde A_k,\tilde f_k^j,s_k))$. 

By DR-SEU 2, for all $f_t$ we have 
\[
\rho_t(f_t,A_t,s_t|h^{t-1}) = \frac{\sum_{\pi_s(\theta_0,\dots,\theta_{t-1})=(s_0,\dots,s_{t-1})}\prod_{l=0}^{t}\psi_{l}^{\theta_{l-1}}(\theta_l)\tau_{\pi_{qu}(\theta_l)}(f_l,A_l)}{\sum_{\pi_s(\theta_0,\dots,\theta_{t-1})=(s_0,\dots,s_{t-1})}\prod_{l=0}^{t-1}\psi_{l}^{\theta_{l-1}}(\theta_l)\tau_{\pi_{qu}(\theta_l)}(f_l,A_l)},
\]
and by definition also 
\[
\rho_t(f_t,A_t,s_t|H^{t-1}) = \frac{\sum_{j=1}^m\rho(\tilde h^{t-1}(j))\rho_t(A_t,f_t,s_t|\tilde h^{t-1}(j))}{\sum_{j=1}^m\rho(\tilde h^{t-1}(j))}.
\]

For each $j=1,\dots,m$ DR-SEU 2 yields 

 \resizebox{0.96\hsize}{!}{%
        $
\rho_t(f_t,A_t,s_t|\tilde h^{t-1}(j)) = \frac{\sum_{\pi_s(\theta_0,\dots,\theta_t)=(s_0,\dots,s_t)}\left(\prod_{l=0,l\neq k}^{t}\psi_{l}^{\theta_{l-1}}(\theta_l)\tau_{\pi_{qu}(\theta_l)}(f_l,A_l)\right)\cdot \psi_{k}^{\theta_{k-1}}(\theta_k)\tau_{\pi_{qu}(\theta_{k-1})}(\tilde f_k^j,\tilde A_k)}{\sum_{\pi_s(\theta_0,\dots,\theta_{t-1})=(s_0,\dots,s_{t-1})}\left(\prod_{l=0,l\neq k}^{t-1}\psi_{l}^{\theta_{l-1}}(\theta_l)\tau_{\pi_{qu}(\theta_l)}(f_l,A_l)\right)\cdot \psi_{k}^{\theta_{k-1}}(\theta_k)\tau_{\pi_{qu}(\theta_k)}(\tilde f_k^j,\tilde A_k)},
$%
        }
        \\as well as       
       \begin{align*}
       &\rho(\tilde h^{t-1}(j))  = \prod_{l=0,l\neq k}^{t-1}\rho_l(f_l,A_l,s_l|\tilde h^{l-1})\rho_k(\tilde f_k^j,\tilde A_k,s_k|\tilde h^{k-1}) \\& = \sum_{\pi_s(\theta_0,\dots,\theta_{t-1})=(s_0,\dots,s_{t-1})}\left(\prod_{l=0,l\neq k}^{t-1}\psi_l^{\theta_{l-1}}(\theta_l)\tau_{\pi_{qu}(\theta_l)}(f_l,A_l)\right)\cdot \psi_k^{\theta_{k-1}}(\theta_k)\tau_{\pi_{qu}(\theta_k)}(\tilde f_k^j,\tilde A_k).
       \end{align*}
       
       We put the last three formulas together and rearrange to obtain 
  \vspace{2mm}\\     
 \resizebox{1.0\hsize}{!}{%
        $
 \rho_t(f_t,A_t,s_t|H^{t-1}) = \frac{\sum_{\pi_s(\theta_0,\dots,\theta_t)=(s_0,\dots,s_t)}\left(\prod_{l=0,l\neq k}^{t}\psi_{l}^{\theta_{l-1}}(\theta_l)\tau_{\pi_{qu}(\theta_l)}(f_l,A_l)\right)\cdot \psi_{k}^{\theta_{k-1}}(\theta_k)\left(\sum_{j=1}^m\tau_{\pi_{qu}(\theta_k)}(\tilde f_k^j,\tilde A_k)\right)}{\sum_{\pi_s(\theta_0,\dots,\theta_{t-1})=(s_0,\dots,s_{t-1})}\left(\prod_{l=0,l\neq k}^{t-1}\psi_l^{\theta_{l-1}}(\pi_{qu}(\theta_l))\tau_{\pi_{qu}(\theta_l)}(f_l,A_l)\pi_q(\theta_l)(s_l)\right)\cdot \psi_k^{\theta_{k-1}}(\pi_{qu}(\theta_k))\left(\sum_{j=1}^m\tau_{\pi_{qu}(\theta_k)}(\tilde f_k^j,\tilde A_k)\right)\pi_q(\theta_k)(s_k)}.
 $%
        }
\vspace{3mm}

     But note that for all $\theta_k\in \Theta_k$ it holds
     
     \begin{align*}
     &\sum_{j=1}^m\tau_{\pi_{qu}(\theta_k)}(\tilde f_k^j,\tilde A_k) = \sum_{j=1}^m \tau_{\pi_{qu}(\theta_k)}\left((q',u')\in \Delta(S_k)\times \R^{X_k}: f_k^j\in M(M(\tilde A_k;\pi_{qu}(\theta_k));(q',u'))\right)\\
     &=\sum_{g_k^j\in B_k} \tau_{\pi_{qu}(\theta_k)}\left((q',u')\in \Delta(S_k)\times \R^{X_k}: f_k\in M(M(A_k;\pi_{qu}(\theta_k));(q',u')), g_k^j\in M(M(B_k;\pi_{qu}(\theta_k));(q',u'))\right)\\& =  \tau_{\pi_{qu}(\theta_k)}\left((q',u')\in \Delta(S_k)\times \R^{X_k}: f_k\in M(M(A_k;\pi_{qu}(\theta_k));(q',u'))\right)\\& = \tau_{\pi_{qu}(\theta_k)}\left(f_k, A_k\right).
     \end{align*}

By plugging this into the formula for $\rho_t(f_t,A_t,s_t|H^{t-1})$ we see that $$\rho_t(f_t,A_t,s_t|h^{t-1}) = \rho_t(f_t,A_t,s_t|H^{t-1}).$$ 

\end{proof}

\textbf{Claim 3.} $\rho$ satisfies History Continuity.

\begin{proof}
Fix any $(f_t,A_t,s_t)$ and $h^{t-1} = (f_0,A_0,s_0;\dots;f_{t-1},A_{t-1},s_{t-1})\in h_{t-1}$. Let $\Theta_{t-1}(h^{t-1})\subseteq \Theta_{t-1}$ denote the set of period-$(t-1)$ states that are consistent with $h^{t-1}$. Define $\rho^{\theta_{t-1}}_t(f_t,A_t,s_t) = \sum_{\theta_t}\psi_t^{\theta_{t-1}}(\theta_t)\tau_{\pi_{qu}(\theta_t)}(f_t,A_t)$. By Lemma \ref{thm:drseu2formulaextended} we have 
\begin{align*}
&\rho_t(f_t,A_t,s_t|h^{t-1}) = \frac{\sum_{\pi_s(\theta_0,\dots,\theta_t)=(s_0,\dots,s_t)}\prod_{k=0}^t\psi_k^{\theta_{k-1}}(\theta_k)\tau_{\pi_{qu}(\theta_k)}(f_k,A_k)}{\sum_{\pi_s(\theta_0,\dots,\theta_{t-1})=(s_0,\dots,s_{t-1})}\prod_{k=0}^{t-1}\psi_k^{\theta_{k-1}}(\theta_k)\tau_{\pi_{qu}(\theta_k)}(f_k,A_k)}
\end{align*}
\resizebox{1.0\hsize}{!}{%
        $
= \frac{\sum_{\pi_s(\theta_0,\dots,\theta_{t})=(s_0,\dots,s_{t})}\prod_{k=0}^{t-1}\psi_k^{\theta_{k-1}}(\theta_k)\tau_{\pi_{qu}(\theta_k)}(f_k,A_k)\cdot \sum_{\theta_t}\psi_t^{\theta_{t-1}}(\theta_t)\tau_{\pi_{qu}(\theta_t)}(f_t,A_t)}{\sum_{\pi_s(\theta_0,\dots,\theta_{t-1})=(s_0,\dots,s_{t-1})}\prod_{k=0}^{t-1}\psi_k^{\theta_{k-1}}(\theta_k)\tau_{\pi_{qu}(\theta_k)}(f_k,A_k)}.
$%
        }
      \vspace{3mm}\\  
       We see that $\rho_t(f_t,A_t,s_t|h^{t-1})\in co\{\rho^{\theta_{t-1}}_t(f_t,A_t,s_t):\theta_{t-1}\in\Theta_{t-1}(h^{t-1})\}$. Fix any $\theta_{t-1}\in \Theta_{t-1}(h^{t-1})$. To prove the claim it suffices to show that 
        \[
        \rho^{\theta_{t-1}}_t(f_t,A_t,s_t)\in\{\lim_n \rho_t(f_t,A_t,s_t|h_n^{t-1}):h_n^{t-1}\ra^m h^{t-1}, h_n^{t-1}\in \h_{t-1}^*\}. 
        \]
     To this end, let $pred(\theta_{t-1})= (\theta_0,\dots,\theta_{t-2})$ and let $\bar h^{t-1} = (B_0,g_0,s_0;\dots;B_{t-1},g_{t-1},s_{t-1})\in \h_{t-1}^*$ be a separating history for $\theta_{t-1}$. By Lemma \ref{thm:strongapprmenuwoties} for each $k=0,\dots,t-1$ we can find sequences $A_k^n\in \A_k^*(\bar h^{k-1})$ and $f_k^n\in A_k^n$ with $f_k^n\ra^m f_k$ and $QU_{\theta_{k-1}}(A_k^n,f_k^n,s_k) = \{\pi_{qu}(\theta_k)\}$ for all $n$ and all $k=0,\dots,t-1$. Working backwards from $k=t-2$ we can inductively replace $A_k^n$ and $f_k^n$ with a mixture putting small weight on a constant act yielding $(z,A_{k+1}^n)$ for some $z$ so as to ensure that $A_{k+1}^n\in supp^A(f_k^n(s_k))$, irrespective of $s_k\in S_k$. This can be done maintaining the previous properties of $A_k^n$ and $f_k^n$. 
     
By construction it follows $h_n^{t-1} = (A_0^n,f_0^n,s_0;\dots;A_{t-1}^n,f_{t-1}^n,s_{t-1})\in \h_{t-1}^*(A_t)$ and this is also a separating history for $\theta_{t-1}$. 
     
     By Lemma \ref{thm:drseu2formulaextended} the latter fact implies for each $n$ that
     \begin{align*}
     &\rho_t(f_t,A_t,s_t|h_n^{t-1})\\& = \frac{\sum_{\theta_t}\left(\prod_{k=0}^{t-1}\psi_k^{\theta_{k-1}}(\theta_k)\tau_{\pi_{qu}(\theta_k)}(f_k^n,A_k^n)\right)\cdot\psi_t^{\theta_{t-1}}(\theta_t)\tau_{\pi_{qu}(\theta_t)}(f_t,A_t)}{\prod_{k=0}^{t-1}\psi_k^{\theta_{k-1}}(\theta_k)\tau_{\pi_{qu}(\theta_k)}(f_k^n,A_k^n)}\\&=
     \sum_{\theta_t}\psi_t^{\theta_{t-1}}(\theta_t)\tau_{\theta_t}(f_t,A_t)\\&=\rho^{\theta_{t-1}}_t(f_t,A_t,s_t).
     \end{align*}
     The desired claim follows since $h_n^{t-1}\ra^m h^{t-1}$.  
\end{proof}


\subsection{Proofs for the Comparative Statics part}

\subsubsection{Proof of Proposition \ref{thm:morecorrchar}}

This is a trivial application of Lemma 24 in the online appendix.

\subsubsection{Proof of Proposition \ref{thm:speedoflearningchar}}

This is a direct implication of the Proof of the Representation Theorems for Evolving SEU and Gradual Learning (Theorems \ref{thm:evseufiltrthm} and \ref{thm:glfiltrthm} in the main body of the paper) as well as Theorem 1 in \cite{dlst}.  

\end{appendices}

\end{document}


\title{Online Appendix for Dynamic Random Subjective Expected Utility}

\author{Jetlir Duraj\footnote{duraj@g.harvard.edu}}

\maketitle



\begin{abstract}
 The online appendix is organized as follows.
 \\In the first section we extend the results of the online appendix of \cite{lu} by explicitly modeling the tie-breaking of the agent.
 \\In section 2 we prove properties of the induced menu preference from Definition 13 in the main body of the paper and show how the sophistication axiom ties together ex-ante and ex-post choices in our dynamic setting with Anscombe-Aumann acts and stochastic taste and beliefs. 
 \\In section 3 we use the results from section 2 to give the proofs for the Evolving SEU and Gradual Learning representations (Theorems 2 and 3 in the paper).
 \\Section 4 establishes the equivalence between the Ahn-Sarver representations and the filtration-based representations. It also establishes the uniqueness for Ahn-Sarver representations.
 \\Section 5 extends the classical option value theory to the case of subjective expected utility where both taste and beliefs can be stochastic.
 \\Section 6 extends the main result of \cite{as} to the case of subjective expected utility where both taste and beliefs can be stochastic.
 \\Section 7 comments on the case of SCF as an observable.
 \\Finally, Section 8 gathers auxiliary results for Section 4 of the paper (comparative statics).

\end{abstract}



\maketitle

\newpage

In this section $\rho$ will always be a SCF. 

\section{Random Subjective Expected Utility with separable metric space and explicit tie-breaking.}\label{sec:rseuapp}

\subsection{An auxiliary result}

We start by stating a technical Lemma which is proven in \cite{fis} for lotteries. We omit its proof which is a trivial adaptation of the proof of the respective Lemma in \cite{fis}.

\begin{lemma}[Lemma 21 in \cite{fis}]\label{thm:algebrawithsets}
Fix any $X'\subset X$ with $y^*\in X'$. Here $X$ is a prize space. For any collection $\mathcal{S}$ of sets denote $\mathcal{U}(\mathcal{S})$ the collection of all finite unions of elements of $\mathcal{S}$. 

(i). If $E\in \mathcal{N}(X')$ (resp. $E\in \mathcal{N}^+(X')$), then $E^c\in \mathcal{U}(\mathcal{N}^+(X'))$ (resp. $E^c\in \mathcal{U}(\mathcal{N}(X'))$).

(ii). $\mathcal{U}(\mathcal{N}(X'))$ and $\mathcal{U}(\mathcal{N}^+(X'))$ are $\pi-$systems (i.e. closed under intersection). 

(iii). $\mathcal{F}(X')$ is the collection of all $E$ such that $E=\cup_{l\in L}M_l\cap N_l$ for some finite index set $L$ and $M_l\in \mathcal{N}(X'), N_l\in \mathcal{N}^+(X')$ for each $l\in L$.

(iv). $\mathcal{F}(X')$ is the collection of all $E$ for which there exists a finite $Y\subset X'$ with $y^*\in Y$ and $E^Y\in \mathcal{F}(Y)$ such that $E=E^Y\times \R^{X\setminus Y}$. 

\end{lemma}

\subsection{Axioms}


In this section $\rho$ will denote a SCF, i.e. objective states are unobservable for the analyst.

First the definition of the representation we are after in this section. 

\begin{definition}\label{thm:deflu}
1) A measure $\mu$ on $\Delta(S)\times \R^X$ (endowed with product of Borel sigma-Algebras) is \emph{regular} if $q\cdot u(f)=q\cdot u(g)$ with $\mu-$measure 0. 

2) The SCF $\rho$ has a R-SEU representation if there exists a regular (finitely additive or sigma-additive as required per context) probability measure $\mu$ such that for all $f\in \F$ and $A\in \A$ we have 
\[
\rho(f,A) = \mu\left((q,u)\in\Delta(S)\times \R^X:q\cdot u(f)\geq q\cdot u(g), \forall g\in A\right). 
\]
\end{definition}

Whether $X$ is finite or not, unless otherwise stated, we will look at the following (by now) classical axioms on the SCF $\rho$. 

\begin{itemize}
\item Axiom 1: Monotonicity  $\quad\rho(f,A)\geq \rho(f,B)$ for $A\subset B$. 
\item Axiom 2: Linearity $\quad\rho(\lambda f+ (1-\lambda)g, \lambda A+ (1-\lambda)\{g\}) = \rho(f,A)$ for any $A\in \A, g\in\F$ and $\lambda\in (0,1)$. 
\item Axiom 3: Extremeness $\quad\rho(ext(A),A) = 1$ for all $A\in \A$. 
\item Axiom 4: (a) (Mixture) Continuity $\quad\rho(\cdot, \alpha A + (1-\alpha)A')$ is continuous in $\alpha\in [0,1]$ for any $A,A'\in \A$. 

or

(b) Continuity $\quad\rho(\cdot, A )$ is continuous in $A\in \A$. 
\item Axiom 5: State Independence (see below)

\item Axiom 6: Finiteness There is $K>0$ such that for all $A\in \A$, there is $B\subset A$ with $|B|\leq K$ such that for every $f\in A\setminus B$ there are sequences $f^n\ra^m f$ and $B^n\ra^m B$ with $\rho(f_n,\{f_n\}\cup B^n) = 0$.
\end{itemize}

\subsection{Lu's Theorem with explicit tie-breaking and finite prize space.}
In this section $X,S$ are finite sets, resp. of prizes and objective states. We want to prove a version of Theorem S.1 in the supplement of \cite{lu} with explicit tie-breaking. The latter is assumed away in \cite{lu}.

Note that since $X,S$ finite, all acts are vectors in $\R^{|X||S|}$, a euclidean space with the canonical finite orthonormal basis $(w_1,\dots,w_m)$, $m=|S||X|$. 

\begin{lemma}\label{thm:lu1}
If $\rho$ satisfies Monotonicity, Linearity, Extremeness and Mixture Continuity then there exists a regular finitely additive measure $\nu$ over the set $\Delta^f(U)\subset\R^m$ of normalized Bernoulli utility functions (equipped with the Borel sigma-Algebra), such that
\[
\rho(f,A) = \nu\left(w\in \Delta^f(U): w\cdot f\geq w\cdot g, g\in A\right).
\]
\end{lemma}
\begin{proof}
Let $W$ be the affine hull of $\F$ in $\R^{|X||S|}$ with dimension $m$ and consider $\Delta$ be the probability simplex in $W$ as well as $\{w_1,\dots,w_m\}$ an orthonormal basis of $W$. Consider the mapping $T:\F\ra \Delta$ given by 
\[
T(f)_{i} = \lambda\left[f\cdot (w_i-\frac{1}{m}\sum_{j}w_j)\right] +\frac{1}{m}
\]
Note that $f\cdot w_i$ is a number in $[0,1]$, by definition of acts.\footnote{Note that a similar function is defined in the proof of Lemma S.2 of the supplement of \cite{lu}. There is a typo there though: after the difference in the quadratic brackets, the sum must be divided by $m$ just as here.} Also, for all $\lambda>0$ small enough we have $T(f)\geq 0$ and by definition also $T(f)\in \Delta$. Note that this transformation preserves Linearity (and thus also Extremeness), Monotonicity and Mixture Continuity. 

We can do a construction similar to the one in \cite{lu}, proof of Lemma S.2 there. Define by $W'$ the affine hull of $im(T)$ and denote by $P$ the projection from $W$ to $W'$. Note that the projection is an affine map. For each finite set $D\subset \Delta$ we pick a $p^*\in \Delta\cap W'$ and $a\in (0,1)$ such that $aP(D)+(1-a)\{p^*\}\subset im(T)$. This works, because by definition of the affine hull, the relative interior of $im(T)$ w.r.t. $W'$ is nonempty.\footnote{Note that the proof in Lemma S.2 of the supplement of \cite{lu} doesn't use any projection to the affine hull of the image of the map $T$ defined there. His trick only works if there exists an $\epsilon>0$ and $p\in im(T)$ such that the intersection of $\Delta$ with the ball of radius $\epsilon$ around $p$, is contained in the image of $T$. This is not the case though, as elements from the image of $T$, both in our setting and in his setting satisfy additional conditions related to the fact that an act $f\in \F$, seen as an element of $\R^m$, satisfies additional constraints due to the fact that the image of the act are lotteries, and therefore `lives' in a convex set of dimension less than $m$. His trick to shift the menu into the image of $T$ works once we introduce the projection $P$ as above.}

His construction works here for the same reason that it works in his paper: there is a point in the relative interior of $im(T)$ which is also in the relative interior of $\Delta$. Thus we can define a SCF $\tau$ on $\Delta$ by 
\[
\tau(p,D) = \rho(f,A)
\]
where $T(A) = aP(D)+(1-a)\{p^*\}$ and $T(f) = aP(p)+(1-a)\{p^*\}$. Linearity of $\rho$ and the fact that both $T$ and $P$ are affine ensure that the construction of $\tau$ is well-defined, i.e. independent of the pair $(a,p^*)$. Just as in \cite{lu}, axioms of Theorem 2 in \cite{gp} are satisfied, so that there exists a \emph{regular}, finitely additive probability measure $\bar\nu$ over the set $U$ of normalized Bernoulli utility functions\footnote{$U$ contains utilities $u\in \R^X$ so that for some fixed $y^*\in X$ all $u\in U$ satisfy $u(y^*) = 0$. This normalization makes $u\approx 0$ the unique constant Bernoulli utility in $U$.\label{fnlabel}} 
such that 

\[
\rho(f,A) = \tau(p,D) = \bar\nu\left(v\in U: v\cdot T(f)\geq v\cdot T(g), g\in A \right)
\]
where $T(A) = aD+(1-a)\{p^*\}$ and $T(f) = ap+(1-a)\{p^*\}$. Now note that by linearity of $T$ this can be easily written as 
\[
\rho(f,A) = \bar\nu\left(v\in U: v\cdot f\geq v\cdot g, g\in A \right).
\]
\end{proof}

Note that by the fact that $\nu$ found in Lemma \ref{thm:lu1} is regular, we don't need an axiom as Non-degeneracy for states (as \cite{lu} does) as there always exists non-degenerate states because of regularity of the measure. 

Note also, that Lemma \ref{thm:lu1} has given us state-dependent utilities of the form $u(s,x), s\in S,x\in X$. I.e. we already have a State-dependent EU representation. Adding State-Independence as in Lu and strengthening Mixture Continuity to Continuity allows to prove that $\bar\nu$ puts probability one on $u(s,\cdot)\approx u(s_1,\cdot), s\in S$. For this, we first restate State-Independence as in \cite{lu}. 
\vspace{2mm}\\
Call a menu $A$ \emph{constant} if it contains only constant acts. Given a menu $A$ and a state $s\in S$ let $A(s) = \{f(s): f\in A\}$ be the constant menu which is the $s$-section of $A$. 

\paragraph{Axiom: State Independence} Suppose $f(s_1) = f(s_2), A_1(s_1) = A_2(s_2)$ and $A_i(s) = \{f(s)\},s\neq s_i, i=1,2$. Then $\rho(f,A_1) = \rho(f,A_1\cup A_2)$. 

This is the same axiom as in the appendix of \cite{lu}. 

Moreover, the proof of the following Lemma does need \emph{Continuity}, since it uses sigma-additivity of $\bar\nu$ from Lemma \ref{thm:lu1}.\footnote{In the presence of (the stronger axiom) Continuity, $\bar\nu$ is sigma-additive (this is Theorem 3 in \cite{gp}).}


\begin{lemma}\label{thm:lu2}
Let $\rho$ satisfy Monotonicity, Continuity, Extremeness, and Linearity. Then $\rho$ has a R-SEU representation if State Independence is satisfied. 
\end{lemma}

\begin{proof}
The proof is already contained in the pgs. 8-11 of the Supplementary Appendix of \cite{lu}. One just has to note that a non-null state exists by Lemma \ref{thm:lu1} because the measure $\bar\nu$ constructed there is regular. Thus, if all states would be null, then only ties would be possible and this would lead to a contradiction.\footnote{Also note: one doesn't need to prove regularity as in the proof of \cite{lu} since it's given by Lemma \ref{thm:lu1}.}

Given sigma-additivity of $\bar\nu$ from Lemma \ref{thm:lu1} the proof that State Independence, together with Lemma \ref{thm:lu1} gives a R-SEU representation follows word-for-word the proof in \cite{lu}, except that now regularity of the induced $\bar\nu$ over the Borel sigma-Algebra of $U$ doesn't need to be shown. Lu's proof gives a measure $\mu$ on $\Delta(S)\times \R^X$ induced by $\bar\nu$ and so that the representation holds:

\[
\rho(f,A) = \mu\left((q,u)\in\Delta(S)\times \R^X:q\cdot u(f)\geq q\cdot u(g), \forall g\in A\right),\quad f\in\F, A\in\mathcal{A}. 
\]
One checks directly through the string of equalities in pg. 10 of the Supplementary Appendix of \cite{lu}, that if $\mu$ allows ties then so does $\bar\nu$. This would be a contradiction, so that $\mu$ is also regular. 

\end{proof}

This Lemma gives one direction of our version of Lu's Theorem.  

Before completing the proof of \emph{Lu's Theorem with tiebreakers} we need to modify technically some parts of \cite{fis}. 

Consider measures defined on the sigma-Algebra $\mathcal{F}$ generated by sets of the type $N(A,f)$, $N^+(A,f)$ for $A\in\A,f\in\F$.

The support of a finitely additive measure over $\mathcal{F}$ is 
\[
supp(\nu) = \left(\cup\{V\in\mathcal{F}: V\text{ is open  and }\nu(V)=0\}\right)^c. 
\]

Note that here the \emph{openness} is w.r.t. the Borel-sigma Algebra on $\Delta(S)\times\R^X$.

For future purpose note that the following Lemma doesn't depend on the cardinality of $X$ as long as it is a separable metric space (in particular it is true if $X$ is finite). 

\begin{lemma}[Lemma 22 in \cite{fis}]\label{thm:lu3}
Let $\nu$ be a regular finitely-additive probability measure on $\mathcal{F}$ and suppose that $\left(N(A,f)\setminus \Delta(S)\times\{0\}\right)\cap supp(\nu) = \emptyset$ for some $A\in \A, f\in A$, where $0$ denotes the unique constant utility in $U$. Then $\nu(N^+(A,f))  = \nu(N(A,f))= 0$.  
\end{lemma}

\begin{proof}
The proof follows the same lines as Lemma 22 in \cite{fis}, except that now we write $\left(N(A,f)\setminus \Delta(S)\times\{0\}\right)$ instead of $\left(N(A,p)\setminus \{0\}\right)$, we replace $[-1,-1]^X$ with $\Delta(S)\times [-1,-1]^X$. The latter is again compact by Tychonoff's Theorem and the finiteness of $S$. We also look at an element $(q,u^*)\in \left(N(A,f)\setminus \Delta(S)\times\{0\}\right)$ instead of $u^*\in N(A,p)$ and use that $N(A,f)$ is closed under scaling of the second component $u$. Other than that the proof follows virtually the same steps as in \cite{fis}.
\end{proof}

We use Finiteness to see that Lemma 18 in \cite{fis} holds true in our setting as well. Note, this is true for $X$ arbitrary, as long as $X$ is a separable metric space. 

\begin{lemma}\label{thm:lu4}
1) Let $\rho$ have a R-SEU representation with a regular probability measure $\mu$ and satisfy Finiteness with some natural number $K$. Let $Pref(\F)$ denote the set of all SEU preferences over $\F$. Then 
\[
|\{\better\in Pref(\F):\better \text{ is represented by some }(q,u)\in \left(supp(\mu)\setminus \Delta(S)\times\{0\}\right)\}|\in \{1,\dots,K\}.
\]

2) Let $\rho$ have a R-SEU representation with a regular probability measure $\mu$ over $\Delta(S)\times \R^X$ such that it has finite support of size $K'$. Then $\rho$ satisfies Finiteness with $K=K'$. 
\end{lemma}

\begin{proof}
1) We start first with showing that there can't be more than $K$ elements in the support. This is the same argument as in \cite{fis} (proof by contradiction), except that now we use the separation property for AA-acts (Lemma 1 in the appendix of the paper). 

The proof that there is at least one non-constant SEU preference in $supp(\mu)$ is again very similar to \cite{fis} (proof by contradiction), only that now we invoke Lemma \ref{thm:lu3} instead, to arrive at a contradiction. 

2) Let $K=K'$ and consider any $A\in \A$. For each equivalence class (which we represent with only one element) $(q,u)\in supp(\mu)$ pick a $f_{q,u}\in M(A;q,u)$ and take $B=\{f_{q,u}: (q,u)\in supp(\mu)\}$ and note that $|B|\leq K'$.

Consider the case $B\subset A$ strictly (the case $B=A$ being trivial). Pick a $f\in B\setminus A$. Take $Y\subset X$ large but finite so that all preferences in $supp(\mu)$ are non-constant when restricted to $\F(Y)$, the set of acts $g$ where for each state $s\in S$ we have $g(s)\in \Delta(Y)$. For each $(q,u)$ in $supp(\mu)$ let $\delta_{y_u}\in \argmax_{\delta_y:y\in Y} u(\delta_y)$. Denote $h$ the constant act which yields for each $s\in S$ the uniform lottery over $Y$. Following a similar argument in \cite{fis} define $h(u)$ the constant act giving $h(u)(s) = \delta_{y_u}$  for all $s\in S$ and from that $B^n = \frac{n-1}{n}B+\frac{1}{n}\{h(u):(q,u)\in supp(\mu)\}$. Finally, define also $f_n = \frac{n-1}{n}f+\frac{1}{n}h$. Clearly, $B^n\ra^mB$ and $f_n\ra^m f$ and for all $n$ large enough and all $(q,u)\in supp(\mu)$ we have $q\cdot u(\frac{n-1}{n}f_{q,u}+\frac{1}{n}h(u))>q\cdot u(f_n)$. It follows that for all large $n$ $\rho(f_n,\{f_n\}\cup B^n) = 0$. This proves Finiteness. 
\end{proof}

Note that the fact that there is a non-constant SEU $(q,u)$ in the support of $\mu$ above only used the regularity of $\mu$. 

We now show the other direction for our version of Lu's theorem. 

\begin{lemma}\label{thm:lu5}
Suppose that $\rho$ has a R-SEU representation with a regular, sigma-additive probability measure $\mu$. Then $\rho$ satisfies Monotonicity, Linearity, Extremeness, State Independence and Continuity. 
\end{lemma}

\begin{proof}
Linearity, Extremeness, Monotonicity are trivial to check. 

Regularity of $\mu$ means that $u$ in the SEU representations $(q,u)$ is non-constant (nonzero with our normalization) with $\mu$ probability equal to one. This implies that there are no null states $s\in S$. From here one checks that State Independence holds in the same way as in the Supplementary Appendix of \cite{lu}. 

The proof of Continuity holds word for word as in \cite{lu}. 
\end{proof}

By combining the above Lemmata we have proven overall for the case of finite $X$ and finite $S$ the following Proposition.

\begin{proposition}[Lu's Theorem with Explicit Tie-breaking.]\label{thm:LuThm}
Let $X$ be finite. The following are equivalent for a SCF $\rho$. 
\begin{enumerate}
\item $\rho$ satisfies Monotonicity, Linearity, Extremeness, Continuity, State Independence and Finiteness.
\item $\rho$ has a R-SEU representation with a sigma-additive and regular probability measure on $\Delta(S)\times \R^X$ (equipped with the sigma-Algebra $\mathcal{F}$). 
\end{enumerate}
Moreover, $\rho$ satisfies Finiteness if and only if the support of $\mu$ is finite. 
\end{proposition}






\subsection{Lu's Theorem with explicit tie-breaking and general prize space}

Similar to \cite{fis}, we use an increasing sequence of finite sets to define the \emph{sigma-additive} probability measures $\mu^Y$ where the prize space for the acts is restricted to some finite $Y\subset X$. By virtue of being on finitely-dimensional euclidean spaces these measures are automatically inner regular as required by the Kolmogorov extension theorem for sigma-additive probability measures. Kolmogorov consistency of the measures $\{\mu^Y,Y\subset X,\text{ finite}\}$ is proven the same way as in Claim 4 of F.2.2. in \cite{fis}. Regularity is again here given just as there, since we are looking at acts with simple lotteries (it is the fact that only finitely many prizes can occur with acts of a menu that is crucial). Part 2) of Lemma \ref{thm:lu4} is applicable here again to give that $\mu$ has finite support. It is then immediate with the same logic as in \cite{fis} that Proposition \ref{thm:LuThm} remains true for the case of a general separable, metric space $X$. 

A comment about Continuity going forward. Here we ask for full Continuity as in Axiom 4(b) above (as opposed to Mixture Continuity or to the Continuity axiom in the supplementary Appendix of \cite{fis}\footnote{See Axiom 11 there. That axiom only ensures that the Bernoulli utilities are continuous but allows for tie-breakers which are not sigma-additive.}). With the same methods as in \cite{gp} full Continuity in our setting (Axiom 4(b) above) is tantamount under Finiteness to continuity of all SEUs $(q,u)$ and sigma-additivity of the respective tie-breakers $\tau_{q,u}$. 

\subsubsection{Ahn, Sarver-like representation for Lu's Theorem}

Assume for this subsection the version of Proposition \ref{thm:LuThm} for a general separable metric space $X$.

We first note for future use that $$\{(p,w)\in \Delta(S)\times\R^{X}: f\in M(M(A;u,q);w, p)\} = N(M(A;u,q);w,p)\in \mathcal{F}.$$ 

We also note the following helpful Lemma, which corresponds to Lemma 23 in \cite{fis}. 

\begin{lemma}\label{thm:lu6}
Suppose $\mu$ is a regular and finitely-additive probability measure on $\mathcal{F}$ 
and $\left(supp(\mu) \setminus\Delta(S)\times\{0\}\right) = [(q,u)]$ for some $(q,u)\in \Delta(S)\times \R^X$. Then for any $A\in \A$ and $f\in A$ we have $\mu(N(A,f)) = \mu(N(M(A,(q,u)),f)).$ 
\end{lemma}

\begin{proof}
It's the same as Lemma 23 in \cite{fis}, only that now we have to replace $\{0\}$ with $\Delta(S)\times\{0\}$, and use Lemma \ref{thm:lu3} instead.  
\end{proof}

Then one can construct the AS-representation as in \cite{fis}: \\take $\{(q_i,u_i)\in \left(supp(\mu)\setminus \Delta(S)\times\{0\}\right), i=1,\dots,L\}$, use Lemma 1 in main body of paper to construct a separating menu $\{f_{q,u}:(q,u)\in supp(\mu)\}$. Here we have abused notation by identifying a preference from $supp(\mu)$ with a chosen representation $(q,u)$ from the equivalence class of equivalent representations. For the sake of concreteness we pick $(q,u)$ such that $u\in U$ in the following. 

One can then define the open sets $B_{q,u} = N^+(A,f_{q,u})\in \mathcal{F}$. It holds $\mu(B_{q,u})>0$ for each representative $(q,u)$ from $supp(\mu)$. One can then prove just as in \cite{fis} that $\mu(\cup_{[(q,u)]\in supp(\mu)} B_{q,u}) = 1$.\footnote{This is Lemma 19 in \cite{fis}. Note that it implies that any positive probability on a trivial SEU representation $(q,0)$ is already included in the tie-breaking procedure of the agent. In \cite{lu} the case of constant $u$ is assumed away explicitly through an Axiom (Non-Degeneracy) and because ties are per definition not observable in his model, a constant $u$ would make for a trivial R-SEU overall, so has to be excluded. In \cite{fis} the axioms and thus the representation allow for constant $u$. The collection of other axioms (in particular Extremeness) and the fact that one is using SCF as defined in the main body of the paper as an observable implies a tie-breaking procedure for the agent where `ties' occur with probability zero. 

} 
This allows to define 

\[
\tau_{q,u}(V) = \frac{\mu(V\cap B_{q,u})}{\mu(B_{q,u})}, \quad \mu({q,u}):=\mu(B_{q,u}),\quad \forall (q,u)\in supp(\mu).  
\]

All of the $\tau_{q,u}$ are regular probability measures over $\Delta(S)\times \R^X$ equipped with sigma-Algebra $\mathcal{F}$.

One can now prove the pendant of Lemma 20 in \cite{fis}. 

\begin{lemma}\label{thm:lu7}
Let $\rho$ have a R-SEU representation with a regular $\mu$ over $\mathcal{F}$ and satisfy Finiteness.

For each $[(q,u)]\in supp(\mu), A\in\A$ and $f\in A$ we have
\begin{equation}\label{eq:helplu1}
\mu(N(A,f)) = \sum_{[(q,u)]\in supp(\mu)}\mu(q,u)\tau_{q,u}\left(\{(p,w)\in \Delta(S)\times\R^{X}: f\in M(M(A;u,q);w, p)\}\right).
\end{equation}
\end{lemma}

\begin{proof}
One shows that $supp(\tau_{q,u}\setminus \Delta(S)\times \{0\}) = [(q,u)]$. Similar to the proof of Lemma 20 in \cite{fis} this uses Lemma \ref{thm:algebrawithsets} here and the fact that $B_{q,u}\cap B_{q',u'} = \emptyset$ for two non-trivial SEU preferences whenever $(q,u)\not\approx (q',u')$, as well as the fact that sets from $\mathcal{F}$ of the type $N(A,f),N^+(A,f)$ are closed under positive affine transformations of their \emph{second} component. Just as in \cite{fis} the proof is closed by showing precisely \eqref{eq:helplu1} through the use of Lemma \ref{thm:lu6} and the fact that $\mu(\cup_{[(q,u)]\in supp(\mu)} B_{q,u}) = 1$.
\end{proof}

We now define the Ahn-Sarver type of representation. 

\begin{definition}\label{thm:AS-RSEU-SCFdef}
Let $\rho$ be a SCF for acts in $\F$ over $\Delta(X)$ where $X$ is a separable metric space and $S$, the set of objective states is finite. 

We say that $\rho$ admits an \emph{AS-version R-SEU  representation} if there is a triple $$(SubS,\mu,\{((q,u),\tau_{q,u}):(q,u)\in SubS\})$$ such that 

\begin{enumerate}
\item $SubS$ is a finite subjective state space of distinct, continuous and non-constant SEUs and $\mu$ is a full support probability measure on $SubS$. 
\item For each $(q,u)\in SubS$ the tie-breaking rule $\tau_{q,u}$ is a regular sigma-additive probability measure on $\Delta(S)\times U$ endowed with the respective Borel sigma-Algebra.\footnote{$\Delta(S)\times U$ is endowed with Borel sigma-Algebras. The space of once-normalized Bernoulli utilities $U$ is defined in footnote \footref{fnlabel}.}
\item For all $f\in \F$ and $A\in \A$ we have 
\[
\rho(f,A) = \sum_{(q,u)\in SubS}\mu(q,u)\tau_{q,u}(f,A),
\]
where $\tau_{q,u}(f,A) := \tau_{q,u}(\{(p,w)\in \Delta(S)\times U: f\in M(M(A;u,q);w,p)\})$.
\end{enumerate}

\end{definition}



We finally arrive at the version of the Theorem of Lu we use in the proofs. 

\begin{theorem}[Lu's Theorem with general prize space and in the AS-version]\label{thm:ASSREU}
The SCF $\rho$ on $\A$ admits an AS-version R-SEU representation
if and only if it satisfies 
\begin{enumerate}
\item Monotonicity
\item Linearity
\item Extremeness
\item Continuity
\item State Independence
\item Finiteness. 
\end{enumerate}
\end{theorem}

\begin{proof}
\textbf{Sufficiency.} Assume that the Axioms are satisfied. Then Proposition \ref{thm:LuThm} in the case of $X$ separable, metric (see discussion at the start of this sub-sub-section up and including Lemma \ref{thm:lu7}) gives us a regular probability measure $\mu$ over $\mathcal{F}$ which implies a R-SEU representation of $\rho$. Due to Finiteness $\mu$ has finite support (Lemma \ref{thm:lu4}). Lemma \ref{thm:lu7} gives us the candidate for the AS-R-SEU representation: $SubS = \{(q,u):[(q,u)]\in supp(\mu)\setminus (\Delta(S)\times \{0\})\}$, $\mu$ and also $\tau_{q,u}$. By construction all distinct $(q,u)$ represent distinct SEU preferences and $\mu$ is full support (recall that $\mu(\cup_{[(q,u)]\in supp(\mu)} B_{q,u}) = 1$). Each $\tau_{q,u}$ is a regular sigma-additive probability measure on $\Delta(S)\times U$ endowed with the sigma-Algebra $\mathcal{F}$. 

\textbf{Necessity.} Suppose that we have the AS-version R-SEU representation \\$\left(SubS,\mu,\{((q,u),\tau_{q,u}):(q,u)\in SubS\}\right)$. 

Fix some finite $Y\subset X$ with $y^*\in Y$ and look at $(SubS,\mu,\{((q,u|_Y),\tau_{q,u|_Y}):(q,u|_Y)\in SubS\})$. Note that by $\tau_{q,u|_Y}$ we mean $\tau_{q,u|_Y}(B) = \tau_{q,u}(B\times \R^{X\setminus Y})$. This is an AS-version R-SEU representation of $\rho^Y$, as one can easily check.\footnote{$\rho^Y$ is simply $\rho$ restricted to menus of acts from $\F(Y)$.} Arguments from the proof of Proposition \ref{thm:LuThm} for the finite set of prizes $Y$ give that $\rho^Y$ satisfies Linearity, Monotonicity, Extremeness and State-Independence. 
Since the menus are finite and the lotteries simple, one can use this fact to easily see that $\rho$ satisfies Linearity, Monotonicity, Extremeness and State-Independence. By choosing $Y\subset X$ finite but large enough, it follows from Proposition \ref{thm:LuThm} and the application of Lemma \ref{thm:lu4} that Finiteness is satisfied. 
\end{proof}

We now establish uniqueness of the AS-version R-SEU-representation. 

\begin{proposition}\label{thm:SCFuniqueness}
The AS-version R-SEU-representation for a SCF $\rho$ is essentially unique in the sense that for each two representations the only degree of freedom is positive affine transformations of the Bernoulli utilities of the SEUs in the support of the respective measures of the representations.
\end{proposition}

\begin{proof}
Assume that there are two representations for $\rho$, i.e. there are two finite sets of SEU-s $SubS_i,i=1,2$ with normalized Bernoulli utilities in $\U$, measures $\mu_i,i=1,2$ over resp. $SubS_i$ and regular tiebreakers $\tau^j_{q_i,u_i}$ ($j=1,2$) so that for all $f\in A, A\in \A$ we have 

\[
\rho(f,A) = \sum_{(q_1,u_1)\in SubS_1}\mu_1(q_1,u_1)\tau^1_{q_1,u_1}(f,A) =\sum_{(q_2,u_2)\in SubS_2}\mu_2(q_2,u_2)\tau^2_{q_2,u_2}(f,A).
\]

W.l.o.g. we can assume that $SubS_i$ is the support of the respective probability measure $\mu_i$.

\emph{Step 1: }$SubS_1=SubS_2$. 

Assume this is not the case and assume w.l.o.g. that $SubS_1$ contains some $(q_1,u_1)\not\in SubS_2$. Consider then the set of SEU-s $\{(q_1,u_1)\}\cup SubS_2$, all distinct and pick a separating menu for it according to Lemma 1 in appendix of paper. Denote this menu by 
\[
A = \{f(q,u): (q,u)\in \{(q_1,u_1)\}\cup SubS_2\}. 
\]
Then note that by the second representation we have $\rho(f(q_1,u_1),A)=0$ but by the first $\rho(f(q_1,u_1),A)>0$. This is a contradiction. 

\emph{Step 2: }$\mu_1=\mu_2$. By taking the same menu as in Step 1. we get from the representations that 
\[
\rho(f(q,u),A) = \mu_i(q,u)\tau_{q,u}(f(q,u),A)= \mu_i(q,u), \quad (q,u)\in SubS, i=1,2. 
\]

\emph{Step 3: }$\tau^1_{q,u}=\tau^2_{q,u}$ for all $(q,u)\in SubS$. Fix a $(q,u)\in SubS$. Take an arbitrary $g\in \F$ with $g\in B,B\in\A$ and the separating menu $A$ as in the first two steps. Consider the menu $C(\alpha) = (A\setminus \{f(q,u)\})\cup (\alpha B + (1-\alpha)\{f(q,u)\})$. We have first that for all $\alpha\in (0,1)$ small enough it holds 
\[
\rho(\alpha B+ (1-\alpha)f(q,u), C(\alpha)) = \mu(q,u)\tau^i_{q,u}(\alpha B+ (1-\alpha)f(q,u), C(\alpha)). 
\]
In particular, it follows for all these small enough $\alpha$ and the definition of the tie-breakers that 
\begin{align*}
&\tau^i_{q,u}(\alpha g+ (1-\alpha)f(q,u), C(\alpha))  = \tau^i_{q,u}(\alpha g+ (1-\alpha)f(q,u), \alpha B+ (1-\alpha)\{f(q,u)\}) \\
&= \tau^i_{q,u}(g,B).
\end{align*}
This calculation is the same for both $i=1,2$ so that overall it follows $\tau^1_{q,u} = \tau^2_{q,u}$.
\end{proof}

\begin{remark}
Note that Theorem 4 in \cite{fis}, i.e. the case of REU is trivially included in Theorem \ref{thm:ASSREU}; just take $S=\{s\}$. 
\end{remark}

Finally, we note down a simple auxiliary Proposition.
\begin{proposition}\label{thm:revealedsuppmu}
Suppose the SCF $\zeta$ satisfies all the properties of Theorem \ref{thm:ASSREU}. Then for all SEUs $(q,u)$ it holds true (up to positive affine transformations of $u$)
\begin{align*}
(q,u)\in supp(\mu) \quad\equivalent\quad &\forall A\in \A, f\in A\text{ if }(q,u)\in N(A,f)\text{ and }\rho(f,A) = 0,\\&\text{ then there exists }(f_n,A_n)\ra(f,A)\footnote{$f_n\in A_n$ holds for all n. Convergence is in the product metric given by the usual norm on $\F$ and its related induced Hausdorff norm.}\text{ with }\rho(f_n,A_n) > 0.
\end{align*}
\end{proposition}

\begin{proof}[Proof of Proposition \ref{thm:revealedsuppmu}]
Suppose that $(q,u)\not\in supp(\mu)$ and take a separating menu $\bar A$ for $\{(q,u)\}\cup supp(\mu)$. It follows then from the representation in Theorem \ref{thm:ASSREU} that $\rho(f(q,u),\bar A) = 0$ and that this remains true in a neighborhood of the choice data $(f,A)$, where neighborhoods are taken w.r.t. product of the topology on $\F$ with the Hausdorff topology on $\A$. 

Suppose now that $(q,u)\in supp(\mu)$, $\rho(f,A) = 0$ despite $(q,u)\in N(A,f)$. Given the representation in Theorem \ref{thm:ASSREU} we know that $\tau_{q,u}(f,A)=0$ and that there must exist some $g\in N(A,f), g\neq f$. Denote $N(f) = \{g\in A: g\in N(A,f),g\neq f\}$. Take a separating menu $\bar A$ for $supp(\mu)$. If there exists some other $supp(\mu)\ni(q',u')\neq (q,u)$ set $f' = f(q',u')$. If not, pick some $(q',u')\neq (q,u)$ and some $f'$ such that $\{f(q,u),f'\}$ separate the two SEUs. 

Look at $A_n = \{f_n=\frac{1}{n}f(q,u)+(1-\frac{1}{n})f\}\cup A\setminus N(f)\cup \{\frac{1}{n}f'+(1-\frac{1}{n})g:g\in N(f)\}$. Clearly $\tau_{q,u}(h,A_n) = 1$ if and only if $h=f_n$ and otherwise it is zero for $h\in A_n,h\neq f_n$. It follows $\rho(f_n,A_n)>0$ for all $n$. 
\end{proof}

\section{From Stochastic Choice to Menu Preference}\label{sec:stochmenupref}


Given the representations, and $\better_{h^t}$ from Definition 13 in main body of paper, let $\sim_{h^t}$, $\sbetter_{h^t}$ denote the symmetric and asymmetric components of $\better_{h^t}$. Note that $\better_{h^t}$ is potentially incomplete.

\begin{lemma}[pendant of Lemma 5 in \cite{fis} in our setting.]\label{thm:menuprefprops}
Suppose that $\rho$ admits a DR-SEU representation. Consider any $t\leq T-1$, $h^t=(A_0,f_0,s_0;\dots;A_t,f_t,s_t)\in \h_t$ and $g_t,r_t\in \F_t$.

(i) If $g_t\better_{h^t}r_t$ then $q_t\cdot u_t(g_t)\ge q_t\cdot u_t(r_t)$ for all $\theta_t=(q_t,u_t,s_t)\in \Theta_t$ consistent with $h^t$.

(ii) Suppose there exists $g,b\in\Delta(X_t)$ with $\pi_u(\theta_t)(g)>\pi_u(\theta_t)(r_t)$ for all $\theta_t$ consistent with $h^t$. If $\pi_{q} (\theta_t)\cdot \pi_u(\theta_t)(g_t)\ge \pi_{q}(\theta_t)\cdot \pi_u(\theta_t)(r_t)$ for all $\theta_t$ consistent with $h^t$ then  $g_t\better_{h^t}r_t$.

(iii) If $h^t$ is a separating history for $\theta_t$, then $g_t\better_{h^t}r_t$ if and only if $\pi_{q}(\theta_t)\cdot \pi_u(\theta_t)(g_t)\ge \pi_{q}(\theta_t)\cdot \pi_u(\theta_t)(r_t)$. 
\end{lemma}

\begin{proof}
(i) We prove the contrapositive. Assume thus that there is $\theta_t$ consistent with $h^t$ so that $q_t\cdot \pi_u(\theta_t)\cdot\pi_q(\theta_t)(g_t)< \pi_u(\theta_t)\cdot\pi_q(\theta_t)(r_t)$. Consistency of $\theta_t$ with $h^t$ implies $\prod_{k=0}^t\psi_k^{\theta_{k-1}}(\theta_k)\tau_{\pi_{qu}(\theta_k)}(f_k,A_k)>0$ for $pred(\theta_t) = (\theta_0,\dots,\theta_{t-1})$. We also have from the assumption that for any $g_t^n\ra^m g_t, r_t^n\ra^m r_t$ for all large enough $n$ it holds $\pi_u(\theta_t)\cdot\pi_q(\theta_t)(g^n_t)< \pi_u(\theta_t)\cdot\pi_q(\theta_t)(r^n_t)$. It then follows by linearity that for all $n$ large enough we have $\tau_{\pi_{qu}(\theta_k)}(\frac{1}{2}f_k+\frac{1}{2}r^n_t,\frac{1}{2}A_k+\frac{1}{2}\{g_t^n,r_t^n\})\geq \tau_{\pi_{qu}(\theta_k)}(f_k,A_k)>0$. Thus by Lemma 4 in the appendix of the main body of the paper we have $\rho_t(\frac{1}{2}f_k+\frac{1}{2}r^n_t,\frac{1}{2}A_k+\frac{1}{2}\{g_t^n,r_t^n\},s_t|h^{t-1})>0$ for all large $n$. By definition and since the perturbation sequences $\{g_t^n,r_t^n\}$ were arbitrary we have $g_t\not\better_{h^t}r_t$. 

(ii) Let $\Theta_t(h^t)$ be the set of $\theta_t$ consistent with $h^t$ and suppose that $\pi_{q}(\theta_t)\cdot\pi_u(\theta_t)(g_t)\geq \pi_{q}(\theta_t)\cdot\pi_u(\theta_t)(r_t)$ for all $\theta_t\in\Theta_t(h^{t})$. Pick then $g^n_t = \frac{n}{n+1}g_t+\frac{1}{n+1}\delta_{g}$ and $r^n_t = \frac{n}{n+1}r_t+\frac{1}{n+1}\delta_{b}$ for all $n$. Obviously, $g_t^n\ra^m g_t$ and $r_t^n\ra^m r_t$ and also $\pi_{q}(\theta_t)\cdot\pi_u(\theta_t)(g^n_t)\geq \pi_{q}(\theta_t)\cdot\pi_u(\theta_t)(r^n_t)$ for all $\theta_t\in\Theta_t(h^t)$. 

Consider any $(\theta'_0,\dots,\theta'_t)$. Then either $\theta'_t\in\Theta_t(h^t)$ or it holds $\theta'_t\not\in\Theta_t(h^t)$. In the former case $\tau_{\pi_{qu}(\theta_k)}(\frac{1}{2}f_k+\frac{1}{2}r^n_t,\frac{1}{2}A_k+\frac{1}{2}\{g_t^n,r_t^n\})=0$ for all $n$, which implies $$\prod_{k=0}^{t-1}\psi_k^{\theta'_{k-1}}(\theta'_k)\tau_{\pi_{qu}(\theta_k)}(f_k,A_k)\cdot \psi_t^{\theta'_{t-1}}(\theta'_t)\tau_{\pi_{qu}(\theta'_t)}(\frac{1}{2}f_t+\frac{1}{2}r^n_t,\frac{1}{2}A_k+\frac{1}{2}\{r_t^n,g_t^n\})=0.$$

In the case $\theta'_t\not\in\Theta_t(h^t)$ it holds
$$\prod_{k=0}^{t-1}\psi_k^{\theta'_{k-1}}(\theta'_k)\tau_{\pi_{qu}(\theta_k)}(f_k,A_k)\cdot \psi_t^{\theta'_{t-1}}(\theta'_t)\tau_{\pi_{qu}(\theta'_t)}(\frac{1}{2}f_t+\frac{1}{2}r^n_t,\frac{1}{2}A_t+\frac{1}{2}\{r_t^n,g_t^n\})=0.$$
This is because
$$\prod_{k=0}^{t-1}\psi_k^{\theta'_{k-1}}(\theta'_k)\tau_{\pi_{qu}(\theta_k)}(f_k,A_k)\cdot \psi_t^{\theta'_{t-1}}(\theta'_t)\tau_{\pi_{qu}(\theta'_t)}(f_t,A_t)=0.$$

In both cases it follows by Lemma 4 in the appendix of the paper that $\rho_t(\frac{1}{2}f_t+\frac{1}{2}r^n_t,\frac{1}{2}f_t+\frac{1}{2}\{r^n_t,g^n_t\}|h^{t-1}) = 0$ for all $n$, that is also $g_t\better_{h^t}r_t$.

(iii) Let $h^{t}$ be a separating history for $\theta_t$ (note that it exists by Lemma 7 from the appendix of the main paper). One direction is covered by (i). For the other direction note that since $\pi_u(\theta_t)$ is non-constant (by DR-SEU 2) there exists $b,g\in \Delta(X_t)$ satisfying the conditions of (ii). Since $\theta_t$ is the only state consistent with $h^t$ (due to Lemma 2 in the appendix of the main paper), (ii) implies the other direction.
\end{proof}

Definition 13 in main body of paper allows to define a menu preference adapted to the filtration of the $\theta_t$-s. For this, we use Separability for the preference $\better_{h^t}$.

Fix a state $\theta_t$ and a separating history $h^t$ for $\theta_t$. While there will be many separating histories for $\theta_t$ we can see from (iii) in Lemma \ref{thm:menuprefprops} that if we define $\better_{\theta_t}$ as the preference over $\F_t$ equal to $\better_{h^t}$ for an $h^t$ as in (iii) of Lemma \ref{thm:menuprefprops}, $\better_{\theta_t}$ is well-defined, i.e. it doesn't depend on the choice of the separating history. We use this definition in the following.

\begin{lemma}\label{thm:separability+continuity}
For any $\theta_t\in\Theta_t$ there exists functions $v_{\theta_t}:Z\ra\R$ and $V_t^{\theta_t}:\A_{t+1}\ra\R$ such that 
\begin{enumerate}
\item $\pi_u(\theta_t) (\delta_{(z,A)}) = v_{\theta_t}(z)+ V_t^{\theta_t}(A)$ whenever Separability holds.
\item $V_t^{\theta_t}$ is continuous whenever $\better_{\theta_t}$ is continuous. 
\item $V_t^{\theta_t}$ is monotonic w.r.t. set inclusion whenever $\better_{\theta_t}$ satisfies $\delta_{z,A}\better_{\theta_t}\delta_{z,B}$ for all $z\in Z$ and $B\subset A, A\in \A_{t+1}$ (Monotonicity/Option Value). \item $V_t^{\theta_t}$ is linear, i.e. $V_t^{\theta_t}(\alpha A_{t+1}+(1-\alpha)A'_{t+1}) = \alpha V_t^{\theta_t}(A_{t+1})+ (1-\alpha)V_t^{\theta_t}(A'_{t+1})$ for all $A_{t+1},A'_{t+1}\in \A_{t+1}, \alpha\in(0,1)$, whenever Indifference to Timing holds. 
\item There exists $C_{t+1}, C'_{t+1}\in \A_{t+1}$ such that for all $\theta_{t}\in \Theta_t$ we have $V_t^{\theta_t}(C'_{t+1})> V_t^{\theta_t}(C_{t+1})$, whenever Menu-Non-Degeneracy holds.
\item $V_t^{\theta_t}$ has a DLR-SEU representation as in Definition \ref{thm:DLR-SEUgeneraldef} whenever $\better_{\theta_t}$ additionally satisfies Weak Dominance.
\end{enumerate}
\end{lemma}

\begin{proof}
A. Due to Separability we have for two constant acts $\frac{1}{2}\delta_{z,A}+\frac{1}{2}\delta_{x,B}\sim_{\theta_t}\frac{1}{2}\delta_{x,A}+\frac{1}{2}\delta_{z,B}$ whenever $x,z\in Z$ and $A,B\in \A_{t+1}(h^{t})$. By using (iii) of Lemma \ref{thm:menuprefprops} this implies for $u_t=\pi_u(\theta_t)$ that
\[
u_t(z,A)+u_t(x,B) = u_t(x,A)+u_t(z,B).
\]
Define $v_{\theta_t}(z) =u_t(z,B)-u_t(x,B)$ and $V_t^{\theta_t}(A) = u_t(x,A)$. This gives the required identity in the statement.

B. If $\better_{\theta_t}$ is continuous, then so is obviously the Bernoulli utility function $u_t$. This implies the result. 

C. Is also obvious from 1.

D. Indifference to Timing together with part (iii) in Lemma \ref{thm:menuprefprops} and part 1. here imply this part with the same arguments as in \cite{fis}.\footnote{Recall that in this subsection we assume that we have a DR-SEU representation.} 

E. Just as in \cite{fis} take for each $\theta_t$ a separating history $h^{t-1}(\theta_t)$ and then menus $A'_{t+1}(\theta_t)$ and $A_{t+1}(\theta_t)$ with 
\[
V_t^{\theta_t}(A'_{t+1}(\theta_t))> V^{\theta_t}(A_{t+1}(\theta_t)),\text{ for all }\theta_t\in \Theta_t.
\]
Take now just as in \cite{fis}: $C'_{t+1} = \cup_{\theta_t\in\Theta_t}\left(A'_{t+1}(\theta_t)\cup A_{t+1}(\theta_t\right)$ and $C_{t+1} = \frac{1}{|\Theta_t|}\sum_{\theta_t\in\Theta_t}A_{t+1}(\theta_t)$. Again, due to 3. and 4. above the result follows. 

F. Follows directly from the previous points and Theorem \ref{thm:DLR-SEU} in section \ref{sec:ovseu}.
\end{proof}

In the following we assume that $\better_{\theta_t}$
satisfies all properties of Lemma \ref{thm:separability+continuity}.

\begin{corollary}\label{thm:corollary}
Let the conditions of Lemmas \ref{thm:menuprefprops} and \ref{thm:separability+continuity} be satisfied.

Fix any $t\leq T-1$ and $h^t\in \h_t$. Then $g_t\better_{h^t} r_t$ if and only if $q_t\cdot u_t(g_t)\geq q_t\cdot u_t(r_t)$ for all $\theta_t = (q_t,u_t,s_t)\in\Theta_t$ consistent with $h^t$. 
\end{corollary}

\begin{proof}
One direction is just part (i) of Lemma \ref{thm:menuprefprops}. For the other direction let $C_{t+1},C'_{t+1}$ as in E. of Lemma \ref{thm:separability+continuity}. Pick any $z\in Z$ and let the constant acts $g_{t+1},b_{t+1}$ give in each objective state the continuation menus $C'_{t+1}$, respectively $C_{t+1}$. By the separability properties of Lemma \ref{thm:separability+continuity} we get $\pi_{q,u}(\theta_t)(g_{t+1})>\pi_{q,u}(\theta_t)(b_{t+1})$ for all $\theta_t\in\Theta_t$. Hence the other direction follows from part (ii) of Lemma \ref{thm:menuprefprops}. 
\end{proof}

Part A. of Lemma \ref{thm:separability+continuity} allows to define a menu preference on $\A_{t+1}$ from $\better_{\theta_t}$, if Separability holds. 
 
\begin{definition}\label{thm:menupref}
Fix a $z_t\in Z$. Take a $\theta_t$ and define an ex-post menu preference $\better_{\theta_t}$ over $\A_{t+1}$ by \footnote{Here we abuse some notation in that in the future in the proof of the Evolving-SEU representation we will use $\better_{\theta_t}$ to denote revealed preference derived from $\rho$ as in Definition 13 in main body of paper. This shouldn't lead to confusion as it is clear from the context each time which preference is meant.} 
\[
A_{t+1}\better_{\theta_t}B_{t+1},\text{ if  }(z_t,A_{t+1})\better_{\theta_t}(z_t,B_{t+1}).
\]
\end{definition}

This concept is well-defined because of part A. in Lemma \ref{thm:separability+continuity} and part (iii) in Lemma \ref{thm:menuprefprops}. Given Lemma \ref{thm:separability+continuity} we have that $\better_{\theta_t}$ is represented by $V_t^{\theta_t}$. 


\paragraph{Axiom: Finiteness of Menu preference} For $\better$ a menu preference over $\A$, collection of menus of acts based on some set of prizes $X$ say that $\better$ satisfies \emph{Finiteness} if there exists $K\in \N$ such that for menu $A$ there exists $B\subset A$ with $|B|\leq K$ and so that $B\sim A$. 

The next Lemma shows that $\better_{\theta_t}$ satisfies Finiteness for menu preferences if there is a DR-SEU representation.

\begin{lemma}\label{thm:finitenessmenu}[Lemma 18 in \cite{fis}]  Assume that the menu preference $\better_{\theta_t}$ satisfies the properties of the Lemmatas \ref{thm:separability+continuity} and \ref{thm:fislm7}. 
For each $\theta_t\in \Theta_t$ there is $K:=K(\theta_t)$ such that $\better_{\theta_t}$ satisfies Finiteness with $K$.
\end{lemma}

\begin{proof}
Fix a separating history $h^t$ for $\theta_t$. Let $SEU_{t+1}(\theta_t) = \pi_{qu}(supp(\psi_{t+1}^{\theta_t}))$. By DR-SEU 1 $K:=|SEU_{t+1}(\theta_t)|<\infty$. We show that this $K$ works for $\theta_t$. 

\textbf{Case 1.} First consider the case $A_{t+1}\in \A_{t+1}^*(h^{t})$ i.e. a menu without ties occuring with positive probability after $h^{t}$. By Lemma 5 in the appendix of the main paper we have for each $(q_{t+1},u_{t+1})\in SEU_{t+1}(\theta_t)$ that $|M(A_{t+1};q_{t+1},u_{t+1})|=1$. Let $B_{t+1} = \cup_{(q_{t+1},u_{t+1})\in SEU_{t+1}(\theta_t)}M(A_{t+1};q_{t+1},u_{t+1})$. Then $|B_{t+1}|\leq K$ and $B_{t+1}\subset A_{t+1}$. It follows that $\rho^{\theta_t}_{t+1}(A_{t+1}\setminus B_{t+1},A_{t+1},s_{t+1}) = 0$ for all $s_{t+1}\in S_{t+1}$. Lemma \ref{thm:fislm7} shows that $V_t^{\theta_t}(A_{t+1})=V_t^{\theta_t}(B_{t+1})$. 

\textbf{Case 2.} Now take any $A_{t+1}\not\in \A_{t+1}^*(h^{t})$. By Lemma 5 in the appendix of the main paper there exists a sequence $A_{t+1}^n\ra A_{t+1}$ such that $A_{t+1}\in\A_{t+1}^*(h^t)$. 

By Case 1 there exists then $B_{t+1}^n\subset A_{t+1}^n$ with $|B_{t+1}^n|\leq K$ and so that $B_{t+1}^n\sim_{h^t}A_{t+1}^n$. Since $|B_{t+1}^n|\leq K$ and restricting to a subsequence if necessary we can assume w.l.o.g. that $B_{t+1}^n\ra^m B_{t+1}$ for some $B_{t+1}\subset A_{t+1}$. Continuity of $\better_{\theta_t}$ implies then that $B_{t+1}\sim_{\theta_t} A_{t+1}$. But note that $B_{t+1}$ can't have more than $K$ elements. Part A. of Lemma \ref{thm:separability+continuity} now implies $V_{t}^{\theta_t}(A_{t+1})=V_{t}^{\theta_t}(B_{t+1})$. 
\end{proof}

\subsection{Sophistication as in Ahn-Sarver EMA 2013}

\begin{lemma}\label{thm:fislm7}[Pendant of Lemma 7 in \cite{fis}]
For any $\theta_t\in \Theta_t$, separating history $h^t$ for $\theta_t$ and $A_{t+1}\subset A'_{t+1}\in \A_{t+1}(h^t)$, the following are equivalent:
\begin{enumerate}
\item $\rho_{t+1}^{\theta_t}(A'_{t+1}\setminus A_{t+1};A'_{t+1},s_{t+1})>0$ for some $s_{t+1}\in S_{t+1}$. 
\item $V_t^{\theta_t}(A'_{t+1})>V_t^{\theta_t}(A_{t+1})$.
\end{enumerate}
\end{lemma}

\begin{proof}
Pick any separating history $h^t$ for $\theta_t$. By DR-SEU 2 and Lemma 8 from the appendix of the main paper we have $\rho_{t+1}^{\theta_t}(A'_{t+1}\setminus A_{t+1};A'_{t+1},s_{t+1})= \rho_{t+1}(A'_{t+1}\setminus A_{t+1};A'_{t+1},s_{t+1}|h^t)$. By Corollary \ref{thm:corollary} and Lemma \ref{thm:separability+continuity}, part A. and Definition \ref{thm:menupref} we have $V_t^{\theta_t}(A'_{t+1})>V_t^{\theta_t}(A_{t+1})$ if and only if $(z_t,A'_{t+1})\sbetter_{h^t}(z_t,A_{t+1})$ for all $z_t$. By the Sophistication Axiom this implies that $V_t^{\theta_t}(A'_{t+1})>V_t^{\theta_t}(A_{t+1})$ if and only if $\rho_{t+1}^{\theta_t}(A'_{t+1}\setminus A_{t+1};A'_{t+1},s_{t+1})>0$ for some $s_{t+1}\in S_{t+1}$, as claimed. 
\end{proof}



\section{Proofs for the Evolving SEU and Gradual Learning Representations}

\subsection{Proof of the representation for Evolving SEU.}

\subsubsection{Sufficiency}

Note that the conditions for the existence of a DR-SEU representation are given. Assume now that we have an Evolving SEU representation till $t'\leq t-1$ and a DR-SEU representation all the way for all $t'\leq T$. We want to show that the Evolving SEU representation holds for $t'=t$ as well. 

Fix any $\theta_t\in\Theta_{t}$ and consider $\Theta_{t+1}(\theta_t) = supp(\psi_{t+1}^{\theta_t})$. Recall from the DR-SEU proof the construction 

\begin{equation*}\label{eq:th1axiom00}
\rho^{\theta_t}_{t+1}(f_{t+1},A_{t+1},s_{t+1}) = \sum_{\theta_{t+1}\in \Theta_{t+1}(\theta_t)}\psi_{t+1}^{\theta_t}(\theta_{t+1})\tau_{\pi_{qu}(\theta_{t+1})}(f_{t+1},A_{t+1}). 
\end{equation*}

 Thus $\rho^{\theta_t}_{t+1}$ has an AS-version R-SEU representation as in Definition \ref{thm:AS-RSEU-SCFdef}.
Since all the elements of $\Theta_{t+1}(\theta_t)$ are non-constant and induce different SEUs and since by Lemma \ref{thm:separability+continuity} we know that $V_t^{\theta_t}$ is non-constant we can find a finite set $Y\subset X_{t+1}$ such that 
\begin{enumerate}
\item $V_t^{\theta_t}$ is non-constant on $\A_{t+1}(Y) = \{B_{t+1}\in \A_{t+1}: \cup_{f_{t+1}\in B_{t+1}}supp(f_{t+1})\subset Y\}$,\footnote{Here, as before we use the notation $supp(f_{t+1}) = \cup_{s_{t+1}\in S_{t+1}}supp(f_{t+1}(s_{t+1}))$.}
\item for each $\theta_{t+1}\in \Theta_{t+1}(\theta_t)$ we have $\pi_{qu}(\theta_{t+1})$ are non-constant on $Y$.
\end{enumerate}
Now observe that Lemmatas \ref{thm:separability+continuity} and \ref{thm:finitenessmenu}, together with the assumed Weak Dominance axiom, imply that the menu preference $\better_{\theta_t}$ defined as in Definition \ref{thm:menupref} satisfies Weak Dominance as well. Theorem \ref{thm:DLR-SEU} now implies that $\better_{\theta_t}$ has a DLR-SEU representation. Moreover, since $\rho_{t+1}^{\theta_t}$ admits an AS-version R-SEU representation, it also admits an AS-version R-SEU representation when restricted to $\A_{t+1}(Y)$. Lemma \ref{thm:fislm7} implies that the pair $(\better_{\theta_t},\rho^{\theta_t}_{t+1})$ satisfies Axioms AS-1 and AS-2 needed for Theorem \ref{thm:dlrrseuthm}. It follows that there is a DLR-SEU representation where menus are constrained to be in $\A_{t+1}(Y)$. In particular, due to essential uniqueness of the DLR-R-SEU representation in Proposition \ref{thm:dlrrseuunique} we have that the DLR-SEU part of the DLR-R-SEU representation can be taken to also use the measure $\psi_{t+1}^{\theta_{t}}$ for the SEUs. 

Finally, we need to extend past $Y$ and show that the DLR-SEU representation for $V_t^{\theta_t}$ holds for all $A_{t+1}\in\A_{t+1}$. Take an arbitrary $A_{t+1}\in \A_{t+1}$ and extend $Y$ to a finite $Y'$ such that $ Y\cup \left(\cup_{f_{t+1}\in A_{t+1}}supp(f_{t+1})\right)\subset Y'$. We get a DLR-R-SEU representation based on $Y'$ as above. Essential uniqueness on $Y$ gives us the same measure over SEU-s and tie-breakers for $Y$. Due to essential uniqueness on $Y'$ the identified measure $\mu$ from Definition \ref{thm:dlrrseudef} must agree again with $\psi_{t+1}^{\theta_t}$ and the tie-breakers for the R-SEU part must agree as well. 

This argument shows that the Evolving SEU representation holds at $t$. Combining this with the inductive hypothesis and the result from the DR-SEU characterization theorem leads to the required sufficiency.  

\subsubsection{Necessity} 

\textbf{Claim.} There exists $g_t, b_t\in \Delta(X_t)$ with $\pi_u(\theta_t)(g_t)>\pi_u(\theta_t)(b_t)$ for all $\theta_{t}\in\Theta_t$. 

\emph{Proof of Claim.}
Since in the representation for every $\theta_{t+1}$ there exists $g_{t+1}(\theta_{t+1}), b_{t+1}(\theta_{t+1})\in\Delta(X_{t+1})$ with $\pi_u(g_{t+1}(\theta_{t+1}))> \pi_u(b_{t+1}(\theta_{t+1}))$ we can set $C'_{t+1} = \{b_{t+1}(\theta_{t+1}), g_{t+1}(\theta_{t+1}): \theta_{t+1}\in\Theta_{t+1}\}$ and for every $\theta_{t+1}$ let $A_{t+1}(\theta_{t+1}) = \{b_{t+1}(\theta_{t+1})\}$. We then have $V_t^{\theta_t}(C'_{t+1})\geq V_t^{\theta_t}(A_{t+1}(\theta'_{t+1}))$ for all $\theta'_{t+1}\in \Theta_{t+1}$ with strict equality at least in the case $\theta_{t+1} = \theta'_{t+1}$. Letting $C_{t+1} = \frac{1}{|\Theta_{t+1}|}\sum_{\theta_{t+1}}A_{t+1}(\theta_{t+1})$ it follows by linearity $V_t^{\theta_{t}}(C'_{t+1})>V_t^{\theta_{t}}(C_{t+1})$ for all $\theta_{t}$. This is sufficient for the statement of the Claim by separability of the $u_t$-s.

\emph{End of Proof of Claim.}

The Claim already implies Non-Degeneracy of $\better_{h^t}$. 

By Lemma \ref{thm:menuprefprops} we have $f_t\better_{h^t}g_t$ if and only if $\pi_{qu}(\theta_t)(f_t)\geq \pi_{qu}(\theta_t)(g_t)$ for all $\theta_t$ consistent with $h^t$. This allows to easily check Separability, Monotonicity and Indifference to Timing and a direct check of the Weak Dominance axiom (recall also Definition \ref{thm:DLR-SEUdef}). 

Continuity is satisfied with a very similar proof to the one of Theorem 2 in \cite{fis} by noting that for each $f_t\in \F_t$ and $h^t$ we have 
\[
\{g_t: g_t\better_{h^t} f_t\} = \cap_{\theta_t\text{ consistent with }h^t}\{g_t: g_t\better_{\theta_t} f_t\},\quad \{g_t: g_t\worse_{h^t} f_t\} = \cap_{\theta_t\text{ consistent with }h^t}\{g_t: g_t\worse_{\theta_t} f_t\}.
\]

We show that Sophistication is satisfied. Consider any $t\leq T-1, h^t, z_t$ and $A_{t+1}\subset A'_{t+1}\in \A^*(h^{t+1})$. Since $A'_{t+1}\in \A^*(h^{t+1})$ Lemma \ref{thm:menuprefprops} 
implies that $\rho_{t+1}(A'_{t+1}\setminus A_{t+1}; A'_{t+1}|h^t)>0$ holds if and only if for some $\theta_t$ consistent with $h^t$ we have \\(!) $\max_{f_{t+1}\in A'_{t+1}}\pi_{qu}(\theta_{t+1})(f_{t+1})>\max_{f_{t+1}\in A_{t+1}}\pi_{qu}(\theta_{t+1})(f_{t+1})$\\ for some $\theta_{t+1}$, $\pi_{qu}(\theta_{t+1})\in \psi_{t+1}^{\theta_t}$. By the representation, (!) is equivalent to $V_t^{\theta_t}(A'_{t+1})>V_t^{\theta_t}(A_{t+1})$. By Lemma \ref{thm:separability+continuity}, part A. this means $(z_t,A_{t+1})\not\better_{h^t}(z_t,A'_{t+1})$. By Monotonicity this is equivalent to $(z_t,A'_{t+1})\sbetter_{h^t}(z_t,A_{t+1})$.

\subsection{Proof of the representation for Gradual Learning.}

We assume throughout that we have an Evolving SEU representation. 

We normalize the instantaneous Bernoulli utilities $v_t$ to satisfy $$\sum_{z\in Z} \pi_{v}(\theta_t)(z) = 0$$ for all $\theta_t\in\Theta_t$.\footnote{Recall that $\pi_v$ gives the $v_t$ corresponding to $\pi_u(\theta_t)$ from the Evolving SEU representation (see Definition 7 in the main body of the paper).}

\subsubsection{Sufficiency.} 

We show that equation (16) in the appendix of the main paper is satisfied. This gives then the result. 


\begin{lemma}\label{thm:helpgl1}

For any $t=0,\dots,T-1$ and $\theta_t\in\Theta_t$, there exists $l,m\in \Delta(Z)$ such that $\pi_v(\theta_t)(l)\neq \pi_v(\theta_t)(m)$.
\end{lemma}

\begin{proof}
Consider any $t=0,\dots,T-1, \theta_t\in\Theta_t$ and separating history $h^t$ for $\theta_t$. Non-Degeneracy gives existence of $l,m,n\in\Delta(Z)$ such that $(l,n,\dots,n)\not\sim_{h^t}(m,n,\dots,n)$. By part (iii) of Lemma \ref{thm:menuprefprops} we get $\pi_u(\theta_t)(l,n,\dots,n)\neq \pi_u(\theta_t)(m,n,\dots,n)$, whence it follows from the Evolving SEU representation that $\pi_v(\theta_t)(l)\neq \pi_v(\theta_t)(m)$.
\end{proof}

For any $t=0,\dots,T-1$ and $\theta_t\in\Theta_t$ and $l\in \Delta(Z)$, let 
\[
\E[v_{t+1}(l)|\theta_t]:= \sum_{\theta_{t+1}}\psi_{t+1}^{\theta_t}(\theta_{t+1})(\pi_v(\theta_{t+1}))(l)
\]
denote the expected utility in period $t+1$ of lottery $l$ if current state is $\theta_t$. Axiom Stationary Preference for Lotteries, the fact that we have an Evolving SEU representation and part (iii) of Lemma \ref{thm:menuprefprops} imply that $\pi_v(\theta_t)$ and $\E[v_{t+1}(l)|\theta_t]$ induce the same preference over $\Delta(Z)$:

\begin{lemma}\label{thm:helpgl2}\footnote{Its proof is word for word as that of Lemma 10 in \cite{fis}.}
For all $l,m\in \Delta(Z), t=0,\dots,T-1$ and $\theta_t\in\Theta_t$ we have 
\[
\E[v_{t+1}(l)|\theta_t]>\E[v_{t+1}(m)|\theta_t]\quad\equivalent\quad \pi_v(\theta_t)(l)>\pi_v(\theta_t)(m). 
\]
\end{lemma}

Constant Intertemporal Trade-off allows us to obtain a time-invariant and non-random discount factor $\delta>0$. 

\begin{lemma}\footnote{The proof needs only very minor adaptations to the proof of Lemma 11 in \cite{fis} so we just give a proof-sketch. }
There exists $\delta\in (0,1)$ such that for all $t=0,\dots,T-1$ and $\theta_t\in\Theta_t$, we have $\pi_v(\theta_t) = \frac{1}{\delta}\E[\pi_v(\theta_{t+1})|\theta_t]$. 
\end{lemma}

\begin{proof}[Proof Sketch] Lemma \ref{thm:helpgl2} and the normalization required in subsection B.3 in the appendix of the main paper show that $\pi_v(\theta_t) = \frac{1}{\delta(\theta_t)}\E[\pi_v(\theta_{t+1})|\theta_t]$. Taking some $\theta'_t\neq \theta_t$ and separating histories for both $\theta_t$ and $\theta'_t$ and using Constant Intertemporal Trade-off as well as Lemma \ref{thm:helpgl1} one repeats the proof of Lemma 11 of \cite{fis} to show $\delta(\theta) = \delta$ for all $\theta_t$. Finally, the impatience axiom trivially yields that $\delta\in (0,1)$. 

\end{proof}

\subsubsection{Necessity}

This follows with very minor adaptations the proof of Necessity in Section D.2. of \cite{fis}: once one looks only at menus from $\A^c_t,t=0,\dots,T$ (i.e. only containing constant acts), beliefs are not relevant in the Evolving SEU representation and the proofs boil down directly to the one from \cite{fis}.

\section{Equivalence Between Filtration Representations and Ahn-Sarver-based Representations and Proof of Uniqueness}\label{sec:equivalence}

This is an adaptation of the proof of Proposition 5 in the online supplement of \cite{fis}. It establishes a one-to-one correspondence between the (partitional) $\theta_t$-s from the AS-representations and the cells of the Filtration $\mathcal{F}_t,t=0,\dots,T$.  

We note here that the one-to-one correspondence between the two representation types (AS- and filtration form) and the uniqueness results for the AS-version representations yield also the uniqueness results for the Filtration version of the representations. 

\paragraph{From AS-version DR-SEU to Filtration-based DR-SEU.}

Consider the space $G = \prod_{t=0}^T \Theta_t\times (\Delta(S_t)\times \R^{X_t})$ and let $\hat\Omega = \{(\theta_0,p_0,v_0; \dots; \theta_T,p_T,v_T)\in G: \prod_{k=0}^t\psi_k^{\theta_{k-1}}(\theta_k)>0\}$. Let $\mathcal{\hat F}^*$ be the restriction to $\hat\Omega$ of the product sigma-Algebra of the discrete sigma-algebra on $\prod_{t=0}^TS_t$ and the product Borel sigma-algebra on $\Delta(S_t)\times \R^{X_t}$. For each $K=\{\{\theta_0\},K_0,\dots,\{\theta_T\},K_T\}\in \mathcal{\hat F}^*$ we set $\hat\mu(K) = \prod_{t=0}^T\psi_t^{\theta_{t-1}}(\theta_t)\tau_{\pi_{qu}(\theta_t)}(K_t)$ and extend $\hat\mu$ to a probability measure in the natural way.

Let $\Pi_t$ be the finite partition of $\hat\Omega$ whose cells are all cylinders $C(\theta_0,\dots,\theta_t)=\{\hat\omega: proj_{\Theta_0\times\dots\Theta_t}(\hat\omega) = (\theta_0,\dots,\theta_t)\}$. Let $\mathcal{\hat F}_t$ be the sigma-algebra generated by $\Pi_t$. By definition of $\hat\Omega$ we have $\hat\mu(\mathcal{\hat F}_t(\hat\omega))>0$ for all $\hat\omega\in\hat\Omega$. Also by definition we have $\mathcal{\hat F}_t(\hat\omega)=\cup_{\hat\omega'\in\mathcal{\hat F}}\mathcal{\hat F}_{t+1}(\hat\omega)$, so $(\mathcal{\hat F})_{0\leq t\leq T}$ is a Filtration. 

Define $(\hat q_t,\hat u_t):\hat\Omega\ra \Delta(S_t)\times\R^{X_t}$  as the projection of each $\theta_t$ into $\Delta(S_t)\times\R^{X_t}$, whenever $\theta_t$ appears in $(\theta_0,\dots,\theta_t)$ and $\hat\omega\in C(\theta_0,\dots,\theta_t)$. Note that $(\hat q_t,\hat u_t)$ is adapted to $\mathcal{\hat F}_t,t\leq T$ and that it is always a non-constant SEU. Finally, if $\mathcal{\hat F}_{t-1}(\hat\omega) = \mathcal{\hat F}_{t-1}(\hat\omega')$ and $\mathcal{\hat F}_{t}(\hat\omega) \neq \mathcal{\hat F}_{t}(\hat\omega')$ then $proj_{\Theta_{t-1}}(\hat\omega) = proj_{\Theta_{t-1}}(\hat\omega') = \theta_{t-1}$ and $proj_{\Theta_{t}}(\hat\omega) = \theta_t \neq proj_{\Theta_{t}}(\hat\omega') = \theta'_{t}$ for some $\theta_{t-1}\in\Theta_{t-1}$ and $\theta_t,\theta'_t\in supp(\psi_t^{\theta_{t-1}})$. 

Define finally for the tie-breaker process of DR-SEU  the $(\hat p_t,\hat v_t)$ as the projection of $\hat\Omega$ into $\Delta(S_t)\times\R^{X_t}$. 

From here the proof follows the same steps as the `if direction' in Appendix G.1 of \cite{fis} if one makes the following replacements/identifications: $s_t\ra\theta_t$, $\hat W_t\ra(\hat p_t,\hat v_t)$, $\mu\ra\psi$, $\tau_{s_t}\ra\tau_{\pi_{qu}(\theta_t)}$, $\hat U_t\ra (\hat q_t,\hat u_t)$ and finally $p_t\ra f_t$.

We also note here that the preference-based property of the tie-breaking process follows directly from the definition of $\hat\mu$.

We finally note that $\hat\mu(\cdot|\hat q_t,\hat u_t)(\hat\omega)  = \psi_{t}^{\theta_{t-1}}( \hat q_t,\hat u_t,\cdot)$, whenever $\hat\omega\in C(\theta_0,\dots,\theta_{t-1},\theta_t)$. From this, the properties of CIB or NUC follow immediately, as required by the respective representations.

\paragraph{From Filtration-based DR-SEU to AS-version DR-SEU.} For each $t$, let $\Theta_t = \{\mathcal{F}_t(\omega):\omega\in\Omega\}$ denote the partition generating $\mathcal{F}_t$ (finite, since $(\mathcal{F}_t)$ is simple). I.e. we are identifying each element of the filtration with a partition cell. Each $\psi_{t+1}^{\theta_t}$ is defined as the one-step-ahead conditional of $\mu$, i.e. $\psi_{t+1}^{\theta_t} (q_{t+1},u_{t+1},s_{t+1}):= \mu(\mathcal{F}_{t+1}|\mathcal{F}_t)$ where $\mathcal{F}_{t+1}\in \Theta_{t+1}$ corresponds to $\theta_{t+1} = (q_{t+1},u_{t+1},s_{t+1})$ and $\mathcal{F}_{t}\in \Theta_{t}$ corresponds to $\theta_{t}$. For each $\theta_t\in\Theta_t$ define $(\hat q_t,\hat u_t,\hat s_t)(\theta_t) = (\hat q_t,\hat u_t,\hat s_t)(\omega)$ whenever $\omega\in\theta_t$ (and recall that each $\theta_t$ corresponds to some $\mathcal{F}_t\in\Theta_t$). Finally, using the tie-breaker sequence $(p_t,v_t)$ of the Filtration-based DR-SEU representation for any Borel set $B_t\subset \Delta(S_t)\times\R^{X_t}$ we define $\tau_{\pi_{qu}(\theta_t)}(B_t):=\mu(\{(p_t,v_t)\in B_t\}|\mathcal{F}_t)$, where $\mathcal{F}_t$ corresponds to $\theta_t$. The definition is independent of $\theta_t$ as long as $\pi_{qu}(\theta_t) = \pi_{qu}(\theta'_t)$ because of the preference-based property of $\mu$. 

The properties of CIB or NUS in the AS-version form follow now directly from the respective properties in the Filtration-based form and from the definition of $\psi$. 

From here, the proof follows exactly the proof of \cite{fis} once we make the following replacements in text: $\hat\mu\ra \psi$, $s_t\ra\theta_t$, $\hat U_{s_t}\ra \pi_{qu}(\theta_t)$ and $U_t(\omega)\ra (q_t,u_t)(\omega)$. 

\begin{remark}\label{thm:remkequiv}
Note that the case $T=0$ gives the equivalence for the case of static aSCFs between the R-SEU Definition and the AS-based R-SEU version of the representation.
\end{remark}

\paragraph{From Filtration-based to AS-version and back for Evolving Utility.}

Suppose we have an AS-version Evolving SEU representation. Per above we also have then a Filtration-based DR-SEU representation already. Define $\hat v_t:\hat\Omega\ra \R^Z$ for each $t$ by $\hat v_t(\hat\omega) = v_{t}^{\pi_{qu}(\theta_t)}$ whenever $\theta_t$ corresponds to $\hat\omega$, in the sense that $proj_{\Theta_t}(\hat\omega) = \theta_t$. In this way, the process $\hat v_t$ is $\mathcal{F}_t$-adapted. Moreover, for each $\hat\omega\in \Omega$ we have $\hat u_T (\hat\omega)= \pi_u(\theta_T)= \pi_v(\theta_T) = v_T^{\pi_{qu}(\theta_T)} = \hat v_T (\hat\omega)$ and for each $t\leq T-1$ and $f_t\in\F_t, s_t\in S_t$ we have 

\begin{align*}
\hat u_t(\hat\omega)(f_t(s_t))& = u_t(f_t(s_t)) = v_t(f_t^Z(s_t)) +\delta \int\max_{f_{t+1}\in A_{t+1}}\left(q_{t+1}\cdot u_{t+1}\right)(f_{t+1})d\psi_{t+1}^{\theta_t}(q_{t+1},u_{t+1})\\
& = \hat v_t(\hat\omega)(f_t^Z(s_t)) +\delta\sum_{(q_{t+1},u_{t+1})} \hat\mu\left(\pi_{qu}(\theta_{t+1}) = (q_{t+1},u_{t+1})|\theta_t\right)\cdot \max_{f_{t+1}\in A_{t+1}}\left(q_{t+1}\cdot u_{t+1}\right)(f_{t+1})\\
&=  \hat v_t(\hat\omega)(f_t^Z(s_t)) +\delta V_t^{\theta_t}(A_{t+1})
\end{align*}

This gives the Filtration-based Evolving SEU representation. 

Suppose now that we have a Filtration-based Evolving SEU representation instead. We know we have an AS-version DR-SEU representation. Here we just have to reverse the argument from above. 

\begin{equation}\label{eq:evseu-as-1}
\E_{q_t}[u_t(f_t)] = \E_{s_t\sim q_t}[u_t(f_t(s_t))] = \E_{s_t\sim q_t}[v_t(f_t^Z(s_t))] + \delta V_{t}^{\pi_{qu}(\theta_t)}(f_t^A). 
\end{equation}
Here we can define $V_{t}^{\pi_{qu}(\theta_t)}(f_t^A)$ by first defining 
\begin{equation*}
V_t^{\theta_t}(A_{t+1})= \int\max_{f_{t+1}\in A_{t+1}}\left(q_{t+1}\cdot u_{t+1}\right)(f_{t+1})d\psi_{t+1}^{\theta_t}(q_{t+1},u_{t+1}).
\end{equation*}

\paragraph{From Filtration-based to AS-version and back for Gradual Learning}

Suppose we have a AS-version gradual learning representation. Take the corresponding Filtration-based Evolving SEU representation and define $\hat \delta = \delta$. Note that for each $\hat\omega = (\theta_0,p_0,v_0;\dots; \theta_T,p_T,v_T)$ and $t\leq T-1$ we have 

\[
\hat v_t(\hat\omega) = v_t^{\pi_{qu}(\theta_t)} = \frac{1}{\delta}\sum_{(q_{t+1},u_{t+1})}\psi^{\theta_t}_{t+1}(q_{t+1},u_{t+1})\cdot \pi_v(u_{t+1}) = \frac{1}{\hat\delta}\E[\hat v_{t+1}|\mathcal{\hat F}_t(\hat\omega)].
\]

Iterating we are led to $\hat v_t(\hat\omega) = \hat\delta^{t-T}\E[\hat v_T|\mathcal{\hat F}_t(\hat\omega)] = \hat\delta^{t-T}\E[\hat u_T|\mathcal{\hat F}_t(\hat\omega)]$. By replacing $\hat u_t$ with $\hat u_t' = \hat\delta^{T-t}\hat u_t$ for each $t$, going to the AS-version DR-SEU representation, using uniqueness result there (Proposition \ref{thm:drseuuniqueAS}) and then back to the Filtration-based representation we get a new DR-SEU representation with $\hat u_t'$ instead of $\hat u_t$. The rest of the proof of the Gradual Learning representation is identical to the one in \cite{fis}. 

Suppose that we have a Filtration-based gradual learning representation. Let $u'_t = \delta^{t-T} u_t$ for all $t$. This is again a DR-SEU representation by Proposition \ref{thm:drseuuniqueAS}. Moreover, let $v'_t = \delta^{t-T}v_t$, where $v_t = \E[v_T|\mathcal{F}_t(\omega)]$. From here the proof follows identically as in \cite{fis} with the obvious replacements in notation.

\subsection{Uniqueness for AS-versions of the Representations}\label{sec:uniquenessAS}

Together with the equivalence results between AS-version representations and filtration-based representations, this subsection also gives a proof of the uniqueness of filtration-based representations: Proposition 4 in the main body of the paper is proven by combining the results in this subsection with those of section \ref{sec:equivalence}. 


\begin{proposition}\label{thm:drseuuniqueAS}
The DR-SEU representation is unique. In particular, if we have two DR-SEU representations, then there exists a finite sequence of bijective mappings $t\leq T$ $\vartheta_t:\Theta_t\ra\tilde\Theta_t$ such that 
\[
(A)\quad \vartheta_t(q_t,u_t,s_t) = (q_t,\tilde u_t,s_t),\text{ with }u_t\approx \tilde u_t,\footnote{In the following we sometime omit the time subscript for $\vartheta$ for ease of notation.}
\]
and 
\[
(B)\quad \tilde \psi_t^{\vartheta(\theta_{t-1})}(\vartheta(\theta_t)) = \psi_t^{\theta_{t-1}}(\theta_t),\quad\text{ for all }t, \theta_t\in \Theta_t.
\]
\end{proposition}

\begin{proof}
We prove the existence of the mappings $\vartheta_t$ by induction.

\emph{Induction start: $t=0$.} This is just uniqueness for the representation of a static aSCF, which is given by Proposition 7 from the appendix of the main paper. 

\emph{Induction step: take $t\geq 1$. We want to show $t'< t\quad \imply \quad t$.} Assume that we have uniqueness for all indices $t'<t$ for some $t$. This gives the mappings $\vartheta_{t'}$ for all $t'<t$. We next show that $\Theta'_t$ has indeed the form needed and that $\vartheta_t$ exists. Fix some $\theta_{t-1}$ and its image $\theta'_{t-1}$ under $\vartheta_{t-1}$. 

\textbf{Claim.} For all $\theta_t= (q_t,u_t,s_t)\in \Theta_t$ there exists a unique $\theta'_t= (q'_t,u'_t,s'_t)\in \Theta'_t$ with $(q_t,s_t) = (q'_t,s'_t)$ and $u_t\approx u'_t$.  

\begin{proof}
[Proof of Claim.] 

\emph{Uniqueness.} 
If $\theta'_t$ exists then it is clearly unique because by construction and the DREU-1 property:  whenever $\pi_{qu}(\theta'_t) = \pi_{qu}(\theta''_t)$ for two $\theta'_t, \theta''_t\in \Theta'_t$ we must have $\pi_{s}(\theta'_t) \neq \pi_{s}(\theta''_t)$. 

\emph{Existence.} Fix the unique predecessor $(\theta_0,\dots,\theta_{t-1})$ of $\theta_t$ and consider both $\psi_t^{\theta_{t-1}}$ as well as $\tilde\psi_t^{\vartheta_{t-1}(\theta_{t-1})}$. Due to the way the mappings $\vartheta_{t'}$ for $t'<t$ are constructed one can easily see that any separating history for $\theta_{t'}$ under one representation is also a separating history for $\vartheta_{t'}(\theta_{t'})$ under the other representation. 

Pick then such a separating history $h^{t-1}$ for $\theta_{t-1}$ and $\vartheta_{t-1}(\theta_{t-1})$.\footnote{Recall also Definition 23 in the appendix of the main paper.}
We want to show that the subjective SEU supports are equal up to positive affine transformation of the Bernoulli utilities and that the probabilities for each element of the type $(q_t,u_t,s_t)$ are the same. 

Regarding the claim about the support, assume for the sake of contradiction that there exists $(q'_t,u'_t)\in \pi_{qu}\left(supp(\tilde\psi_t^{\vartheta_{t-1}(\theta_{t-1})})\right)$ to which no element in $\pi_{qu}\left(supp(\psi_t^{\theta_{t-1}})\right)$ corresponds. We use Lemma 1 in main body of paper to construct a menu $B_t$ with an element $\tilde f_t$ which separates $(q'_t,u'_t)$ from $\pi_{qu}\left(supp(\psi_t^{\theta_{t-1}})\right)$. It holds 

\[
\tau_{(q'_t,u'_t)}(\tilde f_t, B_t) = 1> 0 = \tau_{(q_t,u_t)}(\tilde f_t, B_t),\text{ for all }(q_t,u_t)\in \pi_{qu}\left(supp(\psi_t^{\theta_{t-1}})\right). 
\]
W.l.o.g. we can assume then, that $h^{t-1}$ chosen above satisfies $h^{t-1}\in \h_{t-1}(B_t)$ (we can just perturb $h^{t-1}$ with a suitable constant act giving in each state a deterministic prize to achieve this). It follows from the DR-SEU 2 property of both representations that for any $s_t\in supp(q'_t)$ we have 
\begin{align*}
&(\text{from the second representation})\quad 0<\rho_t(\tilde f_t,B_t,s_t|h^{t-1})\\& = \rho_t(\tilde f_t,B_t,s_t|h^{t-1}) = 0\quad (\text{from the first representation}).
\end{align*}
This is a contradiction and thus $\pi_{qu}\left(supp\left(\psi_t^{\theta_{t-1}}\right)\right)$ and $\pi_{qu}\left(supp\left(\tilde\psi_t^{\vartheta_{t-1}(\theta_{t-1})}\right)\right)$ are the same, up to positive affine transformations of the Bernoulli utilities $u_t$. In particular, the beliefs $q_t,q'_t$ that can occur under the two probability measures are the same. W.l.o.g. we normalize the Bernoulli utilities corresponding to same SEUs to be the same in both representations. By taking a separating menu $B_t$ for the whole of set of SEUs happening with positive probability under $\psi_t^{\theta_{t-1}}$ and $\tilde\psi_t^{\vartheta_{t-1}(\theta_{t-1})}$ and using the induction hypothesis on the separating history $h^{t-1}$ for  $\theta_{t-1}$ (and $\vartheta_{t-1}(\theta_{t-1})$) we arrive with help of DR-SEU 2 at
\[
\psi_t^{\theta_{t-1}}(q_t,u_t, s_t) = \tilde \psi_t^{\vartheta_{t-1}(\theta_{t-1})}(q_t,u_t,s_t),\text{ for all }(q_t,u_t,s_t).
\]
This establishes existence of the respective $\vartheta_{t}$. 
\end{proof}

\end{proof}


\begin{proposition}\label{thm:evseuuniqueAS}
An Evolving SEU representation is unique. In particular, if we have two Evolving - SEU representations then in addition to the finite sequence of bijective mappings from Proposition \ref{thm:drseuuniqueAS} the following properties hold.
\begin{enumerate}
\item For the positive constants $\alpha_t(\theta_t)$ used for the equivalence $\pi_u(\theta_t) \approx \pi_u(\vartheta_t(\theta_t))$ it holds $\alpha_t = \alpha_0\left(\frac{\tilde\delta}{\delta}\right)^t$ for $t=0,\dots,T$. 
\item Let $\beta_t(\theta_t)$ be the intercepts of the positive affine transformations from Proposition \ref{thm:drseuuniqueAS}. $\pi_u(\theta_t) = \alpha_t\pi_u(\vartheta_t(\theta_t)) + \gamma_t(\theta_t)$, where the functions $\gamma_t,t=0,\dots,T$ are connected through the relations $\gamma_T(\theta_T) = \beta_T(\theta_T)$ and $\gamma_t(\theta_t) = \beta_t(\theta_t) -\delta \E_{\theta_{t+1}}[\beta_{t+1}(\theta_{t+1})|\theta_t]$, if $t\leq T-1$. 
\end{enumerate}
\end{proposition}

\begin{proof}
The proof is a straightforward adaptation of the corresponding part in Proposition 1 of \cite{fis}. Just recall that

(1) the evolution of beliefs $q_t$ is already pinned down uniquely by the Proposition \ref{thm:drseuuniqueAS}, 

(2) Evolving SEU falls back to their Evolving Utility model when acts are restricted to be constant throughout. 

These two facts, and the argument in their proof can be repeated to yield the result.
\end{proof}


\begin{proposition}\label{thm:gluniqueAS}
Under Condition $1$ a Gradual Learning representation is unique. In particular, in addition to the uniqueness relations from Proposition \ref{thm:evseuuniqueAS} the following relations between two Gradual Learning representations hold true. 
\begin{enumerate}
\item $\delta = \hat\delta$.
\item $\beta_t(\theta_t) = \frac{1-\delta^{T-t+1}}{1-\delta}\E_{\theta_T}[\beta_T(\theta_T)|\theta_t]$. 
\end{enumerate}
\end{proposition}

\begin{proof}
By Proposition \ref{thm:evseuuniqueAS} we are given the equivalence in terms of an Evolving SEU representation. So we just need to prove that A. and B. in the statement hold true.

A. Take a $\theta_0\in \Theta_0$ and consider $\vartheta_0(\theta_0)$ as well as $h^0$ a separating history for $\theta_0$. As mentioned in the proof of Proposition \ref{thm:drseuuniqueAS}, this is also a separating history for $\vartheta_0(\theta_0)$. By Condition 1 applied to $h^0$ we have the existence of $l,m,n\in \Delta(Z)$ such that $(l,n,\dots,n)\sbetter_{h^0}(m,n,\dots,n)$. From Corollary \ref{thm:corollary} we have that $\pi_u(\theta_0)(l,n\dots,n)> \pi_u(\theta_0)(m,n\dots,n)$. Using equation (16) in the appendix of the main paper iteratively and defining $v_0^{\theta_0} = \E[v_T^{\theta_T}|\theta_0]$ we get for a general sequence of constant consumption lotteries $(l^0,\dots,l^T)$, that it holds $\pi_u(\theta_0)(l^0,\dots,l^T) = \sum_{k=0}^T\delta^kv_0^{\theta_0}(l^k)$. It follows that $v_0^{\theta_0}(l)>v_0^{\theta_0}(m)$ and that 
\[
\pi_u(\theta_0)(l,m,n\dots,n)-\pi_u(\theta_0)(\eta l+(1-\eta)m,\eta l+(1-\eta)m,n\dots,n) =0\text{ iff }\eta=\frac{1}{1+\delta}. 
\]
We can do the same steps for $\vartheta(\theta_0)$ to arrive at the following.
\[
\pi_u(\vartheta(\theta_0))(l,m,n\dots,n)-\pi_u(\vartheta(\theta_0))(\eta l+(1-\eta)m,\eta l+(1-\eta)m,n\dots,n) =0\text{ iff }\eta=\frac{1}{1+\hat\delta}. 
\]
It follows that $\delta=\hat\delta$. 

2. Given $\delta = \hat\delta$, B. in Proposition \ref{thm:evseuuniqueAS} and equation (16) in the appendix of the main paper for both GL representations, we can apply B. of Proposition \ref{thm:evseuuniqueAS} inductively to arrive at the required B. statement here.  
\end{proof}

\section{Extending Option Value Theory to Subjective Expected Utility}\label{sec:ovseu}

In Section \ref{sec:stochmenupref} we define a menu preference derived from stochastic choice over future decision problems. At the end of a period $t$ the agent knows the realization of the state $(q_t,u_t,s_t)$, her choice $f_t\in A_t$ as well as the realization of the menu $A_{t+1}$ from $f_t^A(s_t)$. Given that we allow for stochastic taste, a generalization of the main result of \cite{dlst} to SEU is needed for the proof of part F. of Lemma \ref{thm:separability+continuity}. 


We start by considering a finite prize space $X$. While one can extend the proof to general spaces $X$ by standard methods as in \cite{ks}, the proof in \cite{fis} can be adapted to show that this additional argument is not needed for our setting. 

Consider acts $f:S\ra\Delta(X)$ and the set of such acts $\F$. Since $S$ is finite the set $\F$ is isomorphic to a (convex) polytope in $\R^{|S|\cdot|X|}$, i.e. it can be prescribed by finitely many linear inequalities. Note that the same is true for $conv(A)$, the convexification of $A$, an  arbitrary (finite) menu of acts from $\F$. Denote the collection of menus as usual by $\A$. Assume throughout that a preference $\better$ is given on $\A$. 

\begin{definition}\label{thm:DLR-SEUdef} Say that $\better$ over $\A$ has an DLR-SEU representation (or Option Value-SEU) representation if there exists a probability measure $\mu$ of finite support on $\Delta(S)\times \R^X$ such that the following functional represents $\better$.

\[
V(A) = \int_{\Delta(S)\times \R^X} \max_{f\in A}\sum_{s\in S}q(s)u(f(s))d\mu(q,u).
\]

\end{definition}

We impose the following Axiom for $\better$ (a collection of by now classical axioms). State Independence will be added later.

\paragraph{Axiom 0: Classical Menu Preference Axioms.}
\begin{enumerate}
\item $\better$ is a weak order.
\item $\better$ is continuous.
\item $\better$ satisfies Independence, i.e. if $A\sbetter B$, then for any $C$ and $\alpha\in(0,1)$ we have $\alpha A +(1-\alpha)C\sbetter \alpha B +(1-\alpha)C$.
\item $\better$ satisfies Monotonicity/Option Value:  $A\subset B$ implies $B\better A$.
\item Non-Triviality: there exists some $p,q\in\Delta(X)$ such that $\{p\}\sbetter\{q\}$.\footnote{Note that this can be relaxed to the usual version as found e.g. in \cite{as}. We work with this version for simplicity.} 
\item Finiteness: there exists $K\in \N$ such that for any $A$, there exists $B\subset A, |B|\leq K$ with $A\sim B$. 
\end{enumerate}
\vspace{4mm}
Consider now the space of normalized state-dependent Bernoulli utilities.
\[
\mathcal{W} = \{w:S\ra \R^X: \sum_{s,x}w(s)(x) = 0, \sum_{s,x}|w(s)(x)|^2 = |S|\}.
\]
Define also the space of state-independent Bernoulli utilities, which can be identified canonically with a \emph{strict} subspace of $\mathcal{W}$ as follows.

\[
\mathcal{U} = \{u\in\R^X: \sum_{x}u(x) = 0, \sum_{x}|u(x)|^2 = 1\}.
\]
Note that both $\W$ and $\mathcal{U}$ do not allow constant preferences. 

The first Lemma shows that similarly to section \ref{sec:rseuapp} (see Lemma \ref{thm:lu1} there) the classical axioms ensure the existence of a (unique) representation with normalized \emph{state-dependent} utilities.

\begin{lemma}\label{thm:as1}
Let $\better$ satisfy parts A-E of Axiom 0. Then there exists a unique probability measure over $\W$ such that V defined as 

\[
V(A) = \int_{\W}\left[\max_{f\in A}\sum_{s,x}w(s)(x)f(s)(x)\right]d\mu(w), 
\]
represents $\better$. Moreover, $\mu$ has finite support if and only if Finiteness (part F. of Axiom 0) is additionally satisfied.  
\end{lemma}

\begin{proof}
\emph{Step 1.}

We start just as in the proof of Lemma \ref{thm:lu1}. Let $W$ be the affine hull of $\F$ in $\R^{|X||S|}$ with dimension $m$ and consider $\Delta$ be the probability simplex in $W$ as well as $\{w_1,\dots,w_m\}$ an orthonormal basis of $W$.
Consider the mapping $T:\F\ra \Delta$ given by 
\[
T(f)_{i} = \lambda\left[f\cdot (w_i-\frac{1}{m}\sum_{j}w_j)\right] +\frac{1}{m}.
\]
Note that by definition of acts $f\cdot w_i$ is a number in $[0,1]$ always. Also, for all $\lambda>0$ small enough\footnote{We assume this in the following.} we have $T(f)\geq 0$ and by definition also $T(f)\in \Delta$. Note that this transformation is linear and thus in particular continuous (since the vector spaces involved are finite-dimensional). Note also that $T$ is injective (one-to-one) and that $im(T)$ is a polytope as well. 

In the following we modify the construction in the proof of Lemma S.2 in the Supplement of \cite{lu} to our set up of menus. 

Define a preference $\better_T$ on menus from $im(T)$ by 
\[
T(A)\better_T T(B)\quad\text{ iff }\quad A\better B.
\]
This is well-defined due to injectivity of $T$ and it is easy to see it satisfies all of the parts A-F from Axiom 0.\footnote{Here again injectivity of $T$ is crucial for the proof of Monotonicity, whereas linearity of $T$ is crucial for Independence.} In particular, we have an Option Value representation within $im(T)$. Now, we can either work with this smaller choice space going forward or we can just extend this preference to all of menus from $\Delta$ with the same trick as in Lemma \ref{thm:lu1}. Let's take the latter route. Given is again the projection $P$ of $W$ into the affine hull $W'$ of $im(T)$. For any finite menus $A,B\subset \Delta$ take a $p^*\in \Delta\cap W'$ and $\alpha\in (0,1)$ such that $aP(A\cup B)+(1-a)\{p^*\}\subset im(T)$. Then we can define 

\[
A\better_{\Delta} B\quad\text{ iff }\quad aP(A)+(1-a)\{p^*\}\better_{T} aP(B)+(1-a)\{p^*\}.
\]
\textbf{Claim.} $\better_{
\Delta}$ is well-defined. 

\begin{proof}
[Proof of Claim.] Consider two pairs $(p_1,a_1)$ and $(p_2,a_2)$ with $p_i\in \Delta\cap W'$ satisfying $a_iP(A\cup B)+(1-a)\{p_i\}\subset im(T)$ for $i=1,2$. Then one can see that the sides of the two polytopes in $im(T)$ defined by $a_iP(A\cup B)+(1-a_i)\{p_i\}$ are pairwise parallel to each other (recall $T$ and $P$ are affine). With a similar argument as in the pictures in \cite{gp} (see pg. 130 there) this shows that $\better_{T}$ ranks $\{a_1P(A)+(1-a_1)\{p_1\}, a_1P(B)+(1-a_1)\{p_1\}\}$ and $\{a_2(A+(1-a_2)\{p_2\}, a_2(B+(1-a_2)\{p_2\}\}$ the same way. For more details: it is trivial to see that $T$ preserves angles. Therefore, it sends parallel lines to parallel lines. In particular, the unique pre-images of the polytopes $\{a_1P(A)+(1-a_1)\{p_1\}, a_1P(B)+(1-a_1)\{p_1\},a_2P(A)+(1-a_2)\{p_2\}, a_2P(B)+(1-a_2)\{p_2\}\}$ under $P\circ T$ in the affine hull of $\F$ have respectively parallel sides. Using now Independence for the original preference $\better$ it is easy to conclude the proof by using the fact that $P\circ T$ is affine. 
\end{proof}
\vspace{2mm}
\textbf{Claim.} $\better_{\Delta}$ satisfies on $\Delta$ Weak Order, Continuity, Independence, Monotonicity, Non-Triviality and Finiteness, i.e. all of the axioms from section S1 in \cite{as}. 
\begin{proof}
[Proof of Claim] This is routine. The `harder' part is showing Independence. This follows again from linearity of the map $T$ and the fact that $T$ preserves angles. 
\end{proof}

Theorem S1 in \cite{as} yields now a functional representing $\better_{\Delta}$. This functional is easily translated into a representing functional for $\better$. This holds because of injectivity of $T$.

Finally, uniqueness of the measure comes from classical results in \cite{dlr} and the fact that we focused on Bernoulli utilities from $\mathcal{W}$, which are normalized twice. 
\end{proof}

Note that for each $w\in supp(\mu)$ from the representation in Lemma \ref{thm:as1} we can interpret each $w(s)(\cdot)$ as a Bernoulli utility, i.e. as an element from $\R^X$.

We now add an axiom which ensures that every element from the support of $\mu$ in the representation of Lemma \ref{thm:as1} has 
\begin{equation}\label{eq:helpas1}
w(s)(\cdot)\approx  w(s')(\cdot)\text{ for all }s,s'\in S.
\end{equation}
Intuitively, if \eqref{eq:helpas1} is not fulfilled for all elements in support of $
\mu$ then we can find a menu $A$ which contains non- constant acts and which is preferred to the menu $\bar A$ which contains all the possible lotteries of any act in the menu $A$ as constant acts. On the other hand, if the representation is state-independent, $\bar A$ is clearly always (weakly) better than $A$. Formally, we define as follows. 

\begin{definition}
For a menu $A\subset \F$ define $\bar A$ given by
\[
\bar A = \{g\in \F: g\text{ constant act with } g(s) = f(s')\text{ for some }f\in A, s,s'\in S\}. 
\]
\end{definition}
Note that this can be written in \cite{lu}-notation as $\bar A = \cup_{s\in S}A(s)$.

We now state our version of State Independence for Menus.

\paragraph{Axiom: Weak Dominance} For all menus $A\in \A$ we have $\bar A\better A$.

\begin{lemma}\label{thm:stateindeptmenus}
Let $\better$ over menus from $\F$ have a representation as in Lemma \ref{thm:as1}. Then all elements in the support of $\mu$ are Subjective Expected Utilities (SEUs) if and only if $\better$ satisfies Weak Dominance. 
\end{lemma}

\begin{proof}
In the first step we show Necessity, whereas the rest of the steps shows Sufficiency.

\emph{Step 1.}  Assume that the $\mu$ from the  representation has a $supp(\mu)$ so that each element $w$ in it can be written as $w(s)(\cdot) = q(s)u(\cdot)$ for some $u\in \U$, independent of $s$, and $q(s)>0$. Then it is easy to see that $\bar A\better A$, since whatever $w\in supp(\mu)\subset \U$ is realized the agent has weakly more flexibility when choosing in $\bar A$ than when choosing in $A$. She can always pick the highest lottery from all possible lotteries in $A$ for the Bernoulli utility realized.

\emph{Step 2.} Assume for the other direction that the representation has $supp(\mu)$ not contained in the subspace $\U$. Then at least one of the rows $w(s)(\cdot)\in (\R^X)^S$ from $supp(\mu)$ is non-constant, i.e. contains different Bernoulli utilities (note that there are finitely many rows due to Finiteness). Partition the set of all (non-constant) Bernoulli utilities 
\[
\{w(s)(\cdot)\}_{s\in S,w\in supp(\mu)}
\]
so that two Bernoulli utilities within an element of the partition represent the same EU-preference over $\Delta(X)$ and so that distinct elements of the partition represent distinct EU-preferences. 

Now use Lemma 13 in \cite{fis} (separation Lemma for lotteries) to pick for each equivalence class in the partition above a lottery with the separation property given in Lemma 13 of \cite{fis}. And keep the same lottery for each Bernoulli utility within the same element of the equivalence class. If we denote the elements from $supp(\mu)$ with $w^j,j=1,\dots,K$ ($K=|supp(\mu)|$) we have thus picked lotteries $p^{j}_i\in \Delta(X), i=1,\dots, |S|, j=1,\dots,K$ with the following property
\[
(SP)\quad w^j(s_i)(p_i^j)\geq w^j(s_i)(p_l^k),\text{ and }w^j(s_i)(p_i^j)>w^j(s_i)(p_l^k)\text{ whenever }w^j(s_i)\not\approx w^k(s_l). 
\]

Consider the acts $f^j$ with $f^j(s_i) = p_i^j$ for all $i=1,\dots |S|$ and $j=1,\dots,K$. Finally, take $A=\{f^j:j=1,\dots,K\}$. Obviously then $\bar A = \{p_i^j:i=1,\dots, |S|;j=1,\dots,K\}$ where we have identified constant acts from $\F$ with the respective lotteries from $\Delta(X)$. Note that for each $j$ we have 
\[
w^j\cdot f^j = \sum_{i}w^j(s_i)(p_i^j)\geq \sum_{i}w^j(s_i)(p_i^k)= w^j\cdot f^k, \text{ whenever }k\neq j.
\]
Thus, whenever each $w^j$ is realized the $\max_{f\in A}w^j\cdot f$ is achieved in $f^j$.
Now it also holds 
\[
w^j\cdot f^j=\sum_{i}w^j(s_i)(p_i^j)\geq   \sum_{i}w^j(s_i)(p_l^k)= w^j\cdot (p_l^k).
\]
Moreover the inequality here is strict whenever  there are at least two non-equivalent $w^j(s_i)\not\approx w^j(s_{i'})$. This is because $p_l^k$ can be equal to at most one of the elements from $\{p_i^j,p_{i'}^j\}$. It follows that for at least one $w^j$ which is state-dependent and occurs with positive probability we have $w^j\cdot f^j>  w^j\cdot (p_l^k)$ for all $(l,k)$. Thus we can write that in general 

\[
\max_{f\in A} w^j\cdot f\geq \max_{g\in \bar A}w^j\cdot g
\]
with strict inequality whenever $w^j$ has at least two non-equivalent $w^j(s_i)\not\approx w^j(s_{i'})$.

It follows $V(A)>V(\bar A)$ contradicting Weak Dominance. 

\emph{Step 3.} Assume now Weak Dominance. It follows that for every $w\in supp(\mu)$ fixed it holds $w(s)(\cdot)\approx w(s')(\cdot)$ for all $s,s'\in S$. Fix such a $w$. There exists at least one $s\in S$ for which $w(s)$ is non-constant. Assume w.l.o.g. that it is $s_1$. By the classical vNM Theorem this means that there exists $a:S\ra [0,\infty), b:S\ra\R$ with $a(s_1) = 1, b(s_1) = 0$ so that 
\[
w(s)(\cdot) = a(s)w(s_1)(\cdot)+b(s),\quad s\in S.
\]
Define for this $w$ the mapping 
\[
q:S\ra[0,1],\quad q(s) = \frac{a(s)}{\sum_{s'\in S}a(s')},\quad s\in S. 
\]
This defines a probability distribution over $S$. Defining $u(\cdot) = w(s_1)(\cdot)$ we can write 
\[
w\cdot f = A(w)\left[\sum_{s\in S}q(s)u(f(s))\right] +B(w),
\]
for suitable $A(w)>0, B(w)\in\R$. We define now $B = \sum_{w\in supp(\mu)}B(w)\mu(w)$, $A = \sum_{w\in supp(\mu)}A(w)$ and finally the probability distribution $\nu(q,u) = A(w)\mu(w)\frac{1}{\sum_{w'\in supp(\mu)}A(w')\mu(w')}$ if $w$ corresponds to $(q,u)$ as defined above. We then have that $\better$ can be represented by 
\[
V(F) = A'\int_{\Delta(S)\times\R^X} \max_{f\in F}\sum_{s\in S}q(s) u(f(s))d\nu(q,u) + B,
\]
for a suitable $A'>0$. 
This implies that the preferences can be represented as well by 
\[
V(F) = \int_{\Delta(S)\times\R^X} \max_{f\in F}\sum_{s\in S}q(s) u(f(s))d\nu(q,u).
\]
Finally, we again can change measure and assume w.l.o.g. the $u$-s are in $\U$. Thus, there exists then a finite support probability distribution $\nu_0$ over $\Delta(S)\times\U$ so that $V$ below represents $\better$.
\[
V(F) = \int_{\Delta(S)\times\U} \max_{f\in F}\sum_{s\in S}q(s) u(f(s))d\nu_0(q,u).
\]
\end{proof}

We now work towards uniqueness of the representation in Definition \ref{thm:DLR-SEUdef}. It is clear by focusing on menus of constant acts (which are then isomorphic to menus of lotteries in the natural way) and using Proposition 3 in \cite{as}, that we have a uniqueness type of result as in Proposition 3 of \cite{as} for the marginal distribution of $\mu$ over $\U$.\footnote{I.e. this means the marginal of $
U$ and its support are unique.} The remaining question is whether we get uniqueness of the distribution of beliefs over $\Delta(S)$. Let us first fix the marginal of $\mu$ over $\U$. Since all elements from $supp(\mu)$ represent different SEU-preferences we have that for two distinct elements $(q,u),(q',u')\in supp(\mu)$: either $u\not\approx u'$ or $u\approx u', q\neq q'$ (clearly two SEU representations  $(q,u),(q',u')$ represent the same SEU preference if and only if $q=q'$ and $u\approx u'$). We write $(q,u)\approx (q',u')$ whenever they represent the same SEU preference. It is also easy to see by a separation argument based on Lemma 1 in the appendix of the main paper, that if one has two representations of a menu preference as in Definition \ref{thm:DLR-SEUdef} then the supports of the probability measures over SEU-s are the same (up to linear monotonic transformations of the $u$-s).

To formally prove uniqueness of the measure $\mu$ in the representation of Definition \ref{thm:DLR-SEUdef} whenever the Bernoulli utility functions come from $\mathcal{U}$ we need to adapt the full machinery of \cite{sarver} to the new setting of SEUs, i.e. we need to work with support functions for acts instead of lotteries.

\subsubsection{Uniqueness in the Option Value model with SEU-s as subjective state space.}

Let $X$ be a finite prize space and $S$ a finite state space of objective states. Let again $\F$ be the set of acts $f:S\ra\Delta(X)$. We look at a preference over (for now arbitrary) menus from $\F$. We consider the following axioms on $\better$.

\begin{itemize}
\item A1: Weak Order.
\item A2: Continuity: for any $A$, the strict upper and lower contour sets $\{A\subset \F:B\sbetter A\}$ and $\{A\subset \F:B\sworse A\}$ are open in the Hausdorff topology.
\item A3: Independence: If $A\sbetter B$, then for all $C$ and $\lambda\in (0,1]$ one has \\$\lambda V(A)+(1-\lambda)V(C)\sbetter \lambda V(B)+(1-\lambda)V(C)$.
\item A4: Monotonicity: $A\subseteq B$ implies $A\worse B$. 
\item A5: Weak Dominance: $\bar A\better A$ for all menus $A$. 
\end{itemize}


In the following we abuse notation by using interchangeably $l,p$ and other small-letter variables to denote both lotteries from $\Delta(X)$ and also constant acts. There is no confusion because of the natural isomorphism between lotteries and constant acts.

We define as in \cite{dlrs}:
$\bar\A$ as the set of all closed, convex, non-empty subsets of $\F$. Let $C(\subs)$ be the set of all continuous functions on $\subs$. Define the set of subjective states 
\[
\subs = \Delta(S)\times \U.
\]
This is again a convex, finite-dimensional, compact topological space.\footnote{We equip it with the euclidean norm of its natural euclidean ambient space $\R^{|X||S|}$.}
Define the support function of some $x\in\bar\A$ as
\[
\sigma_x:\subs\ra \R,\quad   \sigma_x(q,u) = \max_{f\in x}q\cdot (u\circ f).
\]

If we equip the space $\F$ with the euclidean norm of $\R^{|S||X|}$ and the space of closed, non-empty subsets of $\F$ with the corresponding Hausdorff norm $d_h$ then it is easy to see the following inequality holds.

\begin{equation}\label{eq:Lipschitzineq}
||\sigma_x-\sigma_y||_{\infty}\leq d_h(x,y),\quad x,y\subset \F. 
\end{equation}

Define the subset of $C = \{\sigma_x:x\in\bar\A\}$. For any $\sigma\in C$ define as in \cite{dlrs} 

\[
x_{\sigma} = \cap_{(q,u)\in\subs}\{f\in\F: q\cdot u(f)\leq \sigma(q,u)\}. 
\]

Lemma S2 in \cite{dlrs} doesn't depend on the structure of $\subs$, as long as it is compact, convex, non-empty. So it holds here true as well.

Lemma S3 from \cite{dlrs} holds as well. We write it down here as part 1) of the following Lemma for completeness and future reference but also add other parts from the papers \cite{sarver}, \cite{dlr} and \cite{dlrs} which are needed later within this subsection. 

\begin{lemma}\label{thm:S3dlrs}
1) The set $C$ is convex and $\sigma_{(\frac{1}{|X|},\dots,\frac{1}{|X|})} \equiv 0$. In particular, $0\in C$. 

2) There exists a menu $y$ containing only constant acts such that $\sigma_y \equiv c>0$.

3) $x\subseteq y$ implies $\sigma_x\leq \sigma_y$. 
\end{lemma}

Part 2) (the `harder' part) can be verified from footnote 6 in pg. 596 of \cite{dlrs}. 
\vspace{2mm}\\
Virtually the same as in \cite{dlrs} we define as follows. 

\[
H = \cup_{r\geq 0}rC = \{r\sigma\in C(\subs):r\geq 0\text{ and }\sigma\in C\},
\]

as well as

\[
H^* = H-H = \{f\in C(\subs):f= f_1-f_2\text{ for some }f_1,f_2\in H\}. 
\]
Note now that $2)$ in Lemma \ref{thm:S3dlrs} implies that $\textbf{1}\in H^*$.

Perusing the proof of Lemma S10 in \cite{dlrs} one sees that it only uses the algebraic facts from Lemma \ref{thm:S3dlrs} and the fact that the denseness result from \cite{hormander} always holds whenever the ambient space $E$ mentioned in \cite{hormander} is finite dimensional. It follows therefore that Lemma S10 in \cite{dlrs} still holds in our setting. 

\begin{lemma}\label{thm:S10dlrs} 1) The set $H^*$ is a linear subspace of $C(\subs)$. 

2) For any $f\in H^*$, there exists $\sigma_1,\sigma_2\in C$ and $r>0$ with $f = r(\sigma_1-\sigma_2)$. 

3) The set $H^*$ is dense in $C(\subs)$ w.r.t. the topology of uniform convergence.

\end{lemma}

Lemma \ref{thm:S3dlrs}, Lemma \ref{thm:S10dlrs} and Th\'eor\`eme 9 in \cite{hormander} are used just as in \cite{dlrs} to prove the following Lemma.\footnote{Actually this Lemma holds true more generally: it suffices that $\subs$ be compact, convex.}

\begin{lemma}\label{thm:S11dlrs}
Any Lipschitz continuous linear functional $W:C\ra\R$ has a unique continuous linear extension to $C(\subs)$. If $W$ is monotone, then this extension is a positive linear functional.
\end{lemma}

This Lemma can be used word for word as in the proof of Lemma 18 in \cite{sarver} to prove the following.

\begin{lemma}\label{thm:sarver18}
If $\nu$ and $\nu'$ are two finite Borel measures on $\subs$ and if $\int_{\subs}\sigma_A(u)d\nu(u) = \int_{\subs}\sigma_A(u)d\nu'(u)$ for all $A\in \A$, then $\nu = \nu'$.  
\end{lemma}

The measures $\nu,\nu'$ are not necessarily probability measures as in Lemma 18 of \cite{sarver}. The proof nevertheless goes through because of the following facts: 

\begin{itemize}
\item The functional $C(\subs)\ni f\ra\int_{\subs}f(u)d\nu(u)$ is continuous in the maximum norm.
\item The functions $\sigma_A$ for all $A$ are dense in $C(\subs)$.  
\item All finite Borel measures on a compact space are regular.\footnote{This is true for finite Borel measures because it is true for Borel probability measures over metric spaces.} 
\end{itemize}
One uses these facts then to approximate for each compact $K\subset \subs$ the indicator function $\textbf{1}_K$ by continuous functions.\footnote{The procedure is called `mollification'.} From this and regularity one shows that the two measures coincide.
\vspace{4mm}\\
Before going on, we note that the proofs of Lemmas S1-S2 in \cite{as} go through word for word for $\F$ instead of the lottery space. In particular, a preference $\better$ in $\A$ can be extended to a preference in $\bar\A$ whenever it satisfies Continuity, Independence and Monotonicity. 

We now introduce the space $\p$ of convex polytopes in $\F$. Note that again as in \cite{as} we have $\p = \{co(A):A\in \A\}$. Just as in \cite{as} $\p$ is a mixture space and one can extend a preference $\better$ on $\A$ to $\p$. The extension is well-defined because $\p\subset\bar\A$. 

Just as in \cite{dlr} and \cite{dlrs} one uses A-1 till A-4 to find a function $V$ representing $\better$ over $\bar\A$ such that 
\begin{enumerate}
\item $V$ is affine.
\item $V$ is monotonic. 
\end{enumerate}

One then uses $V$ to define a function $W:H\ra\R$ through
\begin{equation}\label{eq:defW}
W(r\sigma) = rV(x_{\sigma}). 
\end{equation}
One extends $W$ to $H^*$ as in the case of lotteries by using Lemma \ref{thm:S10dlrs}. $W$ is again monotonic and also Lipschitz continuous by the same argument as in the proof of Theorem 2 in \cite{dlrs}.  
By the definition in \eqref{eq:defW} and the extension of $W$ we have by a simple argument that $V$ is Lipschitz continuous as well on $\bar\A$:

\[
V(x)-V(y) = W(\sigma_x-\sigma_y)\leq ||\sigma_x-\sigma_y||_{\infty}\leq d_h(x,y).
\]
Here in the last inequality we have used \eqref{eq:Lipschitzineq}. This fact lets us see that Lemma S3 and Lemma S4 from \cite{as} go through word for word for our setting of AA-acts. One can use this and Lemma \ref{thm:sarver18} to establish uniqueness as in the proof of Theorem S1 in \cite{as}. We note down the result in the following Lemma.

\begin{lemma}\label{thm:uniquenessDLR-SEU}
Suppose two probability measures $\mu_1,\mu_2$ over $\subs$ satisfy 
\[
A\better B\quad\equivalent \quad \int_{\subs}\max_{f\in A}q\cdot u(f)d\mu(q,u)\geq \int_{\subs}\max_{f\in B}q\cdot u(f)d\mu(q,u),\quad \forall A,B\in\A.
\]
Then $\mu_1=\mu_2$. 
\end{lemma}
Combining overall Lemmas \ref{thm:as1}, \ref{thm:stateindeptmenus} and Lemma \ref{thm:uniquenessDLR-SEU} we have proved the following Theorem. 

\begin{theorem}[DLR-SEU]\label{thm:DLR-SEU}
A preference $\better$ over $\F$ satisfies Axiom 0 (the Classical Menu Preference Axiom) as well as Weak Dominance if and only if there exists a finite-support probability measure over $\Delta(S)\times \U$ such that 
\[
A\better B\quad\equivalent \quad \int_{\Delta(S)\times \U}\max_{f\in A}q\cdot u(f)d\mu(q,u)\geq \int_{\Delta(S)\times \U}\max_{f\in B}q\cdot u(f)d\mu(q,u),\quad \forall A,B\in\A.
\]
Moreover, $\mu$ is unique with this property. 
\end{theorem}

We close this sub-subsection by writing down how a DLR-SEU changes, if instead of normalizing Bernoulli utilities in $\U$ we instead look at general Bernoulli utilities $u\in\R^Z$. This corresponds to Proposition 3 in \cite{as}. First we define for completeness and reference the general DLR-SEU representation.

\begin{definition}\label{thm:DLR-SEUgeneraldef}
A DLR-SEU representation for a preference $\better$ over $\A$ is a triple $(S,SubS,\mu)$ where $S$ is a finite space of objective states, $SubS$ is a finite space of subjective states, i.e. a finite subset of $\Delta(S)\times \R^Z$ and $\mu$ is a measure over $SubS$ so that 
\begin{enumerate}
\item $A\better B$ if and only if $V(A)\geq V(B)$ where $V:\A\ra\R$ is defined by 
\[
V(A) = \sum_{(q,u)\in SubS}\mu(q,u)\cdot\left(\max_{f\in A}(q\cdot u)(f)\right).
\]
\item \emph{Non-Redundancy:} Any two distinct $(q,u),(q',u')\in SubS$ represent different SEU preferences over $\F$.\footnote{This is equivalent to the statement that either $q\neq q'$ or $q=q'$ and $u\not\approx u'$.} 
\item Minimality: $\mu(q,u)>0$ for every $(q,u)\in SubS$. 
\end{enumerate}
\end{definition}

This is the identification result for general DLR-SEU representations. 
\begin{proposition}\label{thm:DLR-SEUgeneraluniqueness}
Two DLR-SEU representations $(S_i,SubS_i,\mu_i),i=1,2$ as in Definition \ref{thm:DLR-SEUdef} represent the same preference $\better$ over $\A$ if and only if there exists a bijection $\gamma:S_1\ra S_2$, a constant $c>0$, a bijection $\Gamma:SubS_1\ra SubS_2$ and functions $a:SubS_1\ra(0,\infty)$, $b:SubS_1\ra\R$ such that 
\begin{enumerate}
\item For all $(q,u)\in SubS_1$ we have 
\begin{itemize}
\item $q(s) = \pi_{q}(\Gamma(q,u))(\gamma(s))$ for all $s\in S_1$ and $\pi_{u}(\Gamma(q,u)) = a(q,u)u+b(q,u)$.
\end{itemize}
\item For all $(q,u)\in SubS_1$ we have
\begin{itemize}
\item $\mu_1(q,u) = \frac{c}{a(q,u)}\mu_2(\Gamma(q,u))$.
\end{itemize}
\end{enumerate}
\end{proposition}

\begin{proof}
This is almost word for word as the proof of Proposition 3 in \cite{as}: one direction is trivial and for the other direction one reduces to the canonical subjective state space $\Delta(S)\times\U$ and uses Lemma \ref{thm:uniquenessDLR-SEU} there besides the completely analogous algebraic manipulations. 
\end{proof}

We gather together the results for the DLR-SEU model.

\begin{theorem}
\label{thm:DLR-SEUtotal}
Assume that $\better$ over $\A$ satisfies the Classical Menu Preference Axiom (Axiom 0) and Weak Dominance. There exists then a finite support probability distribution $\nu$ over $\Delta(S)\times\U$ so that $V$ below represents $\better$.
\[
V(A) = \int_{\Delta(S)\times\U} \max_{f\in A}\sum_{s\in S}q(s)\cdot u(f(s))d\nu(q,u).
\]
$\nu$ is essentially unique as described in Proposition \ref{thm:DLR-SEUgeneraluniqueness} and it is unique whenever it is required that the support is contained in $\Delta(S)\times\U$. 
\end{theorem}

\subsubsection{Relation between Strong Dominance and Weak Dominance (SEU version)}


Let's recall the axioms from \cite{dlst}. These are for a preference $\better$ over menus of acts $\F$ as we defined in this section of the appendix (we are still only looking at a finite prize space).

\begin{enumerate}
\item DLST-A1: Weak Order.
\item DLST-A2: vNM Continuity: if $A\sbetter B\sbetter C$ then there are $\alpha, \beta\in (0,1)$, such that $\alpha A+(1-\alpha)C\sbetter B\sbetter \beta A+(1-\beta)C$.
\item DLST-A3: Non-triviality: there exists lotteries $p,q\in\Delta(X)$ with $\{p\}\sbetter \{q\}$.
\item DLST-A4: Monotonicity: $A\subseteq B$ implies $A\better B$.
\item DLST-A5: Strong Dominance: if $f\in A$ and $\{f(s)\}\better \{g(s)\}$ then $A\sim A\cup\{g\}$  
\item DLST-A6: Finiteness: there exits a $K\in\N$ such that for all $A\in \A$ there exists $B\subset A$ with $|B|\leq K$ and so that $B\sim A$. 
\end{enumerate}
\cite{dlst} prove the following.

\begin{proposition}\label{thm:DLST-SEU}
If $\better$ satisfies $DLST-A1 - DLST-A5$ then there exists a probability measure $\nu$ over $\Delta(S)$ and a Bernoulli utility function $u:\Delta(X)\ra\R$ with 
\[
V(A) = \int_{\Delta(S)}\max_{f\in A}(q\cdot u)(f)\nu(dq).
\]
$\nu$ is unique. Finally, $\nu$ has finite support if and only if $\better$ satisfies Finiteness (DLST-A6). 
\end{proposition}
\begin{proof}
See Theorem 1 and Appendix A in \cite{dlst}. The \emph{Finiteness} part follows from Theorem \ref{thm:DLR-SEU}, once we show in Proposition \ref{thm:relationDomStateIndependence} below that Strong Dominance implies Weak Dominance under the Classical Axioms of Menu Preference.  
\end{proof}
The main behavioral axiom in this Theorem is DLST-A5: Strong Dominance. It says that adding an act which delivers dominated outcomes (according to $u$) at each realization of the objective states doesn't change the value of a menu. If the $u$ is allowed to be stochastic, then one may still compare lotteries but now only through an average Bernoulli utility $u$. Thus it is impossible to rank lotteries unambiguously ex-ante. Therefore the DLST axiom does not work anymore with our more general model: once the agent knows precisely the $(q,u)$ she is facing she might decide to pick an act, which even though dominated state by state \emph{on average}, is better given the conditional information that the subjective state is $(q,u)$. 

In the following we clarify the relation between our Weak Dominance and the Strong Dominance Axiom from \cite{dlst}. 

\begin{proposition}\label{thm:relationDomStateIndependence}
1) Strong Dominance and Monotonicity imply Weak Dominance. 

2) There are models where subjective states are SEU-s which don't satisfy Strong Dominance even though they satisfy Monotonicity and Weak Dominance.
\end{proposition}

\begin{proof}
1) Note first, that $\better$ implies a ranking $\better_{\Delta}$ over $\Delta(X)$ by defining 
\[
p\better_{\Delta} q\quad\equivalent\quad \{p\}\better \{q\}.
\]

Consider a menu $A$ and an act $f\in A$. Pick $f(s')$ for some $s'\in S$ so that it is the highest ranked lottery from $supp(f)$ according to $\better_{\Delta}$. Let $h_f$ the constant act which gives $f(s')$ in each state $s\in S$. Note that $h_f\in \bar A$ by definition of $\bar A$. Obviously it holds $\{h_f(s)\}\better \{f(s)\}$ for all $s\in S$. Strong Dominance now gives $\{f\}\cup \bar A\sim \bar A$. We can repeat this process with all $f\in A$ to get $A\cup\bar A\sim \bar A$. But now Monotonicity implies $A\worse A\cup\bar A\sim \bar A$. Since $A$ was arbitrary, this implies that Weak Dominance holds. 

2) See the Example below following the Lemma.
\end{proof}

\paragraph{Example: Option Value model with SEU-s which doesn't satisfy Strong Dominance.}
Take as objective states $S=\{s_1,s_2\}$. 
Consider a model with $\mu$ over subjective states $(q,u)$ such that $\mu = \mu_1\delta_{(q,u_1)}+ \mu_2\delta_{(q,u_2)}$. I.e. the belief about the objective state is fixed, while the Bernoulli utility is stochastic and can be $u_1$ or $u_2$. We identify the belief $q$ by the probability $q$ of the state $s_2$. We write an act $f:S\ra\Delta(X)$ as $f= (f_1,f_2)$ where $f_i  = f(s_i)$. The condition that $\{f(s)\}\better \{g(s)\}$ for all $s\in S$ is then equivalent to 

\begin{equation}\label{eq:help1dlst}
\mu_1 u_1(f_i) +\mu_2 u_2(f_i)\geq  \mu_1 u_1(g_i) +\mu_2 u_2(g_i), \text{ for all }i=1,2. 
\end{equation}

We require that for $u_2$ the act $g$ is optimal, while for $u_1$ the act $f$ is optimal. Thus, we require the following inequalities.

\begin{equation}\label{eq:help2dlst}
(1-q)u_1(f_1) + qu_1(f_2)>(1-q)u_1(g_1) + qu_1(g_2),\quad (1-q)u_2(g_1) + qu_2(g_2)>(1-q)u_2(f_1) + qu_2(f_2).
\end{equation}

The condition that $\{f,g\}\sbetter \{f\}$ (which would imply that Strong Dominance is violated) is then 
\begin{align*}
&\mu_1\left((1-q)u_1(f_1) + qu_1(f_2)\right) + \mu_2\left((1-q)u_2(g_1) + qu_2(g_2)\right)\\
&> \mu_1\left((1-q)u_1(f_1) + qu_1(f_2)\right) + \mu_2\left((1-q)u_2(f_1) + qu_2(f_2)\right).
\end{align*}
But this is implied by the second part of \eqref{eq:help2dlst}. Thus we need to find $q,\mu_1,\mu_2$ and values of $u_i(f_j), u_i(g_j),i,j=1,2$ so that both of 
\eqref{eq:help1dlst} and \eqref{eq:help2dlst} hold true. Such values are for example 

\[
q=\mu_2=\frac{1}{3}, u_1(f_1) = u_2(g_1) = u_2(g_2) = u_1(f_2) =\frac{1}{2}, u_1(g_1) = u_1(g_2) = u_2(f_1) = u_2(f_2) = 0.
\]
Note that it is without loss of generality to pick these Bernoulli values and still have Bernoulli utilities $u_1,u_2\in \U$. This is because one can always enlarge the prize space $X$ as needed to still satisfy the normalizations required in $\U$.

\section{Sophistication}\label{sec:soph}

\subsection{Extending the Ahn-Sarver results to DLR-SEU and Lu's model (DLR-R-SEU).}

We first state the representation we are after in this part. We are in the two-period setting as in \cite{as}, i.e. the observable is preference over menus $A\in \A$ and also an aSCF $\rho$ encoding ex-post choice from a menu. We work with a \emph{finite} prize space $X$ in this subsection. 

\begin{definition}[DLR-R-SEU]\label{thm:dlrrseudef}
Let $(\better, \lambda)$ be a pair consisting of a preference over menus and an aSCF. A DLR-R-SEU representation of $(\better, \lambda)$ is a pair $(\mu, \tau)$ consisting of a finite-support probability measure $\mu$ over $\Delta(S)\times \R^X$ and a tie-breaker rule $\tau$ indexed by the elements of the support of $\mu$ such that $\mu$ gives a DLR-SEU representation for $\better$ and $(\mu,\tau)$ a R-SEU representation for $\rho$. 
\end{definition}

The axioms which make the connection between the two representations are now the following ones. 

\paragraph{Axiom AS-1} If $A\cup \{f\}\sbetter A$, then $\rho(f,A\cup\{f\},s)>0$ for some $s\in S$. 
\vspace{2mm}
\paragraph{Axiom AS-2} For any $A\in \A$ and $f\in A$, if there exists $\epsilon>0$ such that $\rho(g,B,s)>0$ for some $s\in S$ whenever $d(f,g)<\epsilon$ and $d_H(A,B)<\epsilon$, then $A\cup\{f\}\sbetter A$.\footnote{Note, that due to the definition of an aSCF we can write in AS-2 the part `$\rho(g,B,s)>0$ for some $s\in S$' as $\rho(g,B)>0$  and we can write the part `$\rho(f,A\cup\{f\},s)>0$ for some $s\in S$ ' as $\rho(f,A\cup\{f\})>0$ in AS-1.} 

The version of the Theorem from \cite{as} is then the following. 

\begin{theorem}\label{thm:dlrrseuthm}
Suppose that $\better$ has a DLR-SEU representation and $\rho$ has a R-SEU representation. Then the pair $(\better, \rho)$ satisfies axioms AS-1 and AS-2 if and only if it has a DLR-R-SEU representation.
\end{theorem}

Just as in \cite{as} we use Lemma 1 in the appendix of the main paper to first prove the following Lemma.

\begin{lemma}\label{thm:helpas1}
Suppose $\better$ has a DLR-SEU representation with $\mu$ and $\rho$ has a R-SEU representation with $(\mu',\tau)$. 

\begin{enumerate}
\item The pair $(\better,\rho)$ satisfies Axiom AS-1 if and only if for the supports of $\mu$ and $\mu'$ it holds $\mu\subseteq \mu'$, where we have identified the elements of the support as equivalence   classes of SEU-s up to positive affine transformations of the Bernoulli utilities. 
\item The pair $(\better,\rho)$ satisfies Axiom AS-2 if and only if for the supports of $\mu$ and $\mu'$ it holds $\mu'\subseteq \mu$, where we have identified the elements of the support as equivalence classes of SEU-s up to positive affine transformations of the Bernoulli utilities.
\end{enumerate}
\end{lemma}

\begin{proof}
For the proof we use the version of AS-1 and AS-2 with the SCF $\bar\rho$ derived from the aSCF $\rho$. 

\textbf{1-Necessity.}
Take the set of SEU representations $\{(q,u)\in\Delta(S)\times \R^X: (q,u)\in supp(\mu)\cup supp(\mu')\}$. Take a separating menu $A$ as in Lemma 1 in main body of paper for this set of SEUs. Fix a $(q,u)\in supp(\mu)$ and take its corresponding $f(q,u)\in A$. It holds thus that $q\cdot u(f(q,u))>q\cdot u(g)$ for all $g\in A\setminus \{f(q,u)\}$. The DLR-SEU representation then implies $A\sbetter A\setminus \{f(q,u)\}$, which by AS-1 implies $\rho(f(q,u),A)>0$. It follows from the S-REU representation that there should exist some $(q',u')\in supp(\mu')$ such that $f(q,u) = \argmax_{f\in A}q'\cdot u'(f)$. By the property of $A$ we must have $(q,u)\approx(q',u')$.

\textbf{1-Sufficiency.} 
Suppose that $supp(\mu)\subset supp(\mu')$. Fix an $A\in \A$ and $f\in \F$ such that $A\cup \{f\}\sbetter A$. This needs (!) $q\cdot u(f)>\max_{g\in A}q\cdot u(g)$ for some $(q,u)\in supp(\mu)$. There exists $(q',u')\in supp(\mu')$ with $(q',u')\approx (q,u)$. It follows from the S-REU representation and (!) that $\bar\rho(f,A\cup\{f\})\geq \mu(q',u')>0$. 

\textbf{2-Necessity.} We use again the menu $A$ from 1-Necessity. Fix a $(q,u)\in supp(\mu')$. There exists $f\in A$ such that $q\cdot u(f)> q\cdot u(g)$ for all $g\in A\setminus\{f\}$. Since $A$ is finite and $q\cdot u$ continuous we have the existence of some $\epsilon>0$ such that $q\cdot u(h)> q\cdot u(g')$ for all $h,g'\in \F$ such that $d(h,f)<\epsilon$ and $d(g,g')<\epsilon$ for some $g\in A\setminus\{f\}$. 

Fix any $h$ and $B$ such that (!) $d(f,h)<\epsilon$ and $d_H(A\setminus\{f\},B)<\epsilon$. Then by definition of the Hausdorff metric for any $g'\in B$ there exists $g\in A\setminus\{f\}$ with $d(g,g')<\epsilon$. Hence, (!!) $q\cdot u(h)>q\cdot u(g')$ for all $g'\in B$. The R-SEU representation now implies $\rho(h,B\cup\{h\})\geq \mu'(q,u)>0$. Since (!!) holds for all $h,B$ satisfying (!), Axiom AS-2 now implies $A\sbetter A\cup\{f\}$. R-SEU representation now requires the existence of some $(q',u')\in supp(\mu)$ with $f\in argmax_{g\in A}q\cdot u(g)$. This implies that $(q,u)\approx(q',u')$ due to the separating property of $A$.

\textbf{2-Sufficiency.} Suppose that $supp(\mu')\subset supp(\mu)$. Fix now a menu $A\in \A$ and $f\in A$ and $\epsilon>0$ such that $\rho(g,B\cup\{g\})>0$ whenever $d(f,g)<\epsilon$ and $d_H(A,B)<\epsilon$. To show that $A\cup \{f\}\sbetter A$ it suffices to show the existence of some $(q,u)\in supp(\mu)$ with (!) $q\cdot u(f)>q\cdot u(g)$ for all $g\in A$. For this, given the statement, it suffices to find $(q,u)\in supp(\mu')$ which satisfies (!). 

Assume on the contrary that $q\cdot u(f)\leq \max_{g\in A}q\cdot u(g)$ for all $(q,u)\in supp(\mu')$. For each $(q,u)\in supp(\mu')$ let $f(q,u)\in A$ be so that $q\cdot u(f(q,u)) = \max_{g\in A}q\cdot u(g)$ and let also $g(q,u)\in \Delta(X)$ be the constant act so that $u\circ g(q,u) = \max_{p\in \Delta(X)}u(p)$. Consider a menu of the type 
\[
B = A\cup \{\alpha g(q,u)+(1-\alpha)f(q,u):(q,u)\in supp(\mu')\}, 
\]
for some $\alpha\in (0,1)$. For all $\alpha$ small we have $d_H(A,B)<\epsilon$. Also, by letting $\bar p$ the uniform distribution over $X$ (and identifying it with its respective constant act), and considering $ g =\alpha \bar p+(1-\alpha)f $ we have $d(g,f)<\epsilon$ for all $\alpha\in (0,1)$ small. Since each $u$ for $(q,u)\in supp(\mu')$ is non-constant (by definition of R-SEU) we have $u(g(q,u))>u(\bar p)$ for all $u$ s.t. $(q,u)\in supp(\mu')$. It follows for all $(q,u)\in supp(\mu')$

\[
q\cdot u(g) = \alpha u(\bar p) +(1-\alpha)q\cdot u(f)< \alpha u(g(q,u)) + (1-\alpha) q\cdot u(f(q,u)) = max_{h\in B}q\cdot u(h). 
\]
This implies $\bar\rho(g,B) = 0$, which is a contradiction. It follows that (!) is true and we are done. 
\end{proof}

We now finish the proof of the Theorem.

\begin{proof}[Proof of Theorem \ref{thm:dlrrseuthm}]
The direction from the representation to the axioms AS-1 and AS-2 is clear due to Lemma \ref{thm:helpas1}. 

Conversely, assume that $(\better,\bar\rho)$ have respectively DLR-SEU and R-SEU representations through $\mu$ and $\mu'$ and that the pair satisfies AS-1 and AS-2. Then Lemma \ref{thm:helpas1} shows that the supports of $\mu$ and $\mu'$ are the same. Define for each $(q,u)\in supp(\mu)$ $(q,u')$ with $u' = \frac{\mu(q,u)}{\mu'(q,u)}u$ (here we may have repetitions if the $u$-s also get repeated several times in $supp(\mu)$). These Bernoulli utilities are well-defined and non-constant. Define now $\bar \mu = \mu'$ and $\bar \tau = \tau$. One shows very easily as in the proof of Theorem 1 of \cite{as} that $(\bar \mu,\bar\tau)$ is a DLR-R-SEU representation of $\better$. 
\end{proof}

We note uniqueness of the DLR-R-SEU representation. 

\begin{proposition}\label{thm:dlrrseuunique}
Two DLR-R-SEU representations $(\mu,\tau)$ and $(\mu',\tau')$ represent the same pair $(\better,\rho)$ if and only if there exists a scalar $\alpha>0$ and a function $\beta:supp(\mu)\ra supp(\mu')$ such that the following holds true.
\begin{enumerate}
\item For any $(q,u)\in supp(\mu)$ we have $u = \alpha u'+ \beta(q,u)$ for a unique $u'$ such that $(q,u')\in supp(\mu')$. 
\item For any $(q,u)\in supp(\mu)$ $\mu(q,u) = \mu(q,u')$ for a unique $(q,u')\in supp(\mu')$. 
\item For any $(q,u)\in supp(\mu)$ $\tau_{q,u} = \tau'_{q,u'}$ for a unique $(q,u')\in supp(\mu')$.
\end{enumerate}
\end{proposition}

\begin{proof}
B. follows from Proposition 7 from the appendix of the main paper. Using this fact together with B. of Proposition \ref{thm:DLR-SEUgeneraluniqueness} implies $a(q,u) = \alpha>0$ constant, i.e. 1. C. follows again from Proposition 7 from the appendix of the main paper.
\end{proof}

\section{SCFs as observable}\label{sec:redhist}

In the main body of the paper we have assumed that the analyst can observe the realization of the objective state in each instance, or he observes a signal about the realized objective state which \emph{in the aggregate} fully identifies the data-generating process of objective states. 

The case of SCFs in the static setting has been studied extensively in \cite{lu}. In the supplement of \cite{lu} a SCF is shown to be a RUM of SEUs if and only it satisfies Monotonicity, Linearity, Extremeness, Continuity and State Independence (see Theorem S.1 in the supplement of \cite{lu}). Theorem \ref{thm:ASSREU} extends this result to a general prize space and models tie-breaking explicitly in the spirit of \cite{gp} and \cite{fis}.


In the dynamic setting we look at \emph{reduced histories}. A reduced history $rh^t$ is a tuple $(A_0,f_0;\dots,A_t,f_t)$. If the observable is actually an aSCF $\rho$ one can look at its derived SCF $\bar\rho$. This defines a set of reduced histories $\mathcal{RH}^t$ with a typical element $rh^t \in \mathcal{RH}^t$ such that there exists some $h^t\in \h_t$ with $\pi_{Af}(h^t) = rh^t$. 

The derived SCF in each period and after a reduced history $rh_{t-1}$ is given through 
\[
\bar\rho_t(f_t,A_t|rh_{t-1}) = \frac{\sum_{s_t,h_{t-1}:\pi_{fA}(h_{t-1})=rh_{t-1}}\rho_t(h_{t-1})\cdot\rho_t(s_t,f_t,A_t|h_{t-1})}{\sum_{h_{t-1}:\pi_{fA}(h_{t-1})=rh_{t-1}}\rho_t(h_{t-1})}.
\]

Representations, Axioms and Characterizations are very similar to the case of full histories only that now they come in aggegrate form instead of statewise form as stated in the main body of the paper.

For completeness we give here the Representation in the AS-version for the case of SCFs as observables. 

\begin{definition}\label{thm:reddrseudef}
We say that a family of history-dependent SCF $\zeta = (\zeta_0,\dots,\zeta_T)$ has a reduced DR-SEU representation (rDR-SEU) if there exists
\begin{itemize}
\item a finite space $S$ of objective states and a collection of partitions $S_t,t=1\dots,T$ of $S$ such that $S_t$ is a refinement of $S_{t-1}$,
\item a collection of finite subjective state spaces $SEU_t, t=0,\dots,T$ (an element is of the type $\theta_t = (q_t,u_t)\in\Delta(S_t)\times\R^{X_t}$),

\item a collection of probability kernels $$\psi_k:SEU_{k-1}\ra \Delta(SEU_k)$$ for $k=0,\dots,T$ \footnote{With the obvious conventions for $k=0$.} with a typical element $\psi_k^{q_{k-1},u_{k-1}}$. In particular, the probability that $(q_k,u_k)$ is realized after $\theta_{k-1}$ is $\psi_k^{\theta_{k-1}}(q_k,u_k)$ (here we have $\theta_t=(q_t,u_t)$).
\end{itemize}
 such that the following two conditions hold. 

\textbf{rDR-SEU 1}
\begin{enumerate}[(a)]
\item every $(q_t,u_t)\in supp(\psi_t^{\theta_{t-1}})$ represents a distinct, non-constant SEU preference.
\item $supp\psi_t^{\theta_{t-1}} \cap supp\psi_t^{\theta'_{t-1}}=\emptyset$ whenever $\theta_{t-1}\neq \theta'_{t-1}$, both in $\Theta_{t-1}$.\footnote{Recall $\theta_{t-1} = (q_{t-1},u_{t-1},s_{t-1})$. There might be repetitions in terms of the SEUs. When that happens we index the SEUs by their respective $\theta_t$. Thus we keep everything partitional.}
\item $\cup_{\theta_{t-1}}supp (\psi_t^{\theta_{t-1}}) = SEU_t$. 
\end{enumerate}

\textbf{rDR-SEU 2} The SCF $\zeta_t$ after a history $h^{t-1} = (A_0,f_0;\dots,A_{t-1},f_{t-1})$ 
\\
\resizebox{1.1 \textwidth}{!} {
$
\zeta_t(f_t,A_t|h^{t-1}) = \frac{\sum_{(\theta_0,\dots,\theta_t)\in \times_{l\leq t}SEU_{l}}\left[\prod_{k=0}^{t-1}\psi_k^{\theta_{k-1}}(q_k,u_k)\tau_{q_k,u_k}(f_k,A_k)\right]\cdot \psi_t^{\theta_{t-1}}(q_t,u_t)\tau_{q_t,u_t}(f_t,A_t)}{\sum_{(\theta_0,\dots,\theta_{t-1})\in \times_{l\leq t-1}SEU_{l}}\prod_{k=0}^{t-1}\psi_k^{\theta_{k-1}}(q_k,u_k)\tau_{q_k,u_k}(f_k,A_k)}.
$
}
\end{definition}

\begin{remark}
If we add C-determinism* on the SCF in every period and after each history, the resulting representation would have a non-stochastic Bernoulli utility each period. This is then the dynamic version of the main model of \cite{lu}. 
\end{remark}

The proofs in the case of reduced histories and/or SCFs are simpler as now one can identify the states $s_t$ from the proofs of \cite{fis} with the realized Subjective Expected Utility of the agent $(q_t,u_t)$. This allows to transport their proof arguments immediately, once one adds State Independence and C-Determinism to their `static' Axiom 0.\footnote{Details available upon request.}

Axioms comparing the beliefs of the agent with the true data-generating process like CIB or NUC are now impossible due to unobservability of the objective states. 

If the data available don't include the objective states, the analyst can test whether a dataset where objective states would be observable satisfies the DR-SEU model by testing the derived SCF (Null-hypothesis being that DR-SEU is correct). Thus, DR-SEU can be rejected as a theory also when objective states are not observable by testing the axioms in the case of reduced histories. 





\section{Auxiliary Results for Section 4 and miscellanea}

\subsection{Preliminaries}

In the following whenever we consider aSCFs coming from different agents we assume here that they share the same taste.\footnote{To shorten the exposition we skip writing out the conditions on the SCFs which imply that the taste of distinct agents we consider are the same. These are available upon request.} 

The following concept is crucial.

\begin{definition}[\cite{lu}]\label{thm:testact}
Given a derived SCF $\bar \rho$, the test function of a menu $A\in\A$ is the function $A_{\bar\rho}:[0,1]\ra[0,1]$ defined by
\[
A_{\bar\rho}(a) = \bar\rho\left(A,A\cup\{a\underline{f}+(1-a)\bar f\}\right). 
\]
\end{definition}

We assume that any $\bar\rho$ in this section is derived from an aSCF which satisfies the conditions of Theorem 0. It is then clear that $A_{\bar\rho}(\cdot)$ is a continuous cumulative distribution function. 

One of the major conceptual contributions in \cite{lu} is \emph{the associated menu preference with an SCF}.

\cite{lu} establishes that there is  following relation between stochastic choice from menus and the valuation of menus.  

\begin{theorem}[\cite{lu}]\label{thm:thm2lu}
The following are equivalent:
\begin{enumerate}
\item The SCF $\bar\rho$ has a representation as in Proposition \ref{thm:LuThm} with distribution of beliefs $\nu\in \Delta(\Delta(S))$ and non-stochastic taste $u$.
\item $\bar\rho$ has an informational representation and $\better_{\bar\rho}$ is represented by 
\begin{equation}\label{eq:dlst1}
V(A) = \int_{\Delta(S)} max_{f\in A}[q\cdot(u\circ f)]\nu(dq). 
\end{equation}
\end{enumerate}
\end{theorem}
The representation \eqref{eq:dlst1} is called a \emph{subjective learning representation}. It is studied and axiomatized in \cite{dlst}.\footnote{See Online Appendix \ref{sec:ovseu} for its relation to the axiomatization of Evolving SEU in this paper.} 

\subsection{An additional Comparative Statics Result: Informativeness}



Assume the analyst has two aSCFs $\rho_i, i=1,2$ which satisfy the conditions of Proposition 1 over the same measure space; that is, we have a pair of probability spaces $(\Omega,\mathcal{F}^*,\{\mu_i\}_{i=1,2})$ with finite $\Omega$. 

We assume the following condition for the random variables $s_i$ from Definition 2, which give the realizations of the objective states. 

\paragraph{Common DGP Assumption:} The distribution of $s_1$ under $\mu_1$ is the same as the distribution of $s_2$ under $\mu_2$. In terms of stochastic choice: for all $s\in S$ it holds $\rho_1(s) = \rho_2(s)$. 
\vspace{3mm}\\
The condition ensures that both agents are facing an identical DGP even though they are using potentially different information 
structures to learn about the realization of the objective state $s$.

\cite{lu} shows how one can get a ranking of signal structures based on Blackwell informativeness. To recall, for two posterior distributions $\mu,\nu\in\Delta(\Delta(S))$ say that \emph{$\mu$ is Blackwell more informative than $\nu$} if there is a mean-preserving transition kernel $K:\Delta(S)\times \Delta(S)\ra [0,1]$ such that, for all $q\in\Delta(S)$ it holds\footnote{Mean-preserving means that $\int_{\Delta(S)}pK(q,dp) =q$.} 

\[
\mu(q) = \int_{\Delta(S)}K(p,q)\nu(dp). 
\]

Say for two distributions $F,G$ over $[0,1]$ that $F$ \emph{second-order stochastic dominates} $G$, that is, $F\better_{SOSD}G$ if $\int_{\R}\phi dF\geq \int_{\R}\phi dG$ for all increasing, concave $\phi:\R\ra\R$. 

\begin{theorem}[\cite{lu}]\label{thm:informativeness}
Let $\rho_i, i=1,2$ fulfill the Common DGP Assumption and let the distribution of $q_i(\cdot)$ over $\Delta(S)$ be given by $\nu_i, i=1,2$. Then $\nu_1$ is Blackwell more informative than $\nu_2$ if and only if for all $A\in \A$ we have $A_{\bar\rho_1}\better_{SOSD}A_{\bar\rho_2}$. 
\end{theorem}

This is the main comparative statics result in \cite{lu} and it immediately finds application in our setting of richer data, whenever the interpretation of the data is that it comes from agents using different information structures to learn about the realization of \emph{same} objective state $s$.\footnote{Say, two pharmaceutical firms experimenting on the viability of the same set of drugs.}

The informativeness ranking also makes sense in a setting where information structures have misspecified priors. This is because two information structures are ordered w.r.t. informativeness only if they use the same pre-signal prior about the state of the world. This common prior may or may not be correctly specified.\footnote{In the setting of Theorem \ref{thm:informativeness} the agents have a common prior if $\int_{\Delta(S)}q\nu_1(dq) = \int_{\Delta(S)}q\nu_2(dq)$, if this common prior doesn't coincide with the probability distribution coming from the Common DGP Assumption.} Given this common bias, all else equal, an agent is still \emph{normatively speaking} better off possessing a more Blackwell informative information structure.

\subsection{Proofs of Theorem \ref{thm:thm2lu} and Theorem \ref{thm:informativeness}}

The proofs of these two results are just small modifications of the proofs in the original setting of \cite{lu}. One just has to adapt and modify his respective proofs in Appendix B by including and excluding (as appropriate) the arguments Lu uses to eschew the explicit modeling of tie-breaking. 

An interested reader can see easily that the proofs go through by using the following easy-to-verify facts.\footnote{All the following statements remain intact if instead of the finite prize space in \cite{lu} one uses a general prize space which is separable, metric. This is because essentially the only result needed to start the work is the Informational Representation Theorem. The latter is a special case of our Theorem \ref{thm:ASSREU}.} 

\begin{enumerate}

\item Given the existence of uniformly best and worst acts, $\bar f$ and $\underline{f}$, and the continuity of $u$, one could go over to the space of utility acts for both agents to see that the set of menus without ties for both agents is dense in the space of all menus $\A$ equipped with Hausdorff topology. 
\item The map $\A\ni A\ra A_{\bar\rho}$ is continuous in the topology of weak convergence of cdf-s whenever $\bar\rho$ is continuous.
\item Lemma B.1 in Appendix B of \cite{lu} goes through word for word. 
\item Lemma B.2 in Appendix B of \cite{lu} just needs to be modified to the following statement: If $\bar\rho$ has an informational representation, then the test function to the menu $A\cup f^{b}$ is equal to $\max\{A_{\bar\rho},f^{b}_{\bar\rho}\}$ \textbf{a.e.} in $b\in [0,1]$. 
\item The statement of Lemma B.3 in Appendix B of \cite{lu} needs to be modified to: let $A_{\bar\rho_1}=A_{\bar\rho_2}$ for all $A\in\A$ without ties for both SCFs $\bar\rho_i, i=1,2$.\footnote{The set of menus without ties for \emph{both} representations is dense as one easily checks. This follows in our setting due to Finiteness. 
} Then $\mu_1 = \mu_2$ in the representations, i.e. the two stochastic choice functions are the same on all menus without ties. 
\item The statements of Lemma A.6 and B.4 remain intact in our setting. This is because of the Maximum Theorem applied to the optimization problem 

\[
\max_{f\in F}q\cdot(u\circ f).
\]
As one can easily show through applications of the Maximum Theorem the value function of this problem is continuous in $q\in \Delta(S)$ and also $F\in \A$ with the respective topologies of weak convergence of probability measures and Hausdorff distance. The latter continuity property remains intact after taking integrals, i.e. $$V(F) = \int_{\Delta(S)}\max_{f\in F}q\cdot(u\circ f)\mu(dq)$$ is continuous in $F\in \A$. It is also continuous in $\mu$ (recall that we are assuming here that there are best and worst prizes). 
\item Due to the above, the proof of Theorem \ref{thm:thm2lu} is as follows: part (1)$\imply$(2) is precisely the same as in \cite{lu} (with the respective adaptation in the Theorems one needs to cite). For the direction (2)$\imply$(1): one does the same steps as in the proof of \cite{lu}, but focuses on menus without ties for both representations and uses in the very end of the proof the adapted statement of Lemma B.3 from E. here instead of the original statement. 
\end{enumerate}

\subsection{Properties of the map connecting menu preferences to biased beliefs}

\begin{lemma}\label{thm:proppsi}
The map $\psi_q$ is continuous, convex and injective.
\end{lemma}
\begin{proof}
Let a direction of bias $q\in \Delta(S)^{\Omega}$ be given. Recall the definition of $V_a$ for some $a\in [0,1]^{\Omega}$. 
\begin{equation}\label{eq:helppsi}
V_a(A) = \int_{\Omega} \max_{f\in A}\left[a(\omega)q(\omega)+(1-a(\omega))\mu(\cdot|\omega)\right]\cdot(u\circ f)\mu(d\omega). 
\end{equation}
We equip the space of weights $[0,1]^{\Omega}$ with the uniform topology (which due to finiteness of $\Omega$ is equivalent to the topology of point-wise convergence). The integrand is continuous in $a$ due to Maximum Theorem. Uniform boundedness of the integrand implies that $V_a$ is continuous in $a$. Injectivity follows from the Uniqueness result in Theorem 1 of \cite{dlst}. $V_a$ namely satisfies all properties of that Theorem. Finally, convexity follows from the fact that the max operator $\max_{f\in A}$ is convex in its argument and the fact that convexity remains intact after integration.   
\end{proof}